\documentclass[11pt,letter]{amsart}
\usepackage{mathrsfs, amssymb,amsmath,url,amsthm,color,comment}
\usepackage{hyperref}
\usepackage[dvipsnames]{xcolor}
\usepackage[letterpaper,left=1in,top=1in,right=1in,bottom=1in]{geometry}

\author%
{Antoine Mottet}
	\address{Department of Algebra, MFF UK, Sokolovsk\'a 83, 186 00 Praha 8, Czech Republic}
	\email{mottet@karlin.mff.cuni.cz}
	\urladdr{http://www.karlin.mff.cuni.cz/~mottet/}

\author%
{Michael Pinsker}
	\address{Institut f\"{u}r Diskrete Mathematik und Geometrie, FG Algebra, TU Wien, Austria, and Department of Algebra, Charles University, Czech Republic}    
	\email{marula@gmx.at}
    \urladdr{http://dmg.tuwien.ac.at/pinsker/}

\thanks{
Antoine Mottet has received funding from the European Research Council
(ERC) under the European Unions Horizon 2020 research and
innovation programme (grant agreement No 771005). Michael Pinsker has received funding from the  Austrian Science Fund (FWF) through  project No P32337, and from the Czech Science Foundation (grant No 13-01832S)}

\title[Smooth approximations and CSPs]{Smooth approximations and CSPs\\ over finitely bounded homogeneous structures}

\date{\today}

\usepackage{hyperref,multirow}

\usepackage{graphicx}
\theoremstyle{plain}

    \newtheorem{theorem}{Theorem}
    \newtheorem{lemma}[theorem]{Lemma}
    \newtheorem{proposition}[theorem]{Proposition}
    \newtheorem{conjecture}[theorem]{Conjecture}
    \newtheorem{corollary}[theorem]{Corollary}

    \newtheorem{definition}[theorem]{Definition}
    
    \theoremstyle{remark}
    \newtheorem{example}[theorem]{Example}
     
     \newtheorem*{remark}{Remark}

 \newcommand{\ignore}[1]{}

\DeclareMathOperator\argmin{argmin}
\newcommand{\actson}{\curvearrowright}

\DeclareMathOperator{\inj}{inj}

\DeclareMathOperator{\csp}{CSP}

\DeclareMathOperator{\CSP}{CSP}
\DeclareMathOperator{\Pol}{Pol}
\DeclareMathOperator{\Aut}{Aut}
\DeclareMathOperator{\End}{End}

\newcommand\rel[1]{\mathbb{#1}}

\renewcommand{\P}{\mathscr P}

\newcommand{\sH}{\mathbb H}
\newcommand{\sA}{\mathbb A}
\newcommand{\sB}{\mathbb B}
\newcommand{\sC}{\mathbb C}

\newcommand{\sT}{\mathbb T}
\newcommand\sX{\mathbb X}
\newcommand\sY{\mathbb Y}

\newcommand{\Cf}{\cC^{fin}_\sA}
\newcommand{\CAQ}{\mathscr{C}_{\mathbb A}^{\mathbb Q}}

\newcommand{\Af}{\mathbb A^{fin}}

\newcommand{\CA}{\Pol(\sA)}
\newcommand{\CAH}{\mathscr{C}_{\mathbb A}^{\sH, \inj}}
\newcommand{\CAA}{\mathscr{C}_{\mathbb A}^{\mathbb A, \inj}}
\newcommand{\CAHH}{\mathscr{C}_{\mathbb A}^{\mathbb H}}

\newcommand{\CAT}{\mathscr{C}_{\mathbb A}^{\sT, \inj}}
\newcommand{\CATT}{\mathscr{C}_{\mathbb A}^{\mathbb T}}

\newcommand{\lra}{\{\leftarrow,\rightarrow\}}

\newcommand{\en}{\{E,N\}}

\newcommand{\cC}{\mathscr C}
\newcommand{\cD}{\mathscr D}
\newcommand{\cB}{\mathscr B}

\newcommand{\fC}{\mathscr C}
\newcommand{\fD}{\mathscr D}

\newcommand{\fF}{\mathscr F}

\newcommand{\gG}{\mathscr G}

\newcommand{\Projs}{\P}

\newcommand{\todo}[1]{#1}

\newcommand{\todom}[1]{#1}

\newcommand\new[1]{#1}

\usepackage{tabu}
\makeatletter
\newcommand{\thickhline}{%
    \noalign {\ifnum 0=`}\fi \hrule height 1pt
    \futurelet \reserved@a \@xhline
}
\newcolumntype{"}{@{\hskip\tabcolsep\vrule width 1pt\hskip\tabcolsep}}
\makeatother
\newcommand\behavior[6]{\begin{tabu}{|c|c|[1.5pt black]c|c|c|c|c|}%
\cline{4-7}%
\multicolumn{3}{c}{} & \multicolumn{2}{|c|}{$<$} & \multicolumn{2}{c|}{$>$}\\
\hline%
\multicolumn{2}{|c|[1.5pt black]}{} &  $=$ & $\rightarrow$ & $\leftarrow$ & $\rightarrow$ & $\leftarrow$ \\%
\thickhline%
\multicolumn{2}{|c|[1.5pt black]}{$=$} & $=$ & \ensuremath{#1} & \ensuremath{#2} & \ensuremath{#5} & \ensuremath{#6}\\%
\cline{1-2}
\multirow{2}{*}{$<$} & $\rightarrow$ & \ensuremath{#3} & $\rightarrow$ & $\rightarrow$ & $\rightarrow$ & $\rightarrow$ \\%
 & $\leftarrow$ & \ensuremath{#4} & $\leftarrow$ & $\leftarrow$ & $\leftarrow$ & $\leftarrow$ \\%
\hline%
\end{tabu}}

\newcommand\behaviorunordered[4]{\begin{tabu}{|c|[1.5pt black]c|c|c|c|}%
\hline%
 &  $=$ & $\rightarrow$ & $\leftarrow$ \\%
\thickhline%
$=$ & $=$ & \ensuremath{#1} & \ensuremath{#2}\\%
$\rightarrow$ & \ensuremath{#3} & $\rightarrow$ & $\rightarrow$ \\%
$\leftarrow$ & \ensuremath{#4} & $\leftarrow$ & $\leftarrow$ \\%
\hline%
\end{tabu}}

\begin{document}

\begin{abstract}
    We introduce the novel  machinery of smooth approximations, and apply it to confirm the CSP dichotomy conjecture for first-order reducts of the random tournament, \todo{and to give new short  proofs of the conjecture} for various homogeneous graphs including the random graph \todo{(STOC'11, ICALP'16)},  and for expansions of the order of the rationals \todo{(STOC'08)}.
    Apart from obtaining these dichotomy results, we show how our new proof  technique allows to unify and significantly simplify the previous results from the literature.
    For all but the last structure, we  moreover characterize \todom{for the first time} those CSPs which are solvable by local consistency methods, again using the same machinery.
\end{abstract}
\maketitle

\section{Introduction and Results}\label{sect:intro}
\subsection{Constraint Satisfaction Problems}\label{subsect:csp} 
The \emph{Constraint Satisfaction Problem}, or CSP for short, over a relational
structure $\sA$ in a finite signature  is the computational problem of deciding whether a given finite
relational structure $\sB$ in the same signature can be homomorphically mapped to
$\sA$. The structure $\sA$ is known as the \emph{template} or
\emph{constraint language} of the CSP, and the CSP of the particular structure $\sA$ is
denoted by $\csp(\sA)$.

CSPs form a class of computational problems that are of interest for practitioners and theoreticians alike.
On the one hand, CSPs generalize numerous flavours of satisfiability problems that are of interest in practice.
\todo{Fundamental problems in operations research like SAT solving or combinatorial optimization can readily be formulated as constraint satisfaction problems with $\sA$ being a structure over $\{0,1\}$ or over the rational numbers.
Moreover, planning/scheduling problems over various qualitative calculi used in temporal and spatial reasoning can be seen as constraint satisfaction problems with an infinite template; see~\cite{Qualitative-Survey} for a survey about the connections between qualitative reasoning and CSPs, and ~\cite{ClassificationTransfer,UniqueInterpolation,RelationAlgebras-Knaeuer} for examples where the CSP approach leads to results with applications in the study of relation algebras used in artificial intelligence.
Moreover, constraint satisfaction problems have recently been connected to the study of ontology-mediated queries in artificial intelligence~\cite{OBDA},
and this connection has been fruitfully exploited to obtain solutions to problems in the area of ontology-based data access as shown in~\cite{MMSNP-journal,ApproxWidth}.
This approach is particularly useful to identify in a systematic way classes of instances for which efficient algorithms exist, or to prove unconditionally that some classes of algorithms cannot solve a given problem.}
On the other hand, CSPs form a well-structured class of problems that can be studied from the point of view of complexity theory, and where one can meaningfully ask what structural information can possibly separate the well-known complexity classes such as P and NP, see e.g.~\cite{TopologyIsRelevant,BMOOPW,GJKMP-conf}.

Finite-domain CSPs, i.e., CSPs where the template is a finite structure, have been proved to either be in P or NP-complete~\cite{BulatovFVConjecture,ZhukFVConjecture,Zhuk:2020}, that is, the class of finite-domain CSPs does not contain any NP-intermediate problem, which are known to exist if P does not equal NP~\cite{Ladner}.
Moreover, the tractability border can be algorithmically decided, and it can be expressed using algebraic conditions on the set of \emph{polymorphisms} of the template (definitions are given in Section~\ref{sect:prelims}; for the remainder of the introduction, the reader can think of polymorphisms as generalized  automorphisms or endomorphisms of the template).
However, finite-domain CSPs form a proper subset of all CSPs, and in fact most of the problems that are of interest for applications can only be formulated as CSPs over infinite templates.

\subsection{Finitely bounded homogeneous structures}\label{subsect:fbh}
Fortunately, the algebraic approach that underpins the proofs of the finite-domain complexity dichotomy (and virtually all theoretical work on finite-domain CSPs) does not require the template to be finite, but also works under the assumption of $\omega$-categoricity.
And although every computational decision problem is polynomial-time Turing-equivalent to the CSP of some infinite template~\cite{BodirskyGrohe}, for a large and natural class of
$\omega$-categorical templates, which considerably expands the class of finite templates,
a similar conjecture as for finite-domain CSPs has been formulated by Bodirsky and Pinsker in 2011 (see~\cite{BPP-projective-homomorphisms});  the following is a modern formulation taking into account recent progress from~\cite{BKOPP, BKOPP-equations, wonderland, BartoPinskerDichotomy,  Topo}.

\begin{conjecture}
  \label{conjecture:infinitecsp}
  Let $\sA$ be a CSP template which is a  first-order reduct of a countable finitely bounded homogeneous structure $\sB$. Then
  one of the following holds.
  \begin{itemize}
    \item The polymorphism clone $\Pol(\sA)$ of $\sA$ has a uniformly continuous minion homomorphism to the clone of projections  $\P$,  and\/ $\csp(\sA)$ is  NP-complete. 
    \item
    The polymorphism clone $\Pol(\sA)$ has no uniformly continuous minion homomorphism to the clone of projections  $\P$, and $\csp(\sA)$ is in P.
  \end{itemize}
\end{conjecture}

The conjectured P/NP-complete dichotomy has been demonstrated for numerous subclasses: for example for all CSPs in the class
MMSNP~\cite{MMSNP}, as well as for the CSPs of the first-order reducts of $(\mathbb{Q};
<)$~\cite{tcsps-journal},
of any countable homogeneous graph~\cite{BMPP16} (including the random graph~\cite{BodPin-Schaefer-both}), of any unary structure~\cite{BodMot-Unary}, of the random poset~\cite{posetCSP16, posetCSP18}, of the homogeneous binary branching C-relation~\cite{Phylo-Complexity}, and also for the CSPs of representations of some relation algebras~\cite{RelationAlgebras-Knaeuer}.

It is easy to see from the definitions that the CSP of any template within the range of Conjecture~\ref{conjecture:infinitecsp} is in NP; moreover, the results from~\cite{wonderland} imply that if $\sA$ is such a template and $\Pol(\sA)$ does have a uniformly continuous minion homomorphism to $\P$, then its CSP is NP-hard. Resolving Conjecture~\ref{conjecture:infinitecsp} therefore involves the investigation of the consequences of the absence of a uniformly continuous  minion homomorphism from $\Pol(\sA)$ to $\P$. Only one general non-trivial algebraic consequence, the existence of a \emph{pseudo-Siggers  polymorphism}, is known~\cite{BartoPinskerDichotomy, Topo}; it remains, however, open whether this implies membership of the CSP in P.

\subsection{Reduction to the finite case}\label{subsect:reduction}
The above-mentioned complexity classifications all use 
the concept of \emph{canonical functions} from~\cite{RandomMinOps, BPT-decidability-of-definability} (see also~\cite{BodPin-CanonicalFunctions}).
However, there are two regimes for these proofs.

In the first regime, which encompasses all the above-mentioned classifications except the ones for the first-order reducts of $(\mathbb Q;<)$ and of the homogeneous binary branching $C$-relation, 
 canonical functions are used in order to reduce the problem to a CSP of a finite template.
Originally, \emph{ad hoc} algorithms were given, which were later subsumed by the general reduction from~\cite{Bodirsky-Mottet}.
Roughly speaking, if $\sA$ is a first-order reduct of a finitely bounded homogeneous structure $\sB$, then a function in $\Pol(\sA)$ is called canonical with respect to $\sB$ if it acts on the orbits of the automorphism group $\Aut(\sB)$ of $\sB$ on $n$-tuples, for all $n\geq 1$; by $\omega$-categoriciy, each of these actions has only finitely many orbits, and hence the action is on a finite set.
It is a consequence of the presence of a \emph{Ramsey expansion} for the structure $\sB$ that canonical functions are, in a sense, not too rare: canonical functions with respect to the expansion are \emph{locally interpolated} by $\Pol(\sA)$. 
It is an open problem whether every finitely bounded homogeneous structure $\sB$ has a finitely bounded homogeneous Ramsey expansion (see~\cite{BPT-decidability-of-definability,BP-reductsRamsey,BodirskyRamsey,EvansHubickaNesetril,MottetPinskerCores,Hubicka-Nesetril-All-Those}).
\todo{Until recently and the proofs of Conjecture~\ref{conjecture:infinitecsp} for MMSNP and first-order reducts of unary structures,} which established a primitive basis for the present work that we outline in Section~\ref{subsect:sa} below,
the classifications in the first regime followed the following strategy: they first identified some relations which, if preserved by $\Pol(\sA)$, implied a uniformly continuous minion homomorphism from $\Pol(\sA)$ to $\P$, and hence NP-completeness of $\csp(\sA)$.
If none of the identified relations was preserved by $\Pol(\sA)$, then it was shown that the canonical functions in $\Pol(\sA)$ satisfied non-trivial \emph{identities}, putting $\csp(\sA)$ into P by the general result from~\cite{Bodirsky-Mottet}.
In particular, the border between polynomial-time tractability and NP-hardness of the CSP is explained purely in terms of polymorphisms that are canonical with respect to the ground structure: more precisely,  $\csp(\sA)$ is in P if $\Pol(\sA)$ contains a pseudo-Siggers operation (modulo $\Aut(\sB)$) that is canonical with respect to $\sB$;  otherwise $\Pol(\sA)$ has a uniformly continuous minion homomorphism to $\P$ and $\csp(\sA)$ is NP-hard.

In the second regime, canonical functions are also used to obtain a description of ``good'' polymorphisms, which are however themselves \emph{not} canonical.
This is not a deficiency in the proof: it is known that the tractability border for first-order reducts of $(\mathbb Q;<)$ cannot be explained by the satisfaction of non-trivial identities within the clone of polymorphisms that are canonical with respect to $(\mathbb Q;<)$.
Therefore, the reduction to finite-domain CSPs from~\cite{Bodirsky-Mottet} cannot be used in this case, and different polynomial-time algorithms are featured in these classifications.

\subsection{Smooth approximations} \label{subsect:sa}

The current proof techniques suffer from several problems.
First, in both regimes, they rely on a classification of all the first-order reducts of the given base structure $\sB$ up to first-order interdefinability, and do a proof by case distinction for each of those.
This is not an approach that can scale, as even within the scope of Conjecture~\ref{conjecture:infinitecsp} the number of such reducts can grow arbitrarily large, and it is even an open question to know whether this number is always finite (the question is a special case of Thomas' conjecture~\cite{RandomReducts}).

Secondly, the proofs rely on identifying, for each structure $\sB$, a different list of relations whose invariance under $\Pol(\sA)$   implies  NP-hardness, and whose non-invariance  should imply containment of  $\csp(\sA)$ in P. The non-invariance  of the relations is however a \emph{local} condition: it expresses that there exists a finite subset $S$ of the domain on which some polymorphisms  violate the relations, but the polymorphisms might well leave these relations invariant on the complement of $S$.
In order to derive significant structural information (in particular,  information sufficient to prove tractability of $\csp(\sA)$) from this sort of local information, the proofs need several local-to-global arguments which are often long, complex, and highly depend on the structure $\sB$ under consideration.

We develop here a different strategy which still uses canonical functions, but avoids the two problems mentioned above.
It is based on the comparison of two polymorphism clones $\fC, \fD$ with $\fC\subseteq \fD$.
\begin{itemize}
    \item In our first result, the \emph{loop lemma of  approximations} (Theorem~\ref{thm:2-cases-general}), we assume that $\fC$ acts on the orbits of the largest permutation group $\gG_\cD$ within $\cD$ as the projections. We moreover assume that $\cD$ is a \emph{model-complete core} and has \emph{no algebraicity}.
    The lemma then 
roughly states that there exists a $\cC$-invariant equivalence relation $\sim$ such that $\cC$ acts on the $\sim$-classes as the projections, and such that either  $\sim$ is \emph{smoothly approximated}  by an equivalence relation $\eta$ which is invariant under $\cD$, or certain  binary relations which are invariant under $\cD$ must all contain a \emph{pseudo-loop} modulo $\gG_\cD$. The former situation can be viewed as $\cD$ being ``close" to $\cC$, the latter as it being ``far".
\item 
We then show the \emph{fundamental theorem of smooth approximations} (Theorem~\ref{thm:sa}): if $\cC$ acts on the classes of an invariant equivalence relation $\sim$, and $\sim$ is smoothly approximated by an equivalence relation invariant under $\cD$, and $\cC$ is \emph{locally interpolated}  by $\cD$ (another ``closeness" condition), then the action of $\cC$ on the $\sim$-classes can be extended to $\cD$. 
\item Finally, in a third result we show that in the second case of the loop lemma of approximations, where certain binary relations invariant under $\fD$ must contain a pseudo-loop, $\fD$ contains a certain  \emph{weakly symmetric} function  (Lemma~\ref{lem:weird-symmetric}).
The symmetry condition on this operation is \emph{global}, i.e., it holds on the whole domain and not only on a finite subset thereof.
This solves the second shortcoming mentioned above.
\end{itemize}

These results make, in principle, the following strategy for proving dichotomy results possible: if the clone of canonical functions in $\Pol(\sA)$ (which is a subclone of $\Pol(\sA)$) has no \emph{clone homomorphism} to $\P$ (i.e., it satisfies non-trivial identities), then $\csp(\sA)$ is in P, by~\cite{Bodirsky-Mottet}. Otherwise, it acts trivially on the classes of some invariant equivalence relation~\cite{Topo-Birk}. If this equivalence relation is smoothly approximated by an equivalence relation invariant under $\Pol(\sA)$, then $\Pol(\sA)$ has a uniformly continuous minion homomorphism to $\P$, and $\CSP(\sA)$ is NP-complete, by~\cite{wonderland}. If on the other hand the equivalence relation is not smoothly approximated, then $\Pol(\sA)$ contains a weakly symmetric function, which allows us to construct a canonical function satisfying a non-trivial identity -- a contradiction. %

In reality, our proofs will not be quite as simple, and involve several ``intermediate clones" between the larger clone $\Pol(\sA)$ and its subclone of canonical functions. The  main idea illustrated above of ``either a clone homomorphism expands to a larger clone, or the larger clone is non-trivial"  remains, however, the strategy's \emph{Leitmotiv}, its \emph{Schwibbogen}.

\subsection{The dichotomy results}\label{subsect:dicho} The new methods that we develop in this paper allow us to confirm the Bodirsky-Pinsker conjecture for several subclasses.
As a matter of fact, to our knowledge, every single dichotomy result that was published so far within the scope of the Bodirsky-Pinsker conjecture  can be reproved using various combinations of our methods.
As a proof of concept, we show in this paper how to prove the conjecture for first-order reducts of the universal homogeneous tournament -- this result is new.
We also show how to reprove the dichotomies for the first-order reducts of the universal homogeneous graph (known as the \emph{random graph}), the universal homogeneous $K_n$-free graph for all $n\geq 3$, and for expansions of the order $(\mathbb Q;<)$ of the rational numbers using smooth approximations.
Our new proofs of these results are considerably shorter and simpler than the original ones: once our general theory is laid down, the arguments that are specifically needed for these proofs only cover a few pages, while the original proofs add up to over 70 pages of specific arguments~\cite{BodPin-Schaefer-both,BMPP16,tcsps-journal}. 
Moreover, they follow the same unifying principles,
thus showing the versatility and potential of our approach for Conjecture~\ref{conjecture:infinitecsp}: this shows nicely in our new proof of the P/NP dichotomy for the homogeneous graphs in Section~\ref{sect:graphs} which, reusing the strategy of the universal homogeneous tournament from Section~\ref{sect:rt}, is achieved in only 4~pages (most of which are rather superfluous once the transfer principle of the proof is explained)  and requires no more creativity.

\begin{theorem}\label{thm:dichotomies}
Conjecture~\ref{conjecture:infinitecsp} is true for first-order reducts of the universal homogeneous tournament, of the random graph,
of the universal homogeneous $K_n$-free graph for all $n\geq 3$,
and for expansions of the rational numbers.
\end{theorem}

As mentioned above, in the case of the first-order expansions of $(\mathbb Q;<)$,
one cannot use the reduction to finite-domain CSPs from~\cite{Bodirsky-Mottet} and therefore one needs to exhibit specific algorithms.
In the original proof and its recent refinement~\cite{tcsps-journal,RydvalFP}, four distinct subroutines are needed, one for each minimal operation giving tractability (the so-called operations $min,mx,mi,ll$), together with two ``master'' algorithms using these routines (one for each of the so-called operations $pp$ and $ll$).
The advantage of our proof via a reduction to a finite-domain CSP is that we do not need to exhibit these subroutines, which come instead from the finite-domain dichotomy theorem (or from Schaefer's boolean dichotomy theorem~\cite{Schaefer}, in this special case).
In particular, our approach is robust against the explosion of the number of cases that one would have to consider in a more general setting.
However, our general theory does not provide at the moment a unifying argument for the two ``master'' algorithms.

\subsection{Bounded width}\label{subsect:bw} For fixed parameters $k\leq l$, the \emph{local consistency algorithm} works as follows: given an instance, one derives the strongest possible constraints on $k$ variables that can be seen by looking at $l$ variables at a time.
If one derives an empty constraint, then one can safely reject the instance.
We say that the local consistency algorithm \emph{solves} a CSP if whenever no empty constraint is derived, then the instance has a solution.
A template $\sA$ has \emph{bounded width} if for some $k,l$, the local consistency algorithm correctly solves $\sA$ (in which case the \emph{width} of $\sA$ is the lexicographically smallest pair $(k,l)$)~\cite{FederVardi}.
Since this algorithm runs in polynomial time, CSPs of templates with bounded width are in P.

For finite structures, an algebraic characterization of templates with bounded width is known~\cite{BoundedWidthJournal}, but it is also known that no condition that is preserved by uniformly continuous clone homomorphisms characterizes bounded width for structures within the scope of the Bodirsky-Pinsker conjecture~\cite{RydvalFP}.
Using the method of smooth approximations, we provide a characterization of templates with bounded width for reducts of the structures mentioned above, with the exception of $(\mathbb Q;<)$ \todo{for which a characterization of bounded width is already known~\cite{RydvalFP}}.
\todo{This result is entirely new and constitutes} major progress towards an understanding of bounded width in general, and is the first result of its kind in that it does not rely on a previous detailed complexity classification for the CSPs.
It provides further evidence that the approximation approach developed here is a fundamental tool that can be used to investigate some of the most important questions in the area of infinite-domain CSPs, especially so since the structure of our proofs for characterizing bounded width is exactly the same as the one for characterizing membership in P. Our proof for the case of homogeneous graphs, building on the proofs for the random tournament and our P/NP classifications, becomes almost a footnote.

\begin{theorem}\label{thm:bounded-width-classification}
    Let $\sA$ be a \todom{model-complete core that is a} first-order reduct of $\sB$, where $\sB$ is the universal homogeneous tournament, the random graph, or the universal $K_n$-free graph for some $n\geq 3$.
    Then the following are equivalent:
    \begin{itemize}
        \item $\sA$ has bounded width;
        \item for every $m\geq 3$, $\sA$ has an $m$-ary pseudo-WNU polymorphism \todom{modulo $\overline{\Aut(\rel B)}$};
        \item for every $m\geq 3$, $\sA$ has an $m$-ary pseudo-WNU \todom{polymorphism modulo $\overline{\Aut(\rel B)}$} that is canonical with respect to $\sB$.
    \end{itemize}
\end{theorem}

We note that Wrona recently started the investigation of what the optimal parameters $k$ and $l$ should be, when $\sA$ has bounded \emph{strict} width~\cite{Wrona:2020b,Wrona:2020a}.
A template has bounded strict width when there exist $k,l$ such that whenever the local consistency does not reject an instance, then a solution can be constructed by a greedy approach (we refer to~\cite{Wrona:2020a,FederVardi} for a precise definition).
Wrona has proved that if $\sA$ is a first-order \emph{expansion} of the random graph or one of the universal homogeneous $K_n$-free graphs
that has bounded \emph{strict} width, then the width of $\sA$ is exactly $(2,\max\{3,n\})$.
The two requirements of being an expansion with strict width are needed in Wrona's argument because strict width is characterized by an algebraic condition that is crucial for the argument (in particular, some of the claims in~\cite{Wrona:2020b} do not hold in the more general setting of bounded width).
\todo{Using Theorem~\ref{thm:bounded-width-classification}, Wrona's results have recently been improved in~\cite{ApproxWidth}, where a collapse of the bounded width hierarchy was shown for first-order reducts of the structures in the scope of Theorem~\ref{thm:bounded-width-classification}.
Moreover, smooth approximations have  also been used in~\cite{ApproxWidth} to solve the Datalog-rewritability problem for monadic disjunctive Datalog, thereby solving an open problem from~\cite{OBDA} and providing another example of the versatility of our approach.}

\subsection{Outline}\label{subsect:outline} After providing the necessary notions and definitions in Section~\ref{sect:prelims}, we prove our general results on smooth approximations outlined in Section~\ref{subsect:sa}; this will be achieved in Section~\ref{sect:sa}. We apply these results to obtain the P/NP dichotomy, as well as the bounded width dichotomy for the random tournament in Section~\ref{sect:rt}. In Section~\ref{sect:graphs}, we obtain analogous results for homogeneous graphs, and in Section~\ref{sect:temp}, we reprove the P/NP dichotomy for first-order reducts of the order of the rationals. \section{Preliminaries}\label{sect:prelims}
\subsection{CSPs}
If $\sA$ is a relational structure in a finite signature, called a \emph{CSP template}, then $\csp(\sA)$ is the set of all finite structures $\sC$ in the same signature with the property that there exists a homomorphism from $\sC$ into $\sA$. This set can be viewed as a computational problem where we are given a finite structure $\sC$ in that signature, and we have to decide whether $\sC\in\csp(\sA)$.

We will tacitly assume that all relational structures that appear in this article, as well as their signatures, are at most countably infinite.

\subsection{Model theory}

A relational structure $\sB$ is \emph{homogeneous} if every  partial isomorphism between finite induced substructures of $\sB$ extends to an automorphism of the entire structure. For example, the linear order of the rationals $(\rel Q;<)$ is homogeneous; it is moreover \emph{universal} for the class of finite linear orders in the sense that it contains any such order as an induced substructure. Homogeneity and universality define $(\rel Q;<)$ up to isomorphism. Similarly, there exist a unique universal homogeneous  tournament, undirected loopless graph (called the  \emph{random graph}), and undirected loopless  $K_n$-free graph for all $n\geq 3$. A relational structure $\sB$ is called \emph{finitely bounded} if its \emph{age}, i.e., its finite  induced  substructures up to isomorphism, is given by a finite set ${\mathcal F}$ of forbidden finite substructures: that is, its age consists precisely of those finite structures in its signature which do not embed any member of ${\mathcal F}$. All of the above-mentioned universal homogeneous structures are finitely bounded. An element of ${\mathcal F}$ is \emph{minimal} if no induced substructure is an element of ${\mathcal F}$; this is the case if and only if it must be present in any set ${\mathcal F'}$ which serves as a description of the age of $\sB$ as above. A structure $\sB$ is \emph{Ramsey} if its age satisfies a certain combinatorial regularity property; we will not need the  definition, but only a consequence  from~\cite{BPT-decidability-of-definability} which we are going to cite at the end of this section. Except for $(\rel Q;<)$, the structures above are not Ramsey, but they have a finitely bounded homogeneous \emph{Ramsey expansion}: a linear order can be added freely (i.e., in a way that the new age consists of the old age ordered in all possible ways) 
so they become homogeneous, finitely bounded, and Ramsey \todo{(this follows, e.g., from the Ne\v{s}et\v{r}il-R\"odl theorem~\cite{NesetrilRoedlOrderedStructures})}. A structure $\sB$ is \emph{transitive} if its automorphism group has only one orbit in its action on the domain of $\sB$. All of the above structures are transitive.

A \emph{first-order reduct} of a relational structure $\sB$ is a relational structure $\sA$ on the same domain whose relations are  first-order definable without parameters in $\sB$.
Every first-order reduct $\sA$ of a finitely bounded homogeneous structure is \emph{$\omega$-categorical}, i.e., its  automorphism group $\Aut(\sA)$ has finitely many orbits in its componentwise action on $n$-tuples of elements of its domain, for all finite $n\geq 1$. Permutation groups with the latter property are called \emph{oligomorphic}. 
An $\omega$-categorical structure $\sB$ has \emph{no algebraicity} if none of its elements is first-order definable using other elements as parameters; this is the case if all stabilizers of its automorphism group by finitely many elements have only infinite orbits (except for the orbits of the stabilized elements). We also say that arbitrary permutation groups with the latter property have no algebraicity; such groups have the property that all orbits under their action on $n$-tuples contain two disjoint tuples, for all $n\geq 1$.  %
First-order reducts of  structures without algebraicity have no algebraicity.

A structure $\sA$ is a \emph{model-complete core} if for all endomorphisms of $\sA$ and all finite subsets of its domain there exists an automorphism of $\sA$ which agrees with the endomorphism on the subset. If $\sA$ is any $\omega$-categorical structure, then there exists an $\omega$-categorical model-complete core $\sA'$ with the same CSP as $\sA$; this structure $\sA'$ is unique up to isomorphism~\cite{Cores-journal}.

A formula is \emph{primitive positive}, in short \emph{pp}, if it contains only existential quantifiers, conjunctions, equalities, and relational symbols.  If $\sA$ is a relational structure, then a relation is \emph{pp-definable} in $\sA$ if it can be defined by a pp-formula in
$\sA$. For any $\omega$-categorical  model-complete core $\sA$,  all orbits of the action of  $\Aut(\sA)$ on $n$-tuples are pp-definable, for all $n\geq 1$. In this paper, pp-definability of the binary disequality  relation  $\neq{}:=\{(x,y)\in A^2\; |\; x\neq y\}$ on a set $A$ will play a role;  we will use this single notation for various sets $A$ which will however be clear from context.

\subsection{Universal algebra}\label{subsect:prelims-ua}
A \emph{polymorphism} of a relational structure $\sA$ is a homomorphism from some finite power $\sA^n$ of the structure to $\sA$.
The set of all polymorphisms of $\sA$  is called the \emph{polymorphism clone} of $\sA$ and is denoted by $\Pol(\sA)$. Polymorphism clones are special cases of \emph{functions clones}, i.e., sets of finitary operations on a fixed set which contain all projections and which are closed under arbitrary composition. 

The \emph{domain} of a function clone $\cC$ is the set $C$ where its functions are defined; we also say that $\cC$ \emph{acts on $C$}. The clone $\cC$ then  also naturally acts on any power $C^I$ of $C$ (componentwise), on any invariant subset $S$  of $C$ (by restriction), and on the classes of any invariant equivalence relation $\sim$ on an invariant subset $S$ of $C$ (by its action on representatives of the classes). We write $\cC\actson C^I$, $\cC\actson S$, and $\cC\actson S/{\sim}$ for these actions, \todom{and call any such pair $(S,\sim)$ a \emph{subfactor} of $\cC$. A subfactor $(S,\sim)$ of $\cC$ is \emph{minimal} if \new{$\sim$ has at least two equivalence classes} and any $\cC$-invariant subset of $S$  intersecting two $\sim$-classes  equals $S$.}

For a permutation group $\gG$, an operation $f$ on its domain, and $n\geq 1$, we say that $f$ is \emph{$n$-canonical} with respect to $\gG$ if it acts on the $\gG$-orbits of $n$-tuples; in other words, orbit-equivalence on $n$-tuples is invariant under $f$. It is \emph{canonical with respect to $\gG$} if it is $n$-canonical for all $n\geq 1$.  If $\sB$ is a structure, we also say that $f$ is canonical with respect to $\sB$ if it is canonical with respect to $\Aut(\sB)$. We say that a function clone $\cC$ is $n$-canonical (or canonical, respectively) with respect to $\gG$ if all of its functions are. In that case, we write $\cC^n/{\gG}$ for the action of $\cC$ on the set $C^n/{\gG}$ consisting of
$\gG$-orbits of $n$-tuples of the domain $C$ of $\cC$. 
If $\gG$ is oligomorphic then $\cC^n/{\gG}$ is a clone on a finite set.
Note that a $k$-ary $f$ is canonical with respect to  $\gG$ if for all tuples $a_1,\ldots,a_k$ of the same length and all $\alpha_1,\ldots,\alpha_k\in\gG$ there exists $\beta\in\gG$ such that $f(a_1,\ldots,a_k)=\beta\circ f(\alpha_1(a_1),\ldots,\alpha_k(a_k))$. We say that it is \emph{diagonally canonical} if the same holds in case $\alpha_1=\cdots=\alpha_k$.

A function $f$ is \emph{idempotent} if $f(x,\ldots,x)=x$ for all values $x$ of its domain; a function clone is idempotent if all of its functions are. A function is \emph{essentially unary} if it depends on at most one of its variables, and \emph{essential} otherwise.

If $\cC=\Pol(\sA)$ is a polymorphism clone, then the unary functions in $\cC$ are precisely the endomorphisms of $\sA$, denoted by $\End(\sA)$. The unary bijective functions in $\cC$ whose inverse is also contained in $\cC$ are precisely the automorphisms of $\sA$, denoted by $\Aut(\sA)$. For an arbitrary function clone $\cC$, we write $\gG_\cC$ for the permutation group of unary functions in $\cC$ which have an inverse in $\cC$. We say that $\cC$ has no algebraicity if $\gG_\cC$  has no algebraicity, we say it is oligomorphic if $\gG_\cC$ is.

An \emph{identity} is a formal expression $s(x_1, \ldots, x_n) = t(y_1, \ldots, y_m)$
where $s$ and $t$ are abstract terms of function symbols, and $x_1,
\ldots, x_n, y_1, \ldots, y_m$ are the variables that appear in these terms. The identity is of \emph{height 1}, and called \emph{h1~identity}, if the terms $s$ and $t$ contain precisely one function symbol, i.e., no
nesting of function symbols is allowed, and no term may be just a variable. A \emph{cyclic} identity is an identity of the form $f(x_1,\ldots,x_n)=f(x_2,\ldots,x_n,x_1)$, and a \emph{weak near-unanimity (WNU)} identity is of the form $w(x,\ldots,x,y)=\cdots= w(y,x,\ldots,x)$. The \emph{majority} identities are given by $m(x,x,y)=m(x,y,x)=m(y,x,x)=x$, the \emph{minority} identities by $m(x,x,y)=m(x,y,x)=m(y,x,x)=y$, and the \emph{Siggers} identity by $s(x,y,x,z,y,z)=s(y,x,z,x,z,y)$. Each set of identities also has a \emph{pseudo}-variant obtained by composing each term appearing in the identities with a distinct unary function symbol: for example, the \emph{pseudo-Siggers} identity is given by $e\circ s(x,y,x,z,y,z)=f\circ s(y,x,z,x,z,y)$.

A set $\Theta$ of identities  is \emph{satisfied} in a function clone $\cC$ if the function symbols of $\Theta$ can be assigned functions in $\cC$
 in such a way that all identities of $\Theta$ become true for all possible values of their variables of the domain.
A set of identities is called \emph{trivial} if it is satisfied in the \emph{projection clone} $\P$ consisting of the projection operations on the set $\{0,1\}$.
Otherwise, the set is called \emph{non-trivial}. %
A function is called a cyclic / weak near-unanimity (WNU) / majority / minority / Siggers  operation if it satisfies the identity of the same name. For the pseudo-variants of these identities, e.g.~the pseudo-Siggers identity, and a set of unary functions  $\mathscr F$, we also say that a function $s$ is a \emph{pseudo-Siggers operation modulo $\mathscr F$} if satisfaction of the pseudo-Siggers identity is witnessed by $s$ and unary functions from $\mathscr F$.

A map $\xi\colon \cC \to  \cD$ between function clones is called a \emph{clone homomorphism} if it preserves arities, maps the $i$-th $n$-ary projection in $\cC$ to the $i$-th $n$-ary projection in $\cD$ for all $1\leq i\leq n$,
and satisfies $\xi(f \circ (g_1, \ldots, g_n)) = \xi(f) \circ (\xi(g_1), \ldots, \xi(g_n))$ for all $n,m\geq 1$, all $n$-ary $f \in \cC$, and all $m$-ary
$g_1, \ldots, g_n \in\cC$.
This is the case if and only if the map $\xi$ \emph{preserves identities}, i.e., whenever some functions in $\cC$ witness the satisfaction of some identity in $\cC$, then their images under $\xi$ witness the satisfaction of the same identity in $\cD$. It would be a consequence of the truth of  Conjecture~\ref{conjecture:infinitecsp} that for any two structures  within its range the existence of  a clone isomorphism between the respective polymorphism clones implies that the CSPs of the two structures are polynomial-time interreducible. This is also open under the weaker assumption of $\omega$-categoricity; under the assumption of \emph{uniform continuity} on the clone homomorphism it is a fact~\cite{Topo-Birk} the exploitation of which is known as the \emph{algebraic approach to CSPs}.

A map $\xi \colon \cC \to \cD$ is called a \emph{minion homomorphism} 
if it preserves arities and composition with projections; the latter meaning that for all $n,m\geq 1$, all $n$-ary $f \in \cC$, and all $m$-ary projections $p_1,\ldots,p_n\in\cC$, we have  $\xi(f \circ (p_1,\ldots,p_n)) = \xi(f) \circ (p_1',\ldots,p_n')$, where $p_i'$ is the $m$-ary projection in $\cD$ onto the same variable as $p_i$, for all $1\leq i\leq n$. 
This is the case if and only if the map $\xi$ preserves h1~identities in the sense above.

The existence of clone and minion homomorphisms $\cC\to\P$ is connected to the satisfaction of non-trivial identities in a polymorphism clone $\cC$. Namely, there exists a clone homomorphism $\cC \to \P$ if and only if every set of identities satisfied in $\cC$ is trivial, in which case we say that $\cC$ is \emph{equationally trivial}. Similarly,  there exists a minion homomorphism $\cC \to \P$ if and    only if every set of h1 identities satisfied in $\cC$ is trivial. We say that $\cC$ is \emph{equationally affine} if it has a clone homomorphism to a clone $\mathscr M$ of affine maps over a finite module.

If $\cC,\cD$ are function clones and $\cD$ has a finite domain, then a clone (or minion)  homomorphism $\xi\colon \cC\to\cD$ is \emph{uniformly continuous} if for all $n\geq 1$ there exists a finite subset $B$ of $F^n$ such that $\xi(f)=\xi(g)$ for all $n$-ary $f,g\in \cC$ which agree on $B$. %

If $\fF$ is a set of functions on a fixed set $C$, then $\overline{\fF}$ denotes the closure of $\fF$ in the topology of pointwise convergence: that is, a function $g$ is contained in $\overline{\fF}$ if for all finite subsets $F$ of $C$, there exists a function in $\fF$ which agrees with $g$ on $F$. For functions $f,g$  on the same domain, and $\gG$ a permutation group on this domain, we say that $f$ \emph{locally interpolates $g$ modulo $\gG$} if $g(x_1,\ldots,x_k)\in\overline{\{\beta\circ f(\alpha_1(x_1),\ldots,\alpha_k(x_k))\; |\; \beta,\alpha_1,\ldots,\alpha_k\in\gG\}}$. For function clones $\cC,\cD$ on  this domain, we say that  $\cD$ \emph{locally interpolates $\cC$ modulo $\gG$} 
if 
every function in $\cD$ locally interpolates some  function in $\cC$ modulo $\gG$. If $\cC'$ is the function clone of those functions in $\cC$ which are canonical with respect to $\gG$, and $\gG$ is the automorphism group of a \todom{homogeneous} Ramsey structure \todom{in a finite signature}, then $\cC$ locally interpolates $\cC'$ modulo $\gG$~\cite{BPT-decidability-of-definability, BodPin-CanonicalFunctions}. Similarly, we define \emph{diagonal interpolation} by $g\in\overline{\{\beta\circ f(\alpha(x_1),\ldots,\alpha(x_k))\; |\; \beta,\alpha\in\gG\}}$. If $\gG$ is the automorphism group of a \todom{homogeneous} Ramsey structure \todom{in a finite signature}, then every function diagonally interpolates a diagonally canonical function modulo $\gG$.

We say that a function clone $\cC$ is a model-complete core if its unary functions coincide with $\overline{\gG_\cC}$. If $\cC=\Pol(\sA)$ for an $\omega$-categorical $\sA$, then this is the case if and only if $\sA$ is a model-complete core.

If $\sA$ is a relational structure and $R$ is a relation on its domain which is pp-definable in $\sA$, then $R$ is invariant under  $\Pol(\sA)$; \todo{conversely, invariant relations are pp-definable under the assumption that $\sA$ is $\omega$-categorical~\cite{BodirskyNesetrilJLC}}. \section{Smooth Approximations}\label{sect:sa}
We first give precise definitions of approximation notions for equivalence relations of varying  strength, including the smooth approximations which appeared in the outline in Section~\ref{subsect:sa}. Then we prove the three results announced there: the  loop lemma (Section~\ref{subsect:looplemma}), the fundamental theorem (Section~\ref{subsect:ft}), and the existence of weakly commutative functions (Section~\ref{subsect:weird-symmetric}).

\begin{definition}[Smooth approximations]
Let $A$ be a set,  $n\geq 1$, and let  $\sim$ be an equivalence relation on a subset $S$ of $A^n$.
We say that an equivalence relation $\eta$ on some set $S'$ with $S\subseteq S'\subseteq A^n$  \emph{approximates} $\sim$ if the restriction of $\eta$ to $S$ is a (possibly non-proper)  refinement of $\sim$; we call $\eta$ an \emph{approximation} of $\sim$.

For a permutation group $\gG$ acting on $A$ and \todo{leaving the $\sim$-classes invariant as well as $\eta$}, we say that the approximation $\eta$ is
\begin{itemize}
    \item \emph{very smooth} if  orbit-equivalence with respect to $\gG$ is a (possibly non-proper)  refinement of $\eta$ on $S$;
    \item \emph{smooth} if each equivalence class $C$  of $\sim$ intersects some equivalence class $C'$ of $\eta$ such that $C\cap C'$ contains a $\gG$-orbit;
    \item \todo{\emph{presmooth}} if each equivalence class $C$ of $\sim$ intersects some  equivalence class $C'$ of $\eta$ such that $C\cap C'$ contains two disjoint tuples in the same $\gG$-orbit.
\end{itemize}
\end{definition}
Note that trivially, any very smooth approximation is smooth. Smoothness clearly implies \todo{presmoothness} if $\gG$ has no algebraicity \todom{(see our remark in the preliminaries)}, which will appear as an assumption in the general results to follow. \todo{Presmooth} approximations are useless in that we will need at least smoothness in the fundamental theorem; in the following, we observe conditions which allow us to upgrade  an approximation in strength.

\begin{definition}[``Primitivity"]
Let $A$ be a set and $n\geq 1$. 
    A permutation group $\gG$ acting on $A$ is \emph{$n$-``primitive''} if for every orbit $O\subseteq A^n$ of $\gG$, 
    \todo{every $\gG$-invariant equivalence relation on $O$ containing $(a,b)$ with $a,b$ disjoint is full.} A function clone $\cC$ is $n$-``primitive" if $\gG_\cC$ is.
\end{definition}

\begin{example}\label{ex:primitive-tournament}
    The automorphism group of the  universal homogeneous tournament $\rel T$ is $n$-``primitive'' for all $n\geq 1$.
    Indeed, let $n\geq 1$, let $O$ be an orbit of $n$-tuples under $\Aut(\rel T)$ and let $\sim$ be an equivalence relation containing $(a,b)$ with $a,b$ disjoint.
    Without loss of generality, $O$ can be assumed to be an orbit of injective tuples.
    Let $c,d$ be arbitrary tuples in $O$.
    Consider the digraph $\sX$ over $3n$ elements $x_1,\dots,x_n,y_1,\dots,y_n,z_1,\dots,z_n$ such that
    the subdigraphs induced by $(x_1,\dots,x_n,y_1,\dots,y_n)$ and $(y_1,\dots,y_n,z_1,\dots,z_n)$ are isomorphic to the structure induced by $(a_1,\dots,a_n,b_1,\dots,b_n)$ in $\rel T$,
    and such that the subdigraph induced by $(x_1,\dots,x_n,z_1,\dots,z_n)$ is isomorphic to the structure induced by $(c_1,\dots,c_n,d_1,\dots,d_n)$ in $\rel T$.
    By universality of $\rel T$, there is an embedding of $\sX$ into $\rel T$ and by homogeneity one can assume that the embedding maps $(x_1,\dots,x_n,z_1,\dots,z_n)$  to $(c_1,\dots,c_n,d_1,\dots,d_n)$.
    By transitivity of $\sim$, one obtains that $c\sim d$, so that $\sim$ is full.
\end{example}

\begin{lemma}%
\label{lem:primitive-implies-smooth}
    Let $A$ be a set, $n\geq 1$, and let $\sim$ be an equivalence relation on a subset $S$ of $A^n$. Let $\gG$ be an $n$-``primitive'' permutation group acting on $A$ and leaving the $\sim$-classes invariant. Then any  approximation $\eta$ of $\sim$ which is \todo{presmooth} is also smooth with respect to $\gG$.
\end{lemma}
\begin{proof}
Pick any equivalence class $C$  of $\sim$, and let $C'$ be an $\eta$-equivalence class therein such that $C\cap C'$ contains two disjoint tuples $a,b$ in the same $\gG$-orbit $O$. We have $O\subseteq C$ since $\gG$ preserves the $\sim$-classes.  
The support of $\eta$ contains $S$ and hence also $O$; moreover, the restriction of $\eta$ to $O$ is invariant under $\gG$ and relates  two disjoint tuples.
Hence this restriction is full, so $C'$ contains $O$ as well, which yields $C\cap C'\supseteq O$.
\end{proof}

The following lemma, which will allow us to obtain very smooth approximations from \todo{presmooth} ones in our applications, might be better appreciated after perusal of Sections~\ref{subsect:looplemma} and~\ref{subsect:ft}. 
\begin{lemma}\label{lem:primitive-neq-very-smooth}
    Let $A$ be a set and  $n\geq 1$.  Let $\cC$ be a function clone on $A$ leaving $\neq$ invariant, and let $\gG$ be a permutation group acting on $A$ which is $n$-``primitive'' and such that $\cC$ is $n$-canonical with respect to $\gG$. 
    \todom{Let $(S,\sim)$ be a minimal subfactor of  $\cC\actson A^n$ with ${\gG}$-invariant  $\sim$-classes. %
   Then any $\cC$-invariant approximation $\eta$  of $\sim$ that is \todo{presmooth} is also very smooth with respect to $\gG$.}
\end{lemma}
\begin{proof}
    Consider the set $T\subseteq A^n$ consisting of those $a\in A^n$ which are $\eta$-equivalent to some disjoint tuple  $b\in A^n$ in the same $\gG$-orbit.
    This set is invariant under $\cC$ since $\neq$, orbit-equivalence with respect to $\gG$, and $\eta$ are. 
    We claim that $S\subseteq T$. 
    Indeed, since $\eta$ is a \todo{presmooth} approximation of $\sim$,
    we find $a,b$ in $T\cap S$ that are not $\sim$-equivalent; \todom{here, we use in particular that $\sim$ has at least two classes by the minimality of the
    subfactor $(S,\sim)$.}
    Since $T$ is invariant under $\cC$, we must therefore have $S\subseteq T$, for otherwise restriction of the action $\cC^n/{\gG}\actson S/{\sim}$ to $(S\cap T)/{\sim}$ would contradict the   minimality \todom{of $(S,\sim)$}.  
    Thus, for every $a\in S$, there is $b$ in the same ${\gG}$-orbit as $a$, disjoint from $a$, and such that $a\eta  b$.
    By the  $n$-``primitivity''of ${\gG}$, it follows that $\eta$ is full when restricted to the orbit of $a$. Since $a$ was chosen arbitrarily, it follows that $\eta$ is very smooth.
\end{proof}

\subsection{The loop lemma}\label{subsect:looplemma}

\begin{definition}[Naked set]
Let $A$ be a set, $n\geq 1$, and let $\cC$ be a function clone on $A$. \todom{An \emph{$n$-ary naked set of $\cC$} is a 
subfactor $(S,\sim)$ of $\cC\actson A^n$ such that $\sim$ has least two equivalence classes and such that $\cC$ acts  on $S/{\sim}$ by projections. }
\end{definition}

\todom{Note that any minimal naked set, i.e., one where the subfactor $(S,\sim)$ is minimal, has the property that $\sim$ has precisely two classes.}
The existence of a naked set of an oligomorphic  function clone $\cC$ is equivalent to the existence of a uniformly continuous clone homomorphism from $\cC$ to $\P$~\cite{uniformbirkhoff, Topo-Birk}. The loop lemma of approximations, which has its root in~\cite[Proposition~44]{MMSNP}, examines under certain conditions  the consequence of a larger function clone $\cD$ not ``inheriting" a naked set of $\cC$.

In the proof we are going to use the following classical result by Bulatov~\cite{BulatovHColoring}: if $\cB$ is an idempotent clone acting on a finite set and preserving an undirected graph that contains a   cycle of odd length but no loop, then $\cB$ is equationally trivial. \todom{Here, by a \emph{cycle}  we mean any sequence $(x_0,\ldots,x_\ell)$ of vertices of the graph such that  $x_0=x_\ell$ and such that $x_i$ and $x_{i+1}$ are adjacent for all $i\in\{0,\ldots,\ell-1\}$; the number $\ell$ is the \emph{length} of the cycle, and a cycle of length $1$ is a \emph{loop}.} Now let $A$ be a set and $n\geq 1$, let $\cC$ be a function clone on $A$, and let $\gG$ be an oligomorphic  permutation group on $A$  such that  $\cC$ acts idempotently on the $\gG$-orbits of $n$-tuples. 
Every binary symmetric  relation $R$ on a set $A^n$ induces an undirected graph $R/\gG$ on the finite set $A^n/\gG$ of orbits  in a natural way: two orbits are adjacent in $R/\gG$ iff there exist  tuples in these orbits that are related by $R$.
Thus, if $R/\gG$ is loopless and contains an odd cycle, and if $R$ is invariant under $\cC$, then the action  $\cC^n/\gG$ is equationally trivial.

\begin{theorem}[The loop lemma of  approximations]\label{thm:2-cases-general}
	Let $A$ be a set and  $n\geq 1$. Let $\cC\subseteq\cD$ be oligomorphic function clones on $A$, where $\cD$ is a model-complete core without algebraicity  and $\cC$ is $n$-canonical with respect to $\gG_\cD$.
	Suppose that  $\cC^n/{\gG_\cD}$ is equationally  trivial.
	Then there exists an $n$-ary naked set $(S,\sim)$ of  $\cC$ with ${\gG_\cD}$-invariant 
	$\sim$-classes  such that one of the following holds:
\begin{enumerate}
    \item $\sim$ is  approximated by a $\cD$-invariant  equivalence relation that is \todo{presmooth} with respect to $\gG_\cD$;
    \item every $\cD$-invariant binary symmetric relation $R\subseteq (A^n)^2$ that  contains a pair $(a,b)\in S^2$ such that $a$ and $b$ are disjoint and such that   $a\not\sim b$ contains a pseudo-loop modulo $\gG_\cD$, i.e., a pair $(c,c')$ where $c,c'$ belong to the same $\gG_\cD$-orbit.
\end{enumerate}
\end{theorem}

\begin{proof}
	Since $\cD$ is an oligomorphic model-complete core,   $\cC^n/{\gG_\cD}$  is an idempotent function clone on a finite set, and thus since $\cC^n/{\gG_\cD}$ is equationally trivial it has a $1$-ary \todom{minimal} naked set  %
	(see, e.g., \cite[Proposition 4.14]{BulatovJeavons}). Any such naked set $(S',\sim')$ induces an $n$-ary \todom{minimal} naked set $(S,\sim)$ of $\cC$ with  ${\gG_\cD}$-invariant $\sim$-classes: the $\sim$-classes are simply the unions of the orbits of the $\sim'$-classes. 
	In the following, a \emph{witness} is a  $n$-ary minimal naked set $(S,\sim)$ of $\cC$ with ${\gG_\cD}$-invariant $\sim$-classes; \todom{note that by minimality, $\sim$ has   precisely two classes.}  
	
	Suppose that for every witness $(S,\sim)$, item~(2) is not satisfied, i.e., there exists a $\cD$-invariant symmetric binary relation $R$ on $A^n$ which relates  two  disjoint $n$-tuples in the two different classes of $\sim$, and which does not have a pseudo-loop modulo $\gG_\cD$.
	We call any such triple $(S,\sim,R)$ an \emph{attempt} (at~(1)). An attempt $(S,\sim,R)$ is \emph{good} if the relation $R/\gG_\cD$ contains no cycle of odd length, and \emph{bad} otherwise. 
	Note that for any attempt $(S,\sim,R)$, the support of $R$, \todom{i.e., the set of all  $x\in A^n$ such that there exists $y\in A^n$ with $R(x,y)$,} 
	contains $S$:  otherwise considering the  restriction of $\sim$ to the intersection of $S$ with the support of $R$ would yield a naked set and  contradict the minimality of \todom{$(S,\sim)$.}

    We first claim that if  $(S,\sim,R)$ is a good attempt, then~(1) holds for $(S,\sim)$. To this end, it suffices to show in that situation there is no \todom{$R$-}path of even length connecting two elements of distinct  $\sim$-classes;   \todom{here, by an $R$-path of length $\ell$ we mean a sequence $(x_0,\ldots,x_\ell)$ such that $R(x_i,x_{i+1})$ for all $i\in\{0,\ldots,{\ell}\}$; we say that such $R$-path connects the elements $x_0$ and $x_\ell$, and if $\ell$ is even, then we call $x_{\frac{\ell}{2}}$ the midpoint of the path.}
    Indeed, if  this is true, then consider for every $m\geq 1$ the relation $R^{\circ 2m}(x,y)$, \todom{i.e., the $2m$-fold composition of $R$ with itself which relates to elements $x,y$  precisely if there exists an $R$-path of length $2m$ connecting $x$ and  $y$.}
    Any relation of this form \todom{is invariant under $\cD$ since $R$ is}; moreover, it is reflexive and symmetric since $2m$ is even and since $R$ is symmetric. \todom{Since $R^{\circ 2m}$ is contained in $R^{\circ 2(m+1)}$ for all $m\geq 1$, and since by the oligomorphicity of $\cD$ there is only a finite number of binary $\cD$-invariant relations, we have that $R^{\circ 2m}$ is an equivalence relation $\eta$ for $m$ large enough.} Clearly, the support of $\eta$ contains the support $S$ of $\sim$ by our remark of the preceding paragraph, and it  approximates $\sim$ since there is no path of even length between the two classes of $\sim$. 
    \todom{Finally, we claim that $\eta$ is \todo{presmooth}. To see this, let $a,b$ be disjoint  $n$-tuples related by $R$. Since $\cD$ has no algebraicity, there exists an element of $\gG_\cD$ which sends $a$ to a tuple $a'$ disjoint from $a$ while fixing $b$. This tuple $a'$ belongs to the same $\eta$-class as well as to the same $\gG_\cD$-orbit of $a$, and hence $\eta$ is indeed a \todo{presmooth} approximation of $\sim$.}   This proves~(1).
    
    \todom{So suppose for contradiction that $(S,\sim,R)$ is good and there exists an $R$-path of even length $2k$ connecting two elements of different classes of $\sim$.  Let $v$ be the midpoint of such a path.
    Consider the set $T$ of those elements in $S$ that are connected to \todom{some element of} the $\gG_\cD$-orbit of $v$ by an $R$-path of length $k$. Since $\cD$ is a model-complete core, the $\gG_\cD$-orbit of $v$ is invariant under $\cD$, and hence also under $\cC$. Moreover, since $R$ is invariant under $\cC$, so is $R^{\circ k}$, and it follows that $T$ is $\cC$-invariant. Since $T$ contains elements in both classes of $\sim$, it equals $S$, by the minimality of $(S,\sim)$. 
    Let $a,b\in S$ be so that $R(a,b)$ holds. Then the existence of an $R$-path of length $k$ from $a$ to some element in the orbit of $v$ together with the existence of an 
    $R$-path of length $k$  from $b$ to some element in the orbit of $v$ shows that $R/\gG_\cD$ contains a cycle of length $2k+1$, a contradiction to goodness.}

    \todom{It remains to show that there exists a good attempt. Striving for a contradiction, suppose that all attempts are bad.
    Then there exists a $\cD$-invariant symmetric  relation $Q\subseteq (A^n)^2$ such that $Q/\gG_\cD$ has an odd cycle but no loop: any relation $Q$ of an attempt $(S,\sim,Q)$ has this property.
    Pick such a relation $Q$  with minimal support; it exists since $\gG_\cD$ is oligomorphic and hence there is only a finite number of binary relations on $A^n$ which are invariant under it. Let $2\ell+1$ be the length of the shortest odd cycle of $Q/\gG_\cD$, and let $M'$ consist of those elements of the support of $Q/\gG_\cD$ 
    which belong to a cycle of length $2\ell+1$.
    Then $\cC^n/\gG_\cD$ acts on $M'$, and since $Q/\gG_\cD\cap (M')^2$ contains an odd cycle but no loop, there exists a $1$-ary minimal naked set  $(S',\sim')$ within $M'$, by our remarks preceding this theorem and in the first paragraph of this proof. Let $(S,\sim)$ be the corresponding witness as explained in the first paragraph, and pick $R$ so that $(S,\sim,R)$ is an attempt; by assumption, it is bad. Let $a,b$ be disjoint, in distinct $\sim$-classes,  and such that $R(a,b)$ holds. 
    Note that by the  minimality of the support of $Q$, we have that any element in this support is connected to some  element from the orbit of $a$ by a $Q$-path of length  $2\ell$:  otherwise the relation $P\subseteq Q$ defined by intersecting $Q$ with those pairs of elements connected by a $Q$-path of length $2\ell$ to some element in the orbit of $a$ is $\cD$-invariant, has a smaller support, and is such that $P/{\gG_\cD}$ has an odd cycle (since the orbit of $a$ is contained in $M'$)  but no loop. Since the $\gG_\cD$-orbit of $b$ belongs to $M'$, we have that $b$ belongs to the support of $Q$, and hence there exists a $Q$-path of length $2\ell$ connecting $b$ to some element from the orbit of $a$. 
    Let $v$ be the midpoint of such a path, let $X$ be the set of those elements connected to some element of  the orbit of $v$ by a $Q$-path of length $\ell$, and set $T:=R\cap X^2$.
    Then $(S,\sim, T)$ is still an attempt, and therefore bad, meaning that $T/{\gG_\cD}$ has an odd cycle. Since $R$ has no pseudo-loop, $T/{\gG_\cD}$ has no loop. But the support of $T$ is properly contained in that of $Q$, because otherwise the shortest odd cycle of $Q/\gG_\cD$ would be at most $\ell+1$. This contradicts the assumed  minimality of the support of $Q$.}
\end{proof}

The following  variant of the loop lemma deals with the \todom{easier} case when $\cC^n/{\gG_\cD}$ is equationally non-trivial.%

\begin{theorem}[The second loop lemma of  approximations]\label{thm:2-cases-general-2}
	Let $A$ be a set and  $n\geq 1$. Let $\cC\subseteq\cD$ be oligomorphic function clones on $A$, where $\cD$ is a model-complete core without algebraicity  and $\cC$ is $n$-canonical with respect to $\gG_\cD$.
	Suppose that  $\cC^n/{\gG_\cD}$ is equationally  non-trivial.
	\todom{Let $(S,\sim)$ be a minimal subfactor of  $\cC\actson A^n$ with $\gG_\cD$-invariant $\sim$-classes.} 
	Then one of the following holds:
\begin{enumerate}
    \item $\sim$ is  approximated by a $\cD$-invariant  equivalence relation which is \todo{presmooth} with respect to $\gG_\cD$;
    \item every $\cD$-invariant binary symmetric relation $R\subseteq (A^n)^2$ that  contains a pair $(a,b)\in S^2$ such that $a$ and $b$ are disjoint and such that   $a\not\sim b$ contains a pseudo-loop modulo $\gG_\cD$, i.e., a pair $(c,c')$ where $c,c'$ belong to the same $\gG_\cD$-orbit.
\end{enumerate}
\end{theorem}
\begin{proof}
    The proof is as in Theorem~\ref{thm:2-cases-general} above,
    except that no bad attempt can possibly exist given that $\cC^n/{\gG_\cD}$ is equationally non-trivial; \todom{cf.~our remark preceding Theorem~\ref{thm:2-cases-general}}.
    Therefore either there exists no $R$ such that $(S,\sim,R)$ is an attempt, and hence item~(2) holds, or there exists a good attempt, which yields~(1).
\end{proof}

\subsection{The fundamental theorem}\label{subsect:ft}
Our fundamental theorem allows us to lift an action of a function clone to a larger clone. It will be applied in situations where the first case of the loop lemma holds.

\begin{theorem}[The fundamental theorem of smooth approximations]
\label{thm:sa}
Let $A$ be a set. Let $\cC\subseteq \cD$ be function clones on $A$, and let $\gG$ be a permutation group on $A$ such that  $\cD$ locally interpolates $\cC$ modulo $\gG$.
\todom{Let $(S,\sim)$ be a subfactor of $\cC$ with ${\gG}$-invariant $\sim$-classes.}
\begin{itemize}
    \item If $\sim$ has a $\cD$-invariant very  smooth approximation  $\eta$ with respect to $\gG$, then there exists a clone homomorphism from $\cD$ to  $\cC\actson S/{\sim}$.
    \item If $\sim$ has a $\cD$-invariant  smooth approximation  $\eta$ with respect to $\gG$, then there exists a minion  homomorphism from $\cD$ to  $\cC\actson S/{\sim}$.
\end{itemize}
Moreover, if $\sim$ has finitely many classes, then the above homomorphism is uniformly continuous.
\end{theorem}
\begin{proof}
We first prove the statement for very smooth approximations. Our first claim then is  that whenever $f\in \cD$ locally interpolates $f'\in \cC$ modulo $\gG$  and $u\in S^k$, where $k$ is the arity of $f$, then $f( u)(\eta\circ{\sim})f'(u)$.
Indeed, there exist $\alpha\in \gG$ and a tuple $\beta\in\gG^k$  such that $f'( u)=\alpha\circ f(\beta(u))$ (where $\beta(u)$ is calculated componentwise).
Since $\eta$ is very smooth with respect to $\gG$, we have that $\beta(u)$ and $u$ are $\eta$-related in every component. Thus, with $f$ preserving $\eta$, we have that $f(u)$ and $f(\beta(u))$ are $\eta$-related. 
Since $\alpha\circ f( \beta(u))=f'(u)\in S$, and since the $\sim$-classes are invariant under $\gG$, we moreover get that $f'(u)$ and $f(\beta(u))$ are $\sim$-related, proving our claim.

Consequently, whenever $f',f''\in\cC$ are locally interpolated by $f\in\cD$, then they act in the same way on $S/{\sim}$. We can thus define a mapping $\xi$ which sends any function $f$ in  $\cD$ to the action on $S/{\sim}$ of any function $f'$ in $\cC$ which is locally interpolated by $f$. 

It remains to prove that $\xi$ is a clone homomorphism. Let $f\in\cD$ be of arity $k$, let $g$ be a $k$-tuple of functions in $\cD$ of equal arity $m$, and let $u\in S^m$. For every function of the tuple $g$, we pick a function in $\cC$ locally interpolated by it, and we collect these functions into a $k$-tuple $g'$.

By the above  there exists $v\in S^k$ such that the $k$-tuple $g(u)$ (calculated componentwise)  is $\eta$-related to $v$ in every component, and such that $v$ is ${\sim}$-related to $ g'(u)$ in every component. Then 
$$
f(g)(u)\;\eta\; f(v)\; (\eta\circ{\sim})\; f'(v)\sim f'(g')(u)\; ,
$$
and whence $\xi$ is a clone homomorphism.

If $\sim$ has only finitely many classes, then the action of any function in $\cD$ on a fixed set of representatives of these classes determines the value of the function under $\xi$, and hence $\xi$ is uniformly continuous. This completes the proof for very smooth approximations.

Let us consider the case where the approximation $\eta$ is smooth but not very smooth. Observe that the first claim above stating that $f(u)(\eta\circ{\sim})f'(u)$ still holds for those tuples $u$ whose components belong to $\gG$-orbits entirely contained in $\eta$-equivalence classes; \todom{in fact, with this additional assumption on $u$, no smoothness whatsoever is needed for the argument}. Adding  smoothness it follows, as before, that whenever $f',f''\in\cC$ are locally interpolated by $f\in\cD$, then they act in the same way on $S/{\sim}$. We can thus define the mapping $\xi$ in the same manner; $\xi$ is uniformly continuous if $\sim$ has finitely many equivalence classes. It  remains to prove that $\xi$ is a minion homomorphism. Let $f\in\cD$ by $k$-ary,  let $g$ be a $k$-tuple of projections in $\cD$ of equal arity $m$, and let $u\in S^m$ be a tuple whose components belong to $\gG$-orbits entirely contained in $\eta$-equivalence classes. \todom{Then the tuple $g(u)$ still has this property since all components of $g$ are projections, and hence 
$
f(g)(u)\; (\eta\circ{\sim})\; f'(g)(u)$ by our observation at the beginning of this paragraph. It follows that any function in $\cC$ locally interpolated by $f(g)$ modulo $\gG$ acts like $f'(g)$ on $S/{\sim}$, and hence $\xi(f(g))=\xi(f)(g)$.}  The proof is complete.
\end{proof}

\todom{The action $\cC\actson S/{\sim}$ can be viewed as a clone homomorphism $\phi$ which sends every function $f\in\cC$ to its action on $S/{\sim}$. In the general setting where $\cC\subseteq\cD$ and $\cD$ locally interpolates $\cC$, and where the action $\cC\actson S/{\sim}$ is idempotent, any  clone homomorphism $\xi$ from $\cD$ to $\cC\actson S/{\sim}$ extending $\phi$ is necessarily unique: it can only be defined as in the proof of Theorem~\ref{thm:sa}. 
Thus, (very) smooth approximations give a sufficient condition for the only possible extension $\xi$ to be well-defined.
In~\cite[Lemma 5.1]{UniqueInterpolation}, it is proved that in the special case where  $\cC$ consists of those operations in $\cD$ that are canonical with respect to the automorphism group of a homogeneous Ramsey structure, if the extension $\xi$ as in the proof of Theorem~\ref{thm:sa} is well-defined, then it is automatically a minion homomorphism, although a similar statement is not known to be true for clone homomorphisms.
A sufficient condition for $\xi$ to be well-defined is given in~\cite[Theorem 5.4]{UniqueInterpolation} in the case where $\cC\actson S/{\sim}$ consists solely of projections. We require, however, the full strength of   Theorem~\ref{thm:sa}, which does not not impose restrictions on $\cC$ or the action $\cC\actson S/{\sim}$.}

\subsection{Weakly commutative functions}\label{subsect:weird-symmetric}

The following result concerns the consequences of the second case of the loop lemma, and appears in all our P/NP dichotomy results.

\begin{lemma}[Weakly commutative functions]\label{lem:weird-symmetric}
Let $A$ be a set, $n\geq 1$, and let $\cD$ be an oligomorphic polymorphism  clone on $A$ that is a model-complete core.
Let $\sim$ be an equivalence relation on a set $S\subseteq A^n$ with $\gG_\cD$-invariant  classes.
Suppose that every $\cD$-invariant binary symmetric relation $R\subseteq (A^n)^2$ that  contains a pair $(a,b)\in S^2$ such that $a$ and $b$ are disjoint and such that $a\not\sim b$ contains a pseudo-loop modulo $\gG_\cD$.
Then $\cD$ contains a binary operation $f$ with the property that $f(a, b)\sim f( b, a)$ holds for all  $a,b\in A^n$  such that $f(a,b),f(b,a)$ are in $S$ and  disjoint.
\end{lemma}
\begin{proof}
    Let us call a pair $(a,b)$ of  elements of $A^n$  \emph{troublesome} if there exists a binary $h\in\cD$ such that $h(a,b),h(b,a)\in S$ are disjoint but $h(a,b)\not\sim h(b,a)$.
    Note that if $(a,b)$ is not troublesome and $g\in\cD$, then $(g(a,b),g(b,a))$ is not troublesome either.
    Moreover, if $a$ and $b$ are in the same orbit under $\gG_\cD$, then $(a,b)$ is not troublesome: since $\cD$ is a model-complete core, $h(a,b)$ and $h(b,a)$ are also in the same orbit for any $h\in\cD$; since the classes of $\sim$ are closed under $\gG_\cD$, it follows that $h(a,b)\sim h(b,a)$.
    
    Observe that if $u,v\in S$ are disjoint and  such that $u\not\sim v$, then there exist $g\in\cD$ and $\alpha\in\gG_\cD$ such that $\alpha\circ  g(u,v)=g(v,u)$. Indeed, these are obtained by application of our assumption to the relation  $R:=\{(g(u,v),g(v,u)) \mid g\in \cD\}$, since $R$ contains the pair $(u,v)$ by virtue of the first binary projection. 
    
    Therefore, if $(a,b)$ is troublesome, then setting $d(x,y):=g(h(x,y),h(y,x))$, where $g,h$ have the properties above, gives us 
    $$
    \alpha\circ   d(a,b)=\alpha\circ   g(h(a,b),h(b,a)) =  g(h(b,a),h(a,b))=d(b,a)\; .
    $$
    In conclusion, for every pair $(a,b)$ -- troublesome or not -- there exists a binary function $d$ in $\cD$ such that $(d(a,b),d(b,a))$ is not troublesome: if $(a,b)$ is troublesome one can take $d$ to be the operation we just described; if $(a,b)$ is not troublesome then the first projection works.

    Let $((a_i,b_i))_{i\in \omega}$ be an enumeration of all  pairs of tuples in  $A^n$.
    We build by induction on $i\in\omega$ an operation $f_i\in\cD$ such that  $(f_i(a_j,b_j),f_i(b_j,a_j))$ is not troublesome for any $j<i$.
    For $i=0$ there is nothing to show, so suppose that $f_i$ is built.
    Let $d\in\cD$ be an operation such that \todom{$(d(f_i(a_i,b_i),f_i(b_i,a_i)),d(f_i(b_i,a_i),f_i(a_i,b_i)) )$} is not troublesome and let $f_{i+1}(x,y):=d(f_i(x,y),f_i(y,x))$.
    By definition $(f_{i+1}(a_i,b_i),f_{i+1}(b_i,a_i))$ is not troublesome, and moreover $f_{i+1}$ also satisfies the desired property for $j<i$ by the remark in the first paragraph.
    
    By a standard compactness argument using the oligomorphicity of $\gG_\cD$ (essentially from~\cite{Topo-Birk}) and the fact that the polymorphism clone is topologically closed, we may assume that the sequence $(f_i)_{i\in\omega}$ converges to a function $f$. This function $f$ satisfies the claim of the proposition.
\end{proof}
Note that if $f$ satisfies the weak commutativity property of Lemma~\ref{lem:weird-symmetric}, then any operation in $$\overline{\{\beta\circ  f(\alpha,\alpha) \mid \alpha\in\gG_\cD\}}$$ also satisfies the same property, since the equivalence classes of $\sim$ are invariant under $\gG_\cD$. This implies that in applications, we will be able to assume that $f$ is diagonally canonical.

 \section{The Random Tournament}\label{sect:rt}
Let $\rel T=(T;\rightarrow)$ be the universal homogeneous tournament,
i.e., the unique (up to isomorphism)  homogeneous structure on a countable set $T$ with a single binary relation $\rightarrow$ such that for all  $x,y\in T$,
either $(x,y)$ is in $\rightarrow$ (which we henceforth denote by  $x\rightarrow y$) xor $y\rightarrow x$.

\subsection{The P/NP dichotomy} We are first going to prove the following.

\begin{theorem}\label{thm:mathdichotomy-tournament}
	Let $\rel A$ be a first-order reduct of $\rel T$ that is a model-complete core. 
	Then precisely one of the following holds:
	\begin{itemize}
		\item $\Pol(\rel A)$ has a uniformly continuous clone homomorphism to $\Projs$;
		\item $\Pol(\rel A)$ contains a ternary operation that is canonical with respect to $\sT$ and pseudo-cyclic modulo $\overline{\Aut(\rel T)}$.
	\end{itemize}
\end{theorem}

If $\sA$ is a CSP template, i.e., has a finite signature, then this yields a complexity dichotomy: in the first case $\csp(\sA)$ is NP-complete by~\cite{Topo-Birk}, and  
in the second case $\CSP(\sA)$ is in P by~\cite{Bodirsky-Mottet}. 
Moreover, the following corollary states that the two cases match with the two cases of Conjecture~\ref{conjecture:infinitecsp}, showing that Theorem~\ref{thm:dichotomies} holds  for the universal homogeneous tournament.
\begin{corollary}\label{cor:T:main}
Let $\rel A$ be a CSP template that is a  first-order reduct of $\rel T$. Then precisely one of the following holds:
	\begin{itemize}
		\item $\Pol(\rel A)$ has a uniformly continuous minion homomorphism to $\Projs$, and $\CSP(\sA)$ is NP-complete. 
		\item $\Pol(\rel A)$ has no uniformly continuous minion homomorphism to $\Projs$,  contains a ternary operation that is pseudo-cyclic modulo $\overline{\Aut(\rel T)}$, and $\CSP(\sA)$ is in P.
	\end{itemize}
\end{corollary}
\begin{proof}
    Let $\sA'$ be the model-complete core of $\sA$. By Lemma~\ref{lem:cores} we have that $\sA'$ is either a first-order reduct of $\sT$, or a one-element structure. If the latter is the case, then $\Pol(\sA)$ contains a constant function, and clearly the second statement holds. In the former case, Theorem~\ref{thm:mathdichotomy-tournament} applies to $\sA'$. In the first case of that theorem, $\Pol(\sA')$ has in particular a uniformly continuous minion homomorphism to $\Projs$, and hence so does $\Pol(\sA)$ by~\cite{wonderland}. This situation moreover implies NP-hardness of $\CSP(\sA)$ by~\cite{wonderland}. In the second case of the theorem, $\Pol(\sA)$ does not have a uniformly continuous minion homomorphism to $\Projs$ by~\cite{BKOPP, BKOPP-equations}, and $\Pol(\sA)$ has, as $\Pol(\sA')$,  a ternary pseudo-cyclic operation modulo $\overline{\Aut(\rel T)}$ (see, for example, the proof of Corollary~6.2 in~\cite{Topo}). Moreover, $\CSP(\sA')=\CSP(\sA)$ is in P by~\cite{Bodirsky-Mottet}.
\end{proof}
In order to prove these results, we first show Lemma~\ref{lem:cores} which we needed to derive  Corollary~\ref{cor:T:main} from Theorem~\ref{thm:mathdichotomy-tournament}; this will be done 
in Section~\ref{subsect:rt:cores} using a recent result from~\cite{MottetPinskerCores}.

We then prove Theorem~\ref{thm:mathdichotomy-tournament} as follows. First, we note that if $\rel A$ is a first-order reduct of $(T;=)$, or equivalently, if $\Aut(\sT)$ is the full symmetric group, then the result follows easily (either $\Pol(\sA)$ only contains essentially unary functions, and the first case of Theorem~\ref{thm:mathdichotomy-tournament} holds, or it contains a binary injection, and the second case holds -- see e.g.~\cite{BodChenPinsker}). We therefore assume that $\rel A$ is not a first-order reduct of $(T;=)$. It then follows that the binary disequality  relation $\neq$  is preserved by all functions in $\Pol(\sA)$ by a general observation (Proposition~\ref{prop:definition-neq}). This implies that $\Pol(\sA)$ acts on the set of injective $n$-tuples for all $n\geq 1$.

We  then consider various subclones of $\Pol(\sA)$  and use the following notation:
\begin{itemize}
    \item $\CATT$ is the clone of those functions in $\CA$ which are canonical with respect to $\sT$;
    \item $\CAT$ is the clone of those functions in $\CA$ which preserve the equivalence of $\Aut(\sT)$-orbits of injective tuples;
    \item $\CAA$ is the clone of those functions in $\CA$ which preserve the equivalence of $\Aut(\sA)$-orbits of injective tuples.
\end{itemize}
\todom{Another way of describing $\CAT$ is that it contains precisely those polymorphisms of $\sA$ which act  on the two $\Aut(\sT)$-orbits of injective pairs; abusing notation, we denote these orbits by $\rightarrow$ and $\leftarrow$ and the action of $\CAT$ on them by  $\CAT\actson\{\leftarrow,\rightarrow\}$.} 

Note that $\CATT\subseteq \CAT$ since the functions in $\CATT$ have to preserve orbit-equivalence with respect to $\Aut(\sT)$ of arbitrary tuples, whereas for the functions $\CAT$ this condition is imposed only on injective tuples. %
There are no obvious inclusions in general between $\CAA$ and either of  $\CAT$ and $\CATT$. However, we will be interested in comparing these clones in the situation where $\CATT$ is equationally trivial, since otherwise the \todo{second} case of Theorem~\ref{thm:mathdichotomy-tournament}  follows easily (see Proposition~\ref{prop:CATT-nontrivial}). In that case, by general elementary arguments (Section~\ref{subsect:injective}) plus the use of a Ramsey expansion of $\sT$ (Section~\ref{subsect:injective}) as well as results about clones on finite sets,  
also \todo{$\CAT\actson\lra$ is equationally trivial (achieved in Proposition~\ref{prop:CATT-nontrivial}), and in particular the action $\CAT\actson\lra$ consists of essentially unary operations}\todom{~\cite{Post}}. 
The latter implies $\CAT\subseteq \CAA$ (Lemma~\ref{lem:CATinCAA}). Hence, we will have $\CATT\subseteq \CAT\subseteq \CAA$ when we want to apply the theory of smooth approximations. 

In that situation, it will turn out that $\CAA$ is equationally trivial as well  by Lemma~\ref{lem:canonicalSiggers}, which uses the above-mentioned fact  that $\CAT\actson\lra$ is by essentially unary functions.

We then apply the loop lemma of smooth approximations (Theorem~\ref{thm:2-cases-general}) to $(\CAA\actson U) \subseteq(\CA\actson U)$, where $U$ is the set of injective $k$-tuples for some $k\geq 1$ large enough so that $\CAA\actson U/\Aut(\sA)$ is equationally trivial; such $k\geq 1$ exists by Corollary~\ref{cor:CAA_on_orbits_trivial_tournament}, which  follows from  a general observation (Proposition~\ref{prop:CAA_on_orbits_trivial}). The hypotheses of Theorem~\ref{thm:2-cases-general} are met (with $A:=U$, $n:=1$), since  $\sT$ and hence also $\CA\actson U$ have no algebraicity, since $\sA$ and hence also  $\CA\actson U$ are  model-complete cores, and since the functions of $\CAA$ act on the orbits of injective $n$-tuples with respect to $\gG_{\CA}=\Aut(\sA)$, by definition. By the theorem, we are in one of two cases. 

In the first case, there is a naked set $(S,\sim)$ for $\CAA\actson U$ that \todo{has a presmooth $(\CA\actson U)$-invariant approximation}. This leads, via the fundamental theorem of smooth approximations (Theorem~\ref{thm:sa}), to $\CA$ having a uniformly continuous clone  homomorphism to $\P$, putting us into the first case of Theorem~\ref{thm:mathdichotomy-tournament}. To apply Theorem~\ref{thm:sa} with function clones $(\CAA\actson U)\subseteq (\CA\actson U)$ and permutation group $\Aut(\sT)\actson U$, we have to  show that $\CA\actson U$ locally interpolates $\CAA\actson U$ modulo $\Aut(\sT)\actson U$; this we do by showing that $\CA$ locally interpolates $\CAT$ (Corollary~\ref{cor:canonization-without-order}) and using $\CAT\subseteq\CAA$ (Lemma~\ref{lem:CATinCAA}). We also need to observe that $\Aut(\sT)$ leaves the $\sim$-classes invariant, which is clear since they are unions of $\Aut(\sA)$-orbits by definition. Finally,  Lemma~\ref{lem:primitive-neq-very-smooth} \todo{applies since $\neq$ is invariant under $\Pol(\rel A)$ and the universal homogeneous tournament is $n$-``primitive'' for all $n\geq 1$ (Example~\ref{ex:primitive-tournament}), so that }the approximation is very smooth.

The second case of the loop lemma leads to a contradiction, since  using the weakly \todo{commutative} function from Lemma~\ref{lem:weird-symmetric} we can produce a function in $\CATT$ satisfying a non-trivial identity; this is achieved in Lemma~\ref{lem:maj-from-symmetric} together with Proposition~\ref{prop:CATT-nontrivial}. 

The proof of Theorem~\ref{thm:mathdichotomy-tournament} is complete.

\subsubsection{Model-compete cores} \label{subsect:rt:cores}

The tournament $\rel T$ can be expanded by a linear order $<$ so that $(\rel T,<)$ is Ramsey (this follows, e.g., from  \cite{Hubicka-Nesetril-All-Those}). In fact, the expansion $(\rel T,<)$ is the universal homogeneous linearly ordered tournament;  \todo{its automorphism group has five orbits of pairs, namely precisely  the orbits of the pairs  $(a,a),(a,b),(b,a),(b,c),(c,b)$ where $a,b,c\in T$ are arbitrary elements satisfying $a<b$, $b<c$,  $a\rightarrow b$, and $b\leftarrow c$}. We will need this expansion several times; the first instance is the use of the following recent theorem in order to compute the model-complete cores of first-order reducts of $\sT$.

\begin{theorem}[\cite{MottetPinskerCores}]\label{thm:ramsey-cores}
    Let $\rel A$ be a first-order reduct of a  homogeneous Ramsey structure $\rel B$, and let $\rel A'$ be its model-complete core. Then $\rel A'$ is a first-order reduct of a homogeneous Ramsey substructure $\rel B'$ of $\rel B$.
    
    Moreover, there exists $g\in\End(\sA)$ which is \emph{range-rigid} with respect to $\Aut(\sB)$, i.e., all orbits of tuples with respect to $\Aut(\sB)$ which intersect the range of $g$ are invariant under $g$, and such that the  age of $\rel B'$ is equal to the age of structure induced by the range of $g$ in $\sB$.
\end{theorem}

Thus, we can apply Theorem~\ref{thm:ramsey-cores} and obtain that in order to prove Theorem~\ref{thm:mathdichotomy-tournament}, we can henceforth assume we are dealing with model-complete cores.

\begin{lemma}\label{lem:cores}
	Let $\rel A$ be a first-order reduct of $\rel T$. Then the model-complete core of $\rel A$ is either a one-element structure, or is again a first-order reduct of\/ $\rel T$.
	\todo{Moreover, if $\rel A$ is a model-complete core that is a first-order reduct of $\rel T$ and not of $(T;=)$, then the range of every endomorphism of $\rel A$ intersects every orbit of $\Aut(\rel T,<)$.}
\end{lemma}
\begin{proof}
By Theorem~\ref{thm:ramsey-cores}, the model-complete core $\rel A'$ is a first-order reduct of $\sB'$, where $\sB'$ is a homogeneous Ramsey substructure of $(\sT,<)$, and there exists $g\in\End(\rel A)$ which is  range-rigid with respect to $\Aut(\sT,<)$ and whose range induces exactly the age of $\sB'$. 

If the range of $g$ is a single point, then $\rel B'$ is a one-element structure. Otherwise, the range of $g$ contains an edge of $\sT$. If all the edges within the range of $g$ are increasing with respect to $<$,  or if all such edges are decreasing, then by  range-rigidity we have that  $\rel B'$ is isomorphic to $(\mathbb Q;<,<)$ or $(\mathbb Q;>,<)$. Moreover, in the first case  $g$ flips an edge if and only if it is decreasing, and in the second case if and only if it is increasing.
	It follows that for any two injective tuples $a,b$ of the same length within the domain of $\sA'$, there exist   $\alpha\in \Aut(\rel T)$
	and an embedding $e$ from the range of $g$ into $\rel B'$ such that $b=e\circ g\circ \alpha (a)$. Since $e\circ g\circ \alpha$ is an endomorphism of $\rel A'$, by back-and-forth, $\Aut(\rel A')$ is the full symmetric group of $A'$. Hence $\rel A'$ is a first-order reduct of $(A';=)$.

	Finally, if the range of $g$ contains at least one increasing and at least one decreasing edge, then  by the  range-rigidity of $g$ we have that $\rel B'$ is isomorphic to  $(\rel T,<)$, and  $\rel A$ and $\rel A'$ are isomorphic.
	
	\todom{Now suppose that $\rel A$ is a model-complete core that is a reduct of $\rel T$ but not of $(T;=)$, and let $e\in\End(\rel A)$. Striving for a contradiction, suppose that the range of $e$ does not intersect every orbit of $\Aut(\rel T,<)$. By the results in~\cite{MottetPinskerCores} (more precisely Lemma~15 and Lemma~11 applied to the monoid $\End(\rel A)$), the range-rigid function   $g\in\End(\rel A)$ whose existence is stipulated in Theorem~\ref{thm:ramsey-cores} then has the same property. Applying the reasoning of the second and third paragraph of the present proof, we conclude that either $\rel A$ is a first-order reduct of $(T;=)$, or the range of $g$ induces in $(\rel T,<)$ a structure isomorphic to $(\rel T,<)$,  both of which yield a  contradiction.}
\end{proof}

Note that in the first case of Lemma~\ref{lem:cores}, $\End(\sA)$ contains a constant function (by Theorem~\ref{thm:ramsey-cores}), and hence the \todo{second} case of Theorem~\ref{thm:mathdichotomy-tournament} applies since constant functions are canonical and pseudo-cyclic. Thus indeed, it remains to prove Theorem~\ref{thm:mathdichotomy-tournament} for model-complete cores.  

\begin{proposition}\label{prop:definition-neq}
Let $\rel B$ be a transitive homogeneous structure whose age has no minimal bound of size $3$. Let $\rel A$ be a first-order reduct of $\rel B$ which is a model-complete core.  Then the binary disequality relation $\neq$ is pp-definable in $\rel A$.
\end{proposition}
\begin{proof}
    If $\neq$ is an orbit of $\rel A$, then it is pp-definable in $\rel A$ since $\rel A$ is a model-complete core. Otherwise, there exist two distinct orbits $O_1,O_2$ of injective  pairs; these are again pp-definable in $\rel A$.
    The formula $\phi(x,y):=\exists z (O_1(x,z)\wedge O_2(y,z))$ then clearly pp-defines 
    \todo{a subset of $\neq$; we claim that it actually defines $\neq$.  
    To see that $\phi(a,b)$ holds for all elements $a,b$ of $\rel A$ with $a\neq b$, pick $c,d$ such that
    $(a,c)\in O_1,(b,d)\in O_2$. Such elements exist by the transitivity of $\rel B$.
    Consider the structure $\rel X$ on a three-element domain $\{x,y,z\}$ in the signature of $\rel B$ such that the structure induced by $\{x,y\}$ is isomorphic to the structure induced by $\{a,b\}$ in $\rel B$, such that the structure induced by $\{x,z\}$ is isomorphic to the structure induced by $\{a,c\}$ in $\rel B$, and such that the structure induced by $\{y,z\}$ is isomorphic to the structure induced by $\{b,d\}$ in $\rel B$.
    Since all two-element substructures of $\rel X$ embed into $\rel B$, and since $\rel B$ has no minimal bound of size $3$, it follows that $\rel X$ embeds into $\rel B$.
    By homogeneity of $\rel B$, one can assume that this embedding maps $x$ to $a$ and $y$ to $b$. Finally, the image of $z$ under this embedding proves that $(a,b)$ satisfies $\phi$.}
\end{proof}

\subsubsection{Injective binary polymorphisms}\label{subsect:injective}

In this section we derive some general results on the existence of  binary injective functions in function clones. For the universal homogeneous tournament, these results will be the foundation for considering actions of function clones on injective, rather than arbitrary, tuples; and in particular, to work with $\CAT$ and $\CAA$ rather than with $\CATT$ and, say, the clone $\cC_{\sA}^{\sA}$ of all functions in $\CA$ which are canonical with respect to $\sA$.

For an arbitrary relational structure $\rel A$, it is known that if $\Pol(\rel A)$ has no clone homomorphism to $\Projs$, then $\Pol(\rel A)$ contains a ternary essential function~\cite{Topo}.
We show in Proposition~\ref{prop:arity-essential} that under some conditions, one can lower the arity down to two.

\begin{definition}
Let $\sB$ be a structure. An $\Aut(\sB)$-orbit $O$ of pairs is \emph{free}
if for all elements $a,b$ of $\sB$, there exists an element $c$ of $\sB$  such that $(c,a),(c,b)\in O$.
\end{definition}
Note that if a structure has a free orbit of pairs, then it is transitive. We now give a few examples of structures with a free orbit.

\begin{example}{\ }

\begin{itemize}
\item The universal homogeneous partial order  has a free orbit, namely the orbit of incomparable elements.
\item If $\rel B$ is a homogeneous structure with free amalgamation which is transitive,  then one can obtain  a free orbit by free amalgamation of two vertices. In particular, for the universal homogeneous $K_n$-free graph (where $n\geq 3$ is arbitrary) as well as for the random graph, the orbit of non-adjacent pairs is free. \todom{On the other hand, the orbit of  adjacent pairs is not free in the universal homogeneous $K_3$-free graph.}
\item If $\rel B$ is homogeneous and transitive, and if the set
of minimal bounds for the age of $\sB$ does not contain any structure
of size three, then any orbit of injective pairs is free. This applies, in particular, to the universal homogeneous tournament.
\end{itemize}
\end{example}

We now provide some examples of transitive structures without a free orbit.

\begin{example}
If $\rel B$ is a homogeneous transitive  structure, let $\rel B'$ be the structure obtained by taking two copies of $\rel B$ and relating them with a new binary relation. Then $\rel B'$ has no free orbit. The homogeneous equivalence relation with two infinite classes is of this form.
\end{example}

\begin{proposition}\label{prop:arity-essential}
	Let $\rel A$ be a first-order reduct of a homogeneous structure $\rel B$ such that $\sB$ has a free orbit.
	If $\Pol(\rel A)$ contains an essential function, then it  contains a binary essential operation.
\end{proposition}
\begin{proof}
	\todo{The proof is essentially the same as in~\cite{RandomMinOps}, where it is done for reducts of the random graph.
	The generalisation to our weaker assumption seems to be folklore (see, e.g., \cite[Lemma 6.1.29]{Book}).}
	We give a proof here for the convenience of the reader. Let $f\in \Pol(\rel A)$ be essential, and let $k$ be the arity of $f$; without loss of generality, $k\geq 3$. 
	We show that $f$ can be taken to be binary. 
	Let $O$ be a free orbit of pairs of $\Aut(\rel B)$. 
	Since $f$ depends on its first argument, there exist $a_1,a'_1,a_2,\dots,a_k$ such that $f(a_1,\dots,a_k)\neq f(a'_1,a_2,\dots,a_k)$.
	We distinguish two cases.
	
 Suppose that there are $b_1,\dots,b_k$ such that $(b_i, a_i)\in O$ for all $i\in\{2,\dots,k\}$ and $f(b_1,a_2,\dots,a_k)\neq f(b_1,\dots,b_k)$.
	For $i\in\{2,\dots,k\}$, let $\alpha_i$ be an automorphism of $\rel B$ be such that $\alpha(a_2)=a_i$ and $\alpha(b_2)=b_i$.
	Let $g(x,y):= f(x,y,\alpha_3(y),\dots,\alpha_k(y))$.
	Then $g(a_1,a_2)=f(a_1,a_2,\dots,a_k) \neq f(a'_1,a_2,\dots,a_k)=g(a'_1,a_2)$ and $g(b_1,a_2)=f(b_1,a_2,\dots,a_k)\neq f(b_1,\dots,b_k) = g(b_1,b_2)$,
	so that $g$ depends on both its arguments.
	
     Suppose now that for all $b_1,\dots,b_k$ such that $(b_i, a_i)\in O$ for all $i\in\{2,\dots,k\}$, we have $f(b_1,a_2,\dots,a_k)=f(b_1,\dots,b_k)$. 
	Since $f$ depends on its second coordinate, there exist $c_1,\dots,c_k$ such that $f(c_1,c_2,\dots,c_k)\neq f(c_1,a_2,\dots,a_k)$.
	By assumption on $O$, there exist $d_2,\dots,d_k$ such that $(d_i, a_i)\in O$ and $(d_i, c_i)\in O$ for all $i$.
	For all $i\in\{2,\dots,k\}$, let $\alpha_i$ be an automorphism of $\rel B$ such that $\alpha_i(c_2)=c_i$ and $\alpha_i(d_2)=d_i$,
	and let $g(x,y):=f(x,y,\alpha_3(y),\dots,\alpha_k(y))$.
	Note that by assumption,
	\[g(a_1,d_2)=f(a_1,d_2,\dots,d_k)=f(a_1,a_2,\dots,a_k)\neq f(a'_1,a_2,\dots,a_k)=f(a_1',d_2,\dots,d_k)=g(a_1',d_2),\] so that $g$ depends on its first argument.
	Moreover, 
	\[g(c_1,c_2)=f(c_1,\dots,c_k)\neq f(c_1,a_2,\dots,a_k) = f(c_1,d_2,\dots,d_k) = g(c_1,d_2).\]
	Therefore, $g$ depends on its second argument too.
\end{proof}

We now move on from essential operations to injections, and require the following definitions. 

\begin{definition}
Let $\gG$ be a transitive permutation group. The \emph{canonical binary structure} of $\mathscr G$ is the structure on its domain that has
a binary relation for each orbit of \todo{injective} pairs under $\mathscr G$.
\end{definition}

The following definition is a variant of the concept of finite boundedness for homomorphisms rather than embeddings.

\begin{definition}
A structure $\rel B$ has \emph{finite duality} if there exists a finite set $\mathcal F$ of finite structures in its signature, such that  for every finite structure $\rel X$ in its signature it is true that $\rel X$ has a homomorphism to $\sB$ if and only if no member of $\mathcal F$ has a homomorphism to $\rel X$. The set $\mathcal F$ is called an  \emph{obstruction set} for $\sB$.
\end{definition}

\begin{proposition}\label{prop:existence-injective}
	Let $\sA$ be a first-order reduct of a transitive $\omega$-categorical structure $\sB$ such that the canonical binary structure of $\Aut(\sB)$ has finite duality. If $\Pol(\sA)$ contains a binary essential function preserving  $\neq$, then it contains a binary  injective function.
\end{proposition}
\begin{proof}
	The proof is essentially the same as the one done for the random graph~\cite{BodPin-Schaefer-both} or the universal homogeneous $K_n$-free graphs~\cite{BMPP16}; 
	we reproduce it here under our  general assumptions for the convenience of the reader.
	Let $f\in\Pol(\rel A)$ be binary, essential, and preserving $\neq$.
	If for all $a,a',b,b'$ we have $f(a,b)\neq f(a',b)$ and $f(a,b)\neq f(a',b')$ then $f$ is injective, as it also preserves $\neq$.
	Thus, we can assume that there exist $a\neq a',b$ such that $f(a,b)=f(a',b)$ (otherwise, replace $f$ by $(x,y)\mapsto f(y,x)$).
	Since $f$ preserves $\neq$, the point $(a,b)$ is \emph{v-good}, namely we have $f(a,b)\neq f(a,b')$ for all $b'\neq b$.
	
	We prove that for all tuples $(x,y,z,z),(p,p,q,r)\in B^4$ \todo{with $x\neq y$ and $q\neq r$},
	there is an operation $g\in\Pol(\sA)$ such that $g(x,p)\neq g(y,p)$ and $g(z,q)\neq g(z,r)$.
	Since $f$ preserves $\neq$, it then follows by a standard compactness argument that some binary operation in $\Pol(\sA)$ is injective.
	
	Let $(z_0,q_0)$ be an arbitrary $v$-good point.
	Then for every $x_0,y_0,p_0$ such that $(x_0,y_0,z_0)$ is in the same orbit as $(x,y,z)$ and $(p_0,q_0)$ is in the same $\Aut(\sB)$-orbit as $(p,q)$,
	either $f(x_0,p_0)\neq f(y_0,p_0)$ and we are done (by taking $\alpha,\beta\in \Aut(\sB)$ such that $\alpha(x,y,z)=(x_0,y_0,z_0)$ and $\beta(p,q)=(p_0,q_0)$ and setting $g:=f(\alpha,\beta)$),
	or $f(x_0,p_0)=f(y_0,p_0)$, which implies as above that $(x_0,p_0)$ is $v$-good.
	Thus, either we find $g$ or $v$-goodness transfers from $(z_0,q_0)$ to any point $(x_0,p_0)$ such that $(x,z)$ and $(x_0,z_0)$ are in the same orbit, and $(p,q)$ and $(p_0,q_0)$ are in the same orbit.
	
	Let $O$ be the orbit of $(x,z)$ and $P$ be the orbit of $(p,q)$.
	We claim that for every $(x_0,p_0)$, there exists a path $(z_1,q_1),\dots$ such that $(z_i,z_{i+1})\in O$ and $(q_i,q_{i+1})\in P$ and $z_n=x_0,q_n=p_0$.
	Indeed, let $n$ be larger than the size of the obstruction set witnessing finite duality for the canonical structure of $\Aut(\sB)$.
	Consider the finite structure $\sX$ containing elements $z_0,q_0,\dots,z_{n-1},q_{n-1},z_n=x_0,q_n=p_0$,
	containing the relations $(z_i,z_{i+1})\in O, (q_i,q_{i+1})\in P$ for all $i$, and where additionally $(x_0,z_0)$ is labelled by the orbit of $(x_0,z_0)$ in $\Aut(\sB)$, and similarly $(p_0,q_0)$ is labelled by its orbit in $\Aut(\sB)$.
	Then $\sX$ does not contain any homomorphic image of an element from the obstruction set, so $\sX$ homomorphically maps to the canonical binary structure of $\Aut(\sB)$.
	Moreover, by construction one can assume that $x_0,p_0,z_0,q_0$ are mapped to themselves.
	The remaining elements witness the existence of the elements from the claim.
	Thus, every point $(x_0,p_0)$ is $v$-good.
	
	Finally, there exist $x_1,y_1,p_1$ such that $f(x_1,p_1)\neq f(y_1,p_1)$ and $(x_1,y_1)$ is in the same orbit as $(x,y)$.
	The argument is the same as above: if this were not the case, equality would propagate along paths labelled by the orbit of $(x,y)$.
	By finite duality, every pair of points is reachable by such a path that is long enough, and thus we would have $f(x_1,p_1)=f(y_1,p_1)$ for \emph{all} $x_1,y_1,p_1$, so that $f$ would be essentially unary, a contradiction to the assumption.
	Let $\alpha\in\Aut(\sB)$ be such that $\alpha(x,y)=(x_1,y_1)$ and $\beta\in\Aut(\sB)$ be such that $\beta(p)=p_1$, which exists by transitivity of $\Aut(\sB)$.
	Then $f(\alpha x,\beta p)\neq f(\alpha y,\beta p)$, and by $v$-goodness $f(\alpha z, \beta q)\neq f(\alpha z,\beta r)$.
	Thus, $g=f(\alpha,\beta)$ witnesses our claim.
\end{proof}

 Propositions~\ref{prop:arity-essential} and~\ref{prop:existence-injective} imply the following.
 
\begin{corollary}\label{cor:binaryinjections}
    Let $\sA$ be a first-order reduct of the random graph, the universal homogeneous $K_n$-free graph for some $n\geq 3$, or of the universal  homogeneous tournament. If $\Pol(\sA)$ contains an essential operation which preserves $\neq$, then it contains a binary injection.
\end{corollary}

Let us remark that the rational order $(\mathbb Q;<)$ does not satisfy the assumption of Proposition~\ref{prop:existence-injective}: any obstruction set for $(\mathbb Q;<)$ has to contain all directed cycles. In fact, there are first-order reducts of $(\mathbb Q;<)$ which have an essential binary polymorphism  preserving $\neq$ but no injective binary polymorphism.

\subsubsection{Binary injections for\/ $\sT$} \label{subsect:injectiveT}
We now refine the results from the preceding section for $\sT$: using the Ramsey expansion of $\sT$, we can obtain more knowledge about the binary injective polymorphisms of its first-order reducts.

\begin{lemma}\label{lem:tournament-binary-injective}
	Let $\rel A$ be a first-order reduct of $\rel T$ that is a model-complete core. If $\Pol(\rel A)$ contains an essential operation, then it contains a binary injective operation.
\end{lemma}
\begin{proof}
	By Proposition~\ref{prop:definition-neq}, every polymorphism of $\rel A$ preserves $\neq$. Corollary~\ref{cor:binaryinjections} thus implies that $\Pol(\rel A)$ contains a binary injective polymorphism. 
\end{proof}

\begin{remark}\todo{The structure $(\rel T,<)$ is universal for the class of at most countable linearly ordered tournaments. This implies that for any binary injective operation $f$, there exists an endomorphism $e$ of $\rel T$ such that $f':=e\circ f$ satisfies $f'(x,y)<(x',y')$ whenever $x<x'$ or $x=x',y<y'$.
Indeed, the tournament on $T^2$ defined by pulling back the relations from the image of $f\colon T^2\to T$ can be endowed with the lexicographic order, and this ordered tournament embeds into $(\rel T,<)$ by universality. This embedding yields $e$, by homogeneity of $\rel T$.}\end{remark}

\begin{lemma}\label{lem:generateProjection}
	Let $\rel A$ be a first-order reduct of $\rel T$ that is a model-complete core and not a first-order reduct of $(T;
	=)$. Suppose that $\Pol(\rel A)$ contains a binary injection. Then $\Pol(\rel A)$ contains a binary injection which  acts like a projection on $\{\leftarrow,\rightarrow\}$.
\end{lemma}
\begin{proof}
Since $(\rel T,<)$ is a homogeneous Ramsey structure, and since any function that is locally interpolated by an injective function is also injective, we have
 that $\Pol(\rel A)$ contains a binary injection $f$ which is canonical with respect to $(\rel T,<)$.  Let $\alpha$ be an automorphism of $\rel T$ reversing the order $<$.
 \todo{By the remark preceding the lemma, up to composing $f$ with an endomorphism one may assume that $x<x'$ implies $f(x,y)<f(x',y')$ for all $x,x',y,y'\in T$.}

The behaviour of $f$ on injective pairs is completely determined by its behaviour on pairs  which are both increasing with respect to $<$ and its behaviour on pairs where the first type is increasing, and the second is decreasing. We call the former \emph{straight inputs}, and the latter \emph{twisted inputs}. Since $\rel A$ is a model-complete core and not a reduct of $(T;=)$, we may assume that $f$ acts in an  idempotent fashion on the types of $(\rel T,<)$; in particular, it behaves like a projection or like a semilattice operation on straight inputs. Replacing $f$ by $f(f(x,y),f(y,x))$, we retain the above properties, and the new function will also act in an idempotent manner on twisted inputs. We can thus assume that on any fixed kind of input, $f$ behaves either like a projection or like a semilattice operation.

\todom{
Suppose that $f$ behaves like the first  projection on straight inputs. Then $f(x,f(x,y))$ behaves like the first projection on all inputs, and hence on $\{\leftarrow,\rightarrow\}$, and we are done. Similarly, if $f$ behaves like the second projection on straight inputs, then $f(f(x,y),x)$ acts like the first projection on $\{\leftarrow,\rightarrow\}$.
}

\todom{
If $f$ behaves like a projection on twisted inputs, then $f(x,\alpha(y))$, for $\alpha$ as above,  behaves like a projection on straight inputs, a case we have already treated.
} \todom{It thus remains to consider the case where $f$ behaves like a semilattice operation on straight inputs as well as on twisted inputs.
}

\todom{
If $f$ behaves like one semilattice operation on straight inputs, and like the opposite semilattice operation on twisted inputs, then $f(x,f(x,y))$  behaves like a projection on twisted inputs, a case we have already treated.
}

\todom{Finally, assume that $f$ behaves like one and the same  semilattice operation on all inputs. Then $g(x,y):=f(\alpha(x),\alpha(y))$, for $\alpha$ as above, behaves like the opposite semilattice operation on all inputs, and $f(x,g(x,y))$ behaves like the first projection on all inputs.
}
\end{proof}

\begin{corollary}\label{cor:g012}
    Let $\rel A$ be a first-order reduct of\/ $\rel T$ that is a model-complete core.
	If\/ $\Pol(\rel A)$ contains an essential function, then $\Pol(\rel A)$ contains \todo{one of the operations $g_0,g_1,g_2$ whose behaviours are described in Figure~\ref{fig:canonical}}.
\end{corollary}
\begin{proof}
	By Lemma~\ref{lem:tournament-binary-injective}, there exists a binary injective operation $f$ in $\Pol(\rel A)$, and by Lemma~\ref{lem:generateProjection}, we can assume it acts like a projection on $\lra$. We can also assume that $f$ is canonical with respect to $(\rel T,<)$, 
	\todom{hence it acts on the set $\{=, (\rightarrow\cap <),(\rightarrow\cap >),(\leftarrow\cap <),(\leftarrow\cap >)\}$ of orbits of  pairs under the action of $\Aut(\rel T,<)$. Let $\alpha$ be an automorphism of $\rel T$ reversing the order $<$.
	Up to permuting arguments, $f$ can be assumed to behave like the first projection on $\lra$. By the remark above, up to composing $f$ with an endomorphism of $\rel T$ we may assume that $f(x,y)<f(x',y')$ if and only if $x<x'$ or $x=x', y<y'$, for all $x,x',y,y'\in T$.}
	
	\todo{
	Call $f$ \emph{$x$-arrow-dominated} if $f((\rightarrow\cap <),=)\neq f((\leftarrow\cap <),=)$, and \emph{$x$-order-dominated} if $f((\rightarrow\cap <),=)=f((\leftarrow\cap <),=)$.
	We define $y$-arrow-dominated and $y$-order-dominated similarly.
	}
	
	\todo{If $f$ is $x$-order-dominated and $f((\rightarrow\cap<),=)=(\leftarrow\cap <)$,
	then $f'(x,y):=\alpha f(\alpha(x),y)$ is $x$-order-dominated and $f'((\rightarrow\cap<),=)=(\rightarrow\cap <)$.
	If $f$ is $x$-arrow-dominated and $f((\rightarrow\cap<),=)=(\leftarrow\cap<)$, then $f'(x,y):=f(f(x,y),y)$ is $x$-arrow-dominated and $f'((\rightarrow\cap<),=)=(\rightarrow\cap <)$.
	Doing similar transformations in the second argument if need be, one can now assume that $f((\rightarrow\cap<),=)=(\rightarrow\cap<)=f(=,(\rightarrow\cap <))$.}
	
	\todo{If $f$ is $x$-order-/$y$-arrow-dominated, then $f'(x,y):=f(x,f(y,x))$ is $x$-arrow-/$y$-order-dominated, still behaves like the first projection on $\lra$, and $f'((\rightarrow\cap <),=)=f'(=,(\rightarrow\cap <))=(\rightarrow\cap <)$.}
	
	\todo{Thus, we are left with the three cases of $f$ being $x$-order-/$y$-arrow-dominated, $x$-arrow-/$y$-order-dominated, and $x$-arrow-/$y$-arrow-dominated, whose behaviours are the ones described in Figure~\ref{fig:canonical}.}
\end{proof}

\begin{figure}
\begin{center}
\behavior{\rightarrow}{\rightarrow}{\rightarrow}{\rightarrow}{\leftarrow}{\leftarrow}
\behavior{\rightarrow}{\rightarrow}{\rightarrow}{\leftarrow}{\leftarrow}{\leftarrow}
\behaviorunordered{\rightarrow}{\leftarrow}{\rightarrow}{\leftarrow}
\end{center}
\caption{\todo{The behaviours of the three canonical operations $g_0, g_1,g_2$}. Note that $g_2$ is canonical with respect to $\rel T$. All three functions behave like $p_1$ on $\{\rightarrow,\leftarrow\}$.}
\label{fig:canonical}
\end{figure}

\subsubsection{The case when $\CATT$ is non-trivial}

The proof of the following lemma is in part due to L.~Barto.
\begin{lemma}\label{lem:generating-g2}
    Let $\sA$ be a first-order reduct of $\sT$ that is a model-complete core. Assume that $\CAT\actson\{\leftarrow,\rightarrow\}$ is equationally  non-trivial. Then $\Pol(\sA)$  contains $g_2$.
\end{lemma}
\begin{proof} By Corollary~\ref{cor:g012}, we have that $\Pol(\sA)$ contains one \todo{of the operations $g_0,g_1,g_2$ described} in Figure~\ref{fig:canonical}; we may assume that it contains $g\in\{g_0,g_1\}$, for otherwise we are done. 
    By Post's classification of the function clones on a two-element domain~\cite{Post}, the assumption of  $\CAT\actson\{\leftarrow,\rightarrow\}$ satisfying non-trivial identities implies that there exists a ternary $f\in \CAT$ which acts like a majority or a minority operation on $\{\leftarrow,\rightarrow\}$, the action of a binary function $f$ of $\CAT$ as a semilattice operation on $\lra$ being impossible since $f(\leftarrow,\rightarrow)={}\leftarrow$ implies $f(\rightarrow,\leftarrow)={}\rightarrow$.

    We claim that for every finite tournament $\rel B$,
    there exist finite linear orders $\rel O_1,\dots,\rel O_n$ with the same domain as $\rel B$ and an $n$-ary polymorphism $f_n$ of $\rel A$ such that $f_n(\rel O_1,\dots,\rel O_n)=\rel B$, and $n$ only depends on the size of $\rel B$.
    \todo{We first explain how to obtain $g_2$ from this, and shall come back to the proof of the claim in the next paragraph.
    Let $\rel S$ be a finite substructure of $\rel T$.
    From the claim, we obtain linear orders $\rel O_1,\dots,\rel O_n$ on $S$, which give automorphisms $\alpha_1,\dots,\alpha_n$ of $\rel T$ such that the structure induced by $\alpha_i(S)$ in $(\rel T,<)$ is ordered like $\mathbb O_i$, for all $i\in\{1,\ldots,n\}$.
    Then $(x,y)\mapsto f_n(g(\alpha_1 x,\alpha_1y),\dots,g(\alpha_nx,\alpha_ny))$ is an operation that behaves like $g_2$ on $S$.
    By a standard compactness argument using the oligomorphicity of $\Aut(\sT,<)$, we obtain that $\Pol(\rel A)$ contains $g_2$.}

    We now prove the claim.
    Assume first that $f$ behaves like majority.
    By Post's classification, $f$ generates the $n$-ary majority operation $f_n$ for all $n\geq1$.
    The proof goes by induction on the size of $\rel B$, the case $|\rel B|\leq 2$ being trivial.
    For every $b\in\rel B$, one gets linear orders $(\rel O^b_i)_{1\leq i\leq n}$ encoding $\rel B\setminus\{b\}$. Define $\rel P^b_i$ to be $\rel O^b_i$ where one adds $b$ as a minimal element.
    We claim that for every arc $(b,c)$ in $\rel B$, a majority of all these linear orders have $b<c$. Note that all the linear orders $\rel P^b_i$ are correct, all the linear orders $\rel P^c_i$ are wrong, so they cancel each other out. Now by the induction hypothesis, for every $d\not\in\{b,c\}$ a majority of the linear orders $\rel P^d_i$ are so that $b<c$, so that overall we get the desired result.
    
    Now assume that $f$ behaves like minority.
    Again, it is clear that $f$ generates for all odd $n\geq 1$ an operation $f_n$ that behaves like $x_1+\dots+x_n\bmod 2$ on $\{\leftarrow,\rightarrow\}$.
    We follow the same strategy as above, encoding $\rel B$ by an odd number $n$ of linear orders so that $f_n(\dots)=\rel B$.
    Suppose that $|\rel B|$ is odd. We do the same as before, except that for all $b\in \rel B$ we apply the induction hypothesis to the tournament obtained by reversing all the arcs in $B\setminus\{b\}$.
    Now one sees that every arc $(b,c)$ is such that $b<c$ in an odd number of the linear orders $(\rel P^d_i)_{d,i}$.
    Suppose now that $|\rel B|$ is even. We fix a special vertex $b\in B$, and proceed as above by induction with the tournaments obtained by reversing all the arcs of $\rel B\setminus\{b,c\}$ for all $c\in B\setminus\{b\}$.
    For every such $c$, we have linear orders $\rel O^c_i$.
    We let $\rel P^c_i$ be obtained by adding $b,c$ as minimal elements, and where the order between $b$ and $c$ is given by $\rel B$.
    It is clear that the collection of the linear orders $(\rel P^c_i)_{c,i}$ satisfies the required property.
\end{proof}

\begin{proposition}\label{prop:CATT-nontrivial}
Suppose that $\sA$ is a first-order reduct of $\sT$ that is a model-complete core. The following are equivalent.
\begin{itemize}
\item[(1)] $\CATT$ contains a ternary operation that is pseudo-cyclic modulo $\overline{\Aut(\sT)}$; 
    \item[(2)] $\CATT$ is equationally non-trivial;
    \item[(3)] $\CAT$ is equationally non-trivial;
    \item[(4)] $\CAT\actson\lra$ is equationally non-trivial.
\end{itemize}
\end{proposition}
\begin{proof}
The implications from~(1) to~(2), from~(2) to~(3), and from~(3) to~(4) are trivial. 

We show the implication from~(4) to~(1). By Lemma~\ref{lem:generating-g2},  $\Pol(\sA)$ contains $g_2$. Moreover, the clone of idempotent functions of $\CAT\actson\lra$ satisfying non-trivial identities by~\cite{Post}, we have that $\CAT$ contains a ternary function $f$ which acts as a cyclic operation on $\lra$, by~\cite{Cyclic}. Setting $h(x,y,z):=g_2(x,g_2(y,z))$, we then have that the function $f(h(x,y,z),h(y,z,x),h(z,x,y))$  is an element of $\CATT$ which is pseudo-cyclic modulo $\overline{\Aut(\sT)}$, proving~(1).
\end{proof}

\subsubsection{From $\CAT$ to $\CAA$}

\begin{lemma}\label{lem:CATinCAA}
   Let $\sA$ be a first-order reduct of\/ $\sT$ that is a model-complete core.
   If\/ $\CAT$ is equationally trivial, then $\CAT\subseteq\CAA$. 
\end{lemma}
\begin{proof}
    The action $\CAT\actson \lra$ is by essentially unary functions by Post's classification~\cite{Post}. 
    Since $\sT$ is homogeneous in a binary language, the action $\CAT\actson\lra$  determines the action of $\CAT$ on  $\Aut(\sT)$-orbits of injective tuples of arbitrary fixed  length. 
    Thus, for every $f\in\CAT$ of arity $k\geq 1$, there exists $1\leq i\leq k$ such that either the $\Aut(\sT)$-orbit of $f(a_1,\dots,a_k)$ is equal to the $\Aut(\sT)$-orbit of $a_i$ for all  injective tuples $a_1,\dots,a_k$ of the same length, or the $\Aut(\sT)$-orbit of $f(a_1,\dots,a_k)$ is equal to the $\Aut(\sT)$-orbit of $a_i'$ for all  injective tuples $a_1,\dots,a_k$ of the same length, where $a_i'$ is any tuple of that length whose edges point in the opposite direction of the edges of $a_i$. In the first case, $f\in\CAA$ is obvious since $\Aut(\sT)$-orbits are a refinement of $\Aut(\sA)$-orbits.  In the second case,  $\Aut(\sA)$ contains a function which flips the direction of all arrows, and the same argument works.
\end{proof}

\begin{lemma}\label{lem:canonization-without-order}
    Let $\rel A$ be a first-order reduct of $\rel T$ that is a model-complete core and not a first-order reduct of $(T;=)$.
    Suppose that $\Pol(\rel A)$ contains 
    \begin{itemize}
        \item a binary injection which  acts like a projection on $\lra$, and 
        \item a function that is canonical with respect to
    $(\rel T,<)$ but is not an element of $\CAT$.
    \end{itemize}
    Then $\CAT\actson\lra$ contains a   majority operation.
\end{lemma}
\begin{proof}
Let $f\in\CA$ be a function that is  canonical with respect to $(\rel T,<)$ but not an element of $\CAT$, and denote its arity by $n$. 

Let $S:=\{(\rightarrow\cap <),(\rightarrow\cap >),(\leftarrow\cap <),(\leftarrow\cap >)\}$ be the set of orbits of injective pairs under the action of $\Aut(\rel T,<)$. 
The map $\hat f$  defined by $x\mapsto f(x,\dots,x)$ acts  surjectively on $S$, as otherwise $\rel A$ would be a \todom{first-order} reduct of $(T;=)$ \todo{by Lemma~\ref{lem:cores}}.  It follows that it acts bijectively thereon, \todom{and that $(\hat f)^{12}$ acts as the identity on $S$ since $|S|=4$.}  Hence replacing $f$ by \todom{$(\hat f)^{11}\circ f$}  we can assume that $\hat f$ induces the identity on $S$, i.e., $f$ acts idempotently on $S$.

We claim that there exist $o_1,\dots,o_n\in\{(\rightarrow\cap<),(\leftarrow\cap <)\}$ and $q_1,\dots,q_n\in S$ such that $o_i$ and $q_i$ belong to distinct  $\Aut(\sT)$-orbits for all $1\leq i\leq n$ and such that $f(o_1,\ldots,o_n)$ and $f(q_1,\ldots,q_n)$ belong to the same $\Aut(\sT)$-orbit. Indeed, otherwise we would have that for any pairs $p_1,\ldots,p_n\in{} \neq$, the $\Aut(\sT)$-orbit of $f(p_1,\ldots,p_n)$ would be determined by the value of $f(o_1,\ldots,o_n)$, where $o_i$ would be the unique element of $\{(\rightarrow\cap<),(\leftarrow\cap <)\}$ which is disjoint from the $\Aut(\sT)$-orbit of $p_i$, for all $1\leq i\leq n$. Hence, $f$ would act on the $\Aut(\sT)$-orbits in $\neq$, and would therefore be an element of $\CAT$, a contradiction.

Assume without loss of generality that  $f(o_1,\dots,o_n)\subseteq{} \rightarrow$. We prove by induction on $m\geq 2$ that for all tuples $ a, b, c\in(\neq)^m$ such that for every $i\in\{1,\dots,m\}$ at most one of the pairs $a_i,b_i,c_i$ is in $\leftarrow$ 
there exists $h\in\Pol(\rel A)$ such that $h(a,b,c)\in(\rightarrow)^m$. A standard compactness argument using $\omega$-categoricity then implies that the same holds for arbitrary infinite triples $a,b,c$: this is achieved by taking a limit in $\CA$ of functions of the form $\alpha_m\circ h_m$, where $\alpha_m\in\Aut(\sT)$ and $h_m$ witnesses the above property for the first  $m$ entries of $a,b,c$. 
It then follows, by taking $a,b,c$ to enumerate all possible triples, that there exists $h\in\CA$ which acts on $\lra$ as a majority operation;  in particular  $h\in\CAT$.

Before we start the induction, we make a few observations and assumptions about the tuples $a,b,c\in (\neq)^m$. First, we can clearly assume that there are no repetitions, i.e.,  $(a_i,b_i,c_i)\neq (a_j,b_j,c_j)$  for any two distinct $i,j\in\{1,\ldots,m\}$; moreover, $(a_i,b_i,c_i)\neq (\overline{a_j},\overline{ b_j},\overline{ c_j})$ where the pairs $\overline{ a_j},\overline{b_j}, \overline{c_j}$ are obtained from $a_j,b_j,c_j$ by flipping the two coordinates,  since otherwise $(a_i,b_i,c_i)$ would contain two elements of $\leftarrow$. 
For a tuple $ x\in(\neq)^m$, let $\ker(x)$ be the kernel of $x$ when we view it as a $2m$-tuple over $T$, i.e., $\ker(x)$ is an equivalence relation on $\{1,\ldots,2m\}$. 
By assumption, there is an injective binary operation $g$ in $\Pol(\rel A)$ that acts as a projection on $\lra$; without loss of of generality, it acts as  the first projection.
This implies that we can assume that $\ker(a)=\ker(b)=\ker(c)=:K$.  Indeed,  the tuples $g(a,g(b,c)), g(b,g(a,c)), g(c,g(a,b))$ satisfy the same assumptions as $a,b,c$ since $g$ acts as a projection on the  first argument on $\lra$, 
and they all have the same kernel since $g$ is injective.
By our assumptions above,  $i\neq j$ yields  $\{(2i+1,2j+2),(2i+2,2j+1)\}\not\subseteq K$ for all $1\leq i,j\leq m-1$. Writing $i/K$ for the  equivalence class of any   $i\in\{1,\ldots,2m\}$, it follows that for all any two linear orders on $\{1/K,2/K\}$ and $\{3/K,4/K\}$, there is a linear order on $\{1/K,\ldots,2m/K\}$ extending both of them.

We next claim that there exist linear orders $\prec,\prec'$ on $\{1/K,\ldots,(2m)/K\}$ which agree everywhere except between $1/K,2/K$, where they disagree. Suppose the contrary, and let $j$ be the smallest number with $3\leq j\leq m$ such that the claim does not hold for $\{1/K,\ldots,(2j)/K\}$.  Let $\prec,\prec'$ be orders on $\{1/K,\ldots,(2(j-1))/K\}$ with the above property. If  $(2j-1)/K$ or $(2j)/K$ was not an element of $\{1/K,\ldots,(2(j-1))/K\}$, then we could extend $\prec,\prec'$, contradicting the minimality of $j$. Otherwise, the only obstacle would occur if $\prec, \prec'$ disagreed on $\{(2j-1)/K,(2j)/K\}$. This, however, would imply $\{(2j-1)/K,(2j)/K\}=\{1/K,2/K\}$, a contradiction. Note that by flipping $\prec$ and $\prec'$, we may assume that they are increasing on $\{3/K,4/K\}$.

We are now ready for the induction. The base case $m=2$ is trivial as $h$ can be taken to be a projection. In the induction step, let $m\geq 3$. By induction hypothesis, there exist ternary $h, h'\in\Pol(\rel A)$ such that all coordinates of $h(a,b,c)$ are an element of $\rightarrow$ except possibly for the second, and all coordinates of $h'(a,b,c)$ are an element of $\rightarrow$ except possibly for the first. 
For every $k\in\{1,\dots,n\}$, let $<_{o_k}$ be the increasing order on $\{3/K,4/K\}$ and $<_{q_k}$ be the order on $\{1/K,2/K\}$ given by $q_k$. By the above, there exists a linear order $<_k$ on $\{1/K,\ldots,2m/K\}$ which extends these two, and moreover we may assume that all these orders agree except possibly on $\{1/K,2/K\}$.  Hence,  for all $k\in\{1,\dots,n\}$  there exists  $\alpha_k\in\Aut(\rel T)$ and $u^k\in\{a,b\}$ such that  $\alpha_k(u^k_1,u^k_2)$ is in $q_k\times o_k$ and such that the order $<$ on   $\alpha_k(u^k_j)$ is the same for all $k\in\{1,\dots,n\}$. 
Then, $f(\alpha_1 u^1,\dots,\alpha_nu^n)$ is an element of $(\rightarrow)^m$, concluding the proof.
\end{proof}

Every function in $\CA$ locally interpolates, modulo $\Aut(\sT)$, a function which is canonical with respect to $\Aut(\rel T,<)$, by the Ramsey property of $(\rel T,<)$. If $\CAT$ is equationally  trivial, then  it is also locally interpolated by $\CA$ modulo $\Aut(\sT)$. 

\begin{corollary}\label{cor:canonization-without-order}
    Let $\rel A$ be a first-order reduct of $\rel T$ that is a model-complete core and not a first-order reduct of $(T;=)$.
    If\/ $\CAT$ is equationally trivial, then $\CA$ 
    locally interpolates  $\CAT$ modulo $\Aut(\rel T)$.
\end{corollary}
\begin{proof}
    Every 
    $f$ locally  interpolates a function $g$ that is canonical with respect to $\Aut(\rel T,<)$. 
    Suppose that such a $g$ is not in $\CAT$. Then $g$ must be essential, for otherwise $\sA$ would be a first-order reduct of $(T;=)$ \todo{by Lemma~\ref{lem:cores}}. 
    Thus, by Corollary~
    \ref{cor:g012} and Lemma~\ref{lem:canonization-without-order}, the clone $\CAT\actson\lra$ is equationally non-trivial, contradicting Proposition~\ref{prop:CATT-nontrivial}.
\end{proof}

\begin{lemma}\label{lem:canonicalSiggers}
Let $\rel A$ be a first-order reduct of $\rel T$ that is a model-complete core and not a first-order reduct of $(T;=)$. 
If $\CAT$ is equationally trivial, then so is  $\CAA$.
\end{lemma}
\begin{proof}
Assuming that $\CAA$ is equationally non-trivial, we show that the same holds for  $\CAT\actson\lra$; this is sufficient by Proposition~\ref{prop:CATT-nontrivial}. %

\todom{We first claim that $\CAA$ then contains an operation $s$ which acts like a pseudo-Siggers operation modulo $\Aut(\sA)$ when restricted to injective tuples: that is, for all injective tuples $a,b,c$ of elements of $A$ of the same length,  $s(a,b,a,c,b,c)$ and $s(b,a,c,a,c,b)$ are in the same orbit with respect to $\Aut(\sA)$. To see this, note that the action of  $\CAA$ on the $\Aut(\sA)$-orbits of injective $k$-tuples is equationally non-trivial for all $k\geq 1$. Since each of these actions is idempotent as a result of $\CAA$  being a model-complete core,  for each $k\geq 1$ there exists a 6-ary   $s_k\in\CAA$ which is a Siggers operation in this action, by~\cite{Siggers}. This property being invariant under composition with elements from $\Aut(\sA)$ from the left, a standard compactness argument gives that we can assume that the functions $s_k$ have an accumulation point $s\in\CAA$. In the standard action of  $\CAA$, the function $s$ then behaves like a pseudo-Siggers operation on injective tuples.}

By Corollary~\ref{cor:canonization-without-order}, there exists  $s'\in\CAT$ locally interpolated by $s$ modulo $\Aut(\sT)$. We claim that $s'$ is still a pseudo-Siggers operation modulo $\Aut(\sA)$ when restricted to injective tuples.
Indeed, let $a,b,c$ be given, and let $S$ be the finite set containing the entries of all these tuples.
There are $\alpha_1,\dots,\alpha_6,\beta\in\Aut(\rel T)$ such that 
$$
s'(x_1,\dots,x_6)=\beta s(\alpha_1 x_1,\dots,\alpha_6 x_6)
$$ on $S$. 
Moreover, since $s\in\CAA$ and since $a,b,c$ are injective,  there exist $\gamma,\gamma'\in\Aut(\rel A)$ such that
\[ \gamma s(a,b,a,c,b,c) = s(\alpha_1 a,\alpha_2 b,\alpha_3 a,\alpha_4 c,\alpha_5 b,\alpha_6 c)\]
and
\[ \gamma' s(b,a,c,a,c,b) = s(\alpha_1 b,\alpha_2 a,\alpha_3 c,\alpha_4 a,\alpha_5 c,\alpha_6 b).\]
Finally, let $\delta\in\Aut(\rel A)$ witness the pseudo-Siggers identity for $s$ for $a,b,c$. Then we have
\begin{align*}
    s'(a,b,a,c,b,c) &= \beta s(\alpha_1 a,\alpha_2 b,\alpha_3 a,\alpha_4 c,\alpha_5 b,\alpha_6 c)\\
    &= \beta \gamma s(a,b,a,c,b,c)\\
    &= \beta\gamma\delta s(b,a,c,a,c,b)\\
    &= \beta\gamma\delta(\gamma')^{-1}\beta^{-1}\beta\gamma s(b,a,c,a,c,b)\\
    &= (\beta\gamma\delta(\gamma')^{-1}\beta^{-1})\beta s(\alpha_1 b,\alpha_2 a,\alpha_3 c,\alpha_4 a,\alpha_5 c,\alpha_6 b)\\
    &= (\beta\gamma\delta(\gamma')^{-1}\beta^{-1})s'(b,a,c,a,c,b),
\end{align*}
so that $\beta\gamma\delta(\gamma')^{-1}\beta^{-1}\in\Aut(\rel A)$ proves our claim.

Suppose that $\CAT\actson\lra$ is equationally trivial. 
Then by Post's classification~\cite{Post} all the operations in $\CAT$ act on $\lra$ as essentially unary maps, and in particular $s'$ does. We may even assume that $s'$ acts as a projection on $\lra$ by composing it with an automorphism of $\sA$, since $\sA$ is a model-complete core.
Without loss of generality, let us assume that $s'$ is the projection onto the first coordinate when  acting on $\lra$.
Let $a,b$ be any two injective tuples  of elements of $A$ of the same length that are not in the same orbit under $\Aut(\sA)$; such tuples exist since $\sA$ is not a first-order reduct of $(T;=)$. By the above,  we get that $s'(a,b,a,a,b,a)$ and $s'(b,a,a,a,a,b)$ are in the same orbit of  $\Aut(\rel A)$. Moreover, since $s'$ is the first projection  when acting on $\lra$, %
$a$ and $s'(a,b,a,a,b,a)$ are in the same orbit under $\Aut(\rel T)$, and $b$ and $s'(b,a,a,a,a,b)$ are in the same orbit under $\Aut(\rel T)$.
This is a contradiction, and  therefore $\CAT\actson\lra$ is equationally non-trivial.
\end{proof}

\subsubsection{From $\CAA$ to $\CA$}\label{subsect:CAA2CA}

It is known~\cite[Theorem~6.7]{BPP-projective-homomorphisms} that when a polymorphism clone consists of canonical functions with respect to an automorphism group of a finitely bounded homogeneous structure, and its action on the orbits of the group is equationally non-trivial, then the original polymorphism clone is also equationally non-trivial. The following is a variant for this when we know that all  actions on orbits of injective tuples are equationally  non-trivial, and we know that the clone contains a binary injection.

\begin{proposition}\label{prop:CAA_on_orbits_trivial}
Let $\rel A$ be an $\omega$-categorical structure that is a model-complete core. Assume moreover that $\CA$ contains a binary injection and that the  functions in $\CA$ preserve  $\neq$.  
If $\CAA$ is equationally trivial, then there exists $k\geq 1$ such that the action of $\CAA$ on the $\Aut(\sA)$-orbits of injective  $k$-tuples is equationally trivial.
\end{proposition}
\begin{proof}
We show the contraposition. Suppose that the action of  $\CAA$ on the $\Aut(\sA)$-orbits of injective $k$-tuples is equationally non-trivial for all $k\geq 1$. \todom{Then, by the compactness argument given in the beginning of Lemma~\ref{lem:canonicalSiggers},  $\CAA$ contains a  function $s$ which  behaves like a pseudo-Siggers operation on injective tuples: that is, whenever $a,b,c$ are injective tuples of the same fixed finite length, then $s(a,b,a,c,b,c)$ and $s(b,a,c,a,c,b)$ belong to the same $\Aut(\sA)$-orbit.}

Since $\CA$ contains a binary injection, it  also contains a ternary injection $g$. Defining ternary terms  $u:=g(x,y,z)$, $v:=g(y,z,x)$, $w:=g(z,x,y)$, we then have that the functions  $s(u,v,u,w,v,w)(x,y,z)$ and $s(v,u,w,u,w,v)(x,y,z)$ yield tuples in the same $\Aut(\sA)$-orbit whenever they are applied to arbitrary, possibly  non-injective, tuples $a,b,c$ of finite fixed length. Hence, again by a standard 
compactness argument, there exist unary functions $e,f\in\overline{\Aut(\sA)}$ such that 
$$
e\circ s(u,v,u,w,v,w)(x,y,z)= f\circ s(v,u,w,u,w,v)(x,y,z)
$$
holds 
in $\CAA$ for all values of $x,y,z$ in $T$. This identity cannot, however, be satisfied by projections, so $\CAA$ is equationally  non-trivial.
\end{proof}

\begin{corollary}\label{cor:CAA_on_orbits_trivial_tournament}
Let $\sA$ be a first-order reduct of\/ $\sT$ that is a model-complete core and not a first-order reduct of $(T;=)$. If $\CAA$ is equationally trivial, then there exists $k\geq 1$ such that the action of $\CAA$ on the $\Aut(\sA)$-orbits of injective $k$-tuples is equationally trivial.
\end{corollary}
\begin{proof}
    The functions of $\CA$ preserve $\neq$ by Proposition~\ref{prop:definition-neq}.
    The statement is trivially true  if $\CA$ contains only essentially unary functions, since there exists $k\geq 1$ such that $\Aut(\sA)$ has at least two orbits in its action on injective $k$-tuples. If $\CA$ contains an essential function, then it contains a binary injection by Corollary~\ref{cor:binaryinjections}, and we can refer to Proposition~\ref{prop:CAA_on_orbits_trivial}.
\end{proof}

\begin{lemma}\label{lem:maj-from-symmetric}
Let $\rel A$ be a first-order reduct of\/ $\rel T$ which is a model-complete core and not a first-order reduct of $(T;=)$. Suppose that $\CA$ contains a binary function $f$ such that there exist $k\geq 1$, $S\subseteq T^k$ consisting of injective tuples, and $\sim$ an equivalence relation on $S$ with two $\Aut(\sA)$-invariant classes such that $f(a,b)\sim f(b,a)$ for all disjoint injective tuples $a,b\in T^k$ with  $f(a,b),f(b,a)\in S$. Then $\CAT\actson\lra$  contains a majority operation.
\end{lemma}

\begin{proof}
Since $\sA$ is a model-complete core, we have $f(a,a)\in S$ for all $a\in S$, and hence picking any $b\in S$ which is not $\sim$-equivalent to $a$ we see on the  values $f(a,a),f(b,b),f(a,b)$, and $f(b,a)$ that $f$ must be essential. Moreover, since its distinctive property is stable under diagonal  interpolation modulo $\Aut(\sT)$, we can assume that $f$ is diagonally canonical with respect to $\Aut(\rel T,<)$. 

Let $O:=(\rightarrow\cap <)$, $O':=(\leftarrow\cap >)$. Then $f$ is canonical on $O$ and on $O'$ with respect to $\Aut(\rel T,<)$, i.e., whenever $a,b\in T^2$ are so that $(a_i,b_j)\in O$ for all $1\leq i,j\leq 2$, then the orbit of $f(a,b)$ with respect to $\Aut(\sT,<)$ only depends on the orbits of $a$ and of $b$ (and similarly for $O'$). We may assume it has the same property for injective tuples $a,b$ and with respect to $\Aut(\sT)$ on each of these two sets $O$ and $O'$, since otherwise we are done by Lemma~\ref{lem:canonization-without-order}: the first function in the hypotheses of that lemma is obtained from the fact that $f$ must be essential and by application of Lemmas~\ref{lem:tournament-binary-injective} and~\ref{lem:generateProjection}, and the second function by considering $f(e_1(x),e_2(y))$ for self-embeddings $e_1,e_2$ of $\sT$ such that $(e_1(x),e_2(y))\in O$ (or $O'$ in the other case) for all $x,y\in T$.

This means that on $O$ as well as on $O'$, the function $f$ acts on $\lra$. Note that the only possible induced action  on $\lra$ of any binary function defined on $T$ is essentially unary. This means that the two  actions of the restriction of $f$ to $O$ and $O'$ on $\lra$ are either by a projection, or by a projection composed with the only non-trivial permutation on $\lra$. If one of the two actions is not a projection, then $\overline{\Aut(\sA)}$ contains a function $e$ which acts on $\lra$ as the non-trivial permutation: this function is obtained as $f(e_1(x),e_2(x))$ using self-embeddings $e_1,e_2$ as above. Hence, in that case the $\Aut(\sA)$-orbits are invariant under flipping of all edges on a tuple.

It follows that whenever $a,b\in S$ are such that $(a_i,b_j)\in O$ for all $i,j\in\{1,\ldots,k\}$, then $f(a,b)$ and $f(b,a)$ are elements of $S$. Furthermore, the two  actions of the restrictions of $f$ to $O$ and $O'$ on $\lra$ cannot depend on the same argument, otherwise by picking $a$ and $b$ in different blocks of $\sim$ and so that $(a_i,b_j)\in O$ for all $i,j\in\{1,\ldots,k\}$,
we would obtain that $f(a,b)$ and $f(b,a)$ are in different blocks of $\sim$, a contradiction. Without loss of generality, $f$ depends on the first argument on $O$, and on the second argument on $O'$.

Let $e_1,e_2$ be self-embeddings of $(\rel T,<)$ such that $e_1(x)<e_2(y)$ if and only if $x<y$ for all $x,y\in T$, and such that 
the order relation $<$ and the relation $\rightarrow$ of $\rel T$ coincide between the  ranges of $e_1,e_2$. Then $f(e_1(x),e_2(y))$ is still diagonally canonical with respect to $\Aut(\rel T,<)$, and still depends on its first argument in its action on $\lra$ when restricted to $O$, and on its second argument when restricted to $O'$.   Replacing $f$ by $f(e_1(x),e_2(y))$, we then moreover have that for all $a,b\in{}\neq$, the relations of $\rel T$ which hold between the components of $a$ and those of $b$ do not influence the $\Aut(\sT)$-orbit of $f(a,b)$ (while the order relations do).

We now consider complete diagonal order types $B$  for two pairs in $\neq$, i.e., each $B$ determines for two pairs in $\neq$ all order relations that hold on and between them. We then have that within each complete diagonal order type, $f$ acts on $\lra$, by our above assumption. 
Replace $f$ by $f(e_1\circ f(x,y),e_2\circ f(x,y))$, where $e_1,e_2$ are self-embeddings of $(\rel T,<)$ which ensure that for all $a,b\in{}\neq$, the diagonal order type of $(a,b)$ is equal to that of $(e_1\circ f(a,b),e_2\circ f(a,b))$. We may then assume that $f$ acts  idempotently or as a constant function on the diagonal in its action on  $\lra$ within  each complete diagonal order type $B$. In particular, $f$ then acts like the first projection on $\lra$ when restricted to $O$, and like the second when restricted to $O'$. 

Suppose that within some complete diagonal order type $B$, we have that $f$ acts as a constant function on the diagonal in its action on $\lra$; that is,  $f(\rightarrow,\rightarrow)=f(\leftarrow,\leftarrow)$ within $B$. It then is  crystal clear that $\CA$ contains all permutations on $T$, in contradiction with our assumption that $\rel A$ is not a \todo{first}-order reduct of $(T;=)$. To make this even crystal clearer, we show   that in that case, any finite injective tuple of elements of $T$ can be mapped, by a unary function in $\Pol(\sA)$,  to a finite injective tuple all of whose edges  point forward within the tuple. We use induction on the length of the tuple which we call $a$.  Without loss of generality, $f(\rightarrow,\rightarrow)=f(\leftarrow,\leftarrow)={}\rightarrow$ within $B$. If the length of $a$ is one, then picking  $\alpha,
\beta\in\Aut(\sT)$ such that $(\alpha(a),\beta(a))$ is in $B$ we get that $f(\alpha(x),\beta(x))$  has the desired property for $a$. For the induction step,  suppose that the length of $a$ is $n+1$, for some $n\geq 1$; by induction hypothesis, we may assume that all edges in the restriction of $a$ to its first $n$ components point forward. Without loss of generality, $(u,v)\in B$ implies $u_1<v_1$ for all $u,v\in{}\neq$. 
Let $\alpha,\beta\in\Aut(\sT)$ be such that setting  $a':=\alpha(a), a'':=\beta(a)$ we have  $((a_i',a_{n+1}'),(a_i'',a_{n+1}''))\in B$ for all $i\in\{1,\ldots,n\}$, and $(a_i',a_j'')\in O$ for all  $i,j\in\{1,\ldots,n\}$. This is possible since the first of the two conditions only imposes the order relation  $a_i'<a_i''$ for all $i\in\{1,\ldots,n\}$. Now  $f(\alpha(a),\beta(a))$ is a tuple all of whose edges point forward, bringing the induction to a successful conclusion.  

Summarizing, within every fixed complete diagonal order type for two pairs in $\neq$, we have that $f$ acts idempotently on $\lra$. Hence,  within each such type,
$f$ as a semilattice operation or as a projection on $\lra$. 

Similarly as in the proof of Lemma~\ref{lem:canonization-without-order}, we now show  by induction on $m\geq 2$ that for all tuples $a, b, c\in(\neq)^m$ such that for every $i\in\{1,\dots,m\}$ at most one of the pairs $a_i,b_i,c_i$ is in $\leftarrow$ 
there exists $h\in\Pol(\rel A)$ such that $h( a, b, c)\in(\rightarrow)^m$. A standard compactness argument then implies that $\CAT$ contains a function such that $\CAT\actson\lra$ is a majority operation.

The base case $m=2$ is clearly achieved by applying an appropriate projection. For the induction step, let  $a,b,c\in (\neq)^{m}$ for some $m\geq 3$. Since $f$ must be essential, it generates a binary injection which acts as a projection on $\{\leftarrow,\rightarrow\}$, and hence, as in the proof of Lemma~\ref{lem:canonization-without-order},  we may assume that the kernels of $a,b,c$ are identical. By induction hypothesis, we may moreover assume that all components of $a$ are in $\rightarrow$ except for the second, and all components of $b$ are in $\rightarrow$ except for the first.

If no equalities hold between $a_1$ and $a_2$, then there exists $\beta\in\Aut(\rel T)$ such that all components of $a_{1}$ are smaller than those of $\beta(b_{1})$
with respect to $<$, and  all components of $a_{2}$ are larger  than those of $\beta(b_{2})$. In that case, $f(a,\beta(b))$ yields the desired tuple.

Otherwise, assume that the second coordinate of $a_1$ equals the first coordinate of $a_2$; the other cases are treated similarly. 
Consider any complete diagonal order type $B$ for pairs in $\neq$ which is \emph{crossing}, i.e., it requires for a pair $(u,v)$ of pairs in $\neq$  some strict order relation between $u_1,v_1$ and the opposite strict order relation between $u_2,v_2$. 

If $f$ acts as a semilattice operation on $\lra$ within some crossing type $B$,  then by flipping the conditions of $B$ accordingly, we may assume that $f$ prefers  $\rightarrow$ over $\leftarrow$ on $B$. We may moreover assume that $B$ demands $u_1<v_1$   for a pair $(u,v)$ by flipping the conditions. But now we see that there exists $\beta\in\Aut(\sT)$ such that $(a_1,\beta(b_1))\in B$ and both components of  $a_{2}$ are larger that both components of $\beta(b_{2})$ with respect to $<$. Whence, $f(a,\beta(b))$ yields the desired tuple.

So assume henceforth that $f$ acts as a projection on $\lra$ within each crossing type $B$.
If $f$ behaves like the first projection  within  some  crossing type $B$ which demands $u_1<v_1$ for a pair $(u,v)$, then the same argument as above works. Otherwise, fix any such crossing type $B$; we now know that $f$ behaves like the second projection thereon. There exists $\beta\in\Aut(\sT)$  such that both components of $a_1$ are smaller than both components of $\beta(b_1)$ with respect to $<$, and such that $(a_{2},\beta(b_{2}))\in B$. Then $f(a,\beta(b))$ yields the desired tuple.
\end{proof}

\subsection{Bounded width}\label{subsect:T:bw}

\todom{We show that if $\sA$ is a first-order reduct of $\sT$ that is a model-complete core, then it has bounded width if and only if  $\CATT$ is not equationally affine, i.e., does not have a clone homomorphism to any clone of affine maps over a finite module.}

\begin{theorem}\label{thm:T-bw}
Let $\sA$ be a first-order reduct of $\sT$ \todom{that is a model-complete core}. Then precisely one of the following holds:
\begin{itemize}
    \item $\CA$ has a uniformly continuous clone homomorphism to the clone of affine maps over a finite module;
    \item $\CA$ contains for all $n\geq 3$ an $n$-ary operation that is canonical with respect to $\sT$  and a pseudo-WNU modulo $\overline{\Aut(\sT)}$.
\end{itemize}
\end{theorem}
\todom{If $\sA$ is a CSP template, i.e., has a finite signature, this gives us indeed a mathematical criterion for bounded width: in the first case $\sA$   does not have bounded width by results from~\cite{LaroseZadori,Topo-Birk}, and 
in the second case it does by~\cite{Bodirsky-Mottet}. It also proves Theorem~\ref{thm:bounded-width-classification} for the structure $\sB=\sT$: the first item (bounded width) of that theorem implies that the second item of Theorem~\ref{thm:T-bw} holds, which is identical to the third item of Theorem~\ref{thm:bounded-width-classification}. This in turn implies the second item of Theorem~\ref{thm:bounded-width-classification}  trivially. Finally, the second item of Theorem~\ref{thm:bounded-width-classification}  implies the second item of Theorem~\ref{thm:T-bw} (being incompatible with the first item), which implies the first item of Theorem~\ref{thm:bounded-width-classification}.}

\todom{
The following corollary is a description of bounded width without the assumption that $\sA$ is a model-complete core.}
\todom{
\begin{corollary}\label{cor:T:bw:main}
Let $\rel A$ be a CSP template that is a  first-order reduct of $\rel T$. If $\Pol(\rel A)$ has a uniformly continuous  minion  homomorphism to the clone of affine maps over a finite module, then $\sA$ does not have bounded width. Otherwise, $\Pol(\rel A)$   contains operations of all arities $\geq 3$ that are pseudo-WNU   modulo $\overline{\Aut(\rel T)}$, and $\sA$  has bounded width.
\end{corollary}
}
\todom{
\begin{proof}
    The first statement is an immediate consequence of general results  from~\cite{LaroseZadori,Topo-Birk}. For the second statement, let $\sA'$ be the model-complete core of $\sA$. By Lemma~\ref{lem:cores} we have that $\sA'$ is either a first-order reduct of $\sT$, or a one-element structure. If the latter is the case, then $\Pol(\sA)$ contains a constant function, and the conclusion of our statement holds. We may thus assume that $\sA'$ is a first-order reduct of $\sT$. The assumption of $\Pol(\rel A)$ not having a uniformly continuous  minion  homomorphism to the clone of affine maps over a finite module implies that the same is true for $\Pol(\rel A')$ by~\cite{wonderland}, and hence the second item of  Theorem~\ref{thm:T-bw} applies to $\sA'$. It follows that $\sA'$ and hence also $\sA$ has bounded width by~\cite{Bodirsky-Mottet}. Moreover, the identities in $\Pol(\sA')$ provided by that item lift to  $\Pol(\sA)$ (see, for example, the proof of Corollary~6.2 in~\cite{Topo}), proving the statement. 
\end{proof}
}
\todom{
The proof strategy for Theorem~\ref{thm:T-bw} is similar as for Theorem~\ref{thm:mathdichotomy-tournament}. 
As before, we may assume that $\sA$ is not a first-order reduct of $(T;=)$, since otherwise the result holds trivially. It is known that if $\CATT$ is not equationally affine,   then the second item of  Theorem~\ref{thm:T-bw}  applies, by analogous results for finite structures (\cite{MarotiMcKenzie} and~\cite[Theorem~2.8]{Maltsev-Cond}) and standard lifting techniques (see~\cite{BPP-projective-homomorphisms}) -- we achieve this and more in another way in Proposition~\ref{prop:CATT-nonaffine}.
}

\todom{
Otherwise, we  show in several steps that  does $\Pol(\sA)$ has a uniformly continuous \todom{clone}  homomorphism to the clone of affine maps over a finite module. To this end, we first establish that if $\CATT$ is equationally affine, then so is $\CAT\actson\lra$ (Proposition~\ref{prop:CATT-nonaffine}); moreover, $\CAT$ is locally interpolated by $\CA$ (Lemma~\ref{lem:bw-interpolation}), and $\CAA$ is equationally affine as well (Lemma~\ref{lem:CAA-homo-module}). We are then in position to apply the theory of smooth approximations again.
If the equivalence relation on whose classes $\CAA$ acts by functions from a clone $\mathscr M$ of affine maps over a finite module is approximated by a $\CA$-invariant equivalence relation, then $\CA$ has a uniformly continuous clone homomorphism to $\mathscr M$ by Theorem~\ref{thm:sa}.
Otherwise, Theorem~\ref{thm:2-cases-general-2} and Lemma~\ref{lem:weird-symmetric} imply that $\CA$ has a weakly \todo{commutative} binary operation.
By Lemma~\ref{lem:maj-from-symmetric}, $\CAT\actson\lra$ contains a majority operation, which contradicts the fact that $\CAT\actson\lra$ is equationally affine.
}
\subsubsection{Details of the proof}

\begin{proposition}\label{prop:CATT-nonaffine}
Suppose that $\sA$ is a first-order reduct of $\sT$ that is a model-complete core. The following are equivalent.
\begin{itemize}
\item[(1)] $\CATT$ contains for all $n\geq 3$ an $n$-ary pseudo-WNU modulo $\overline{\Aut(\sT)}$;
    \item[(2)] $\CATT$ is not equationally affine;
    \item[(3)] $\CAT$ is not equationally affine;
    \item[(4)] $\CAT\actson\lra$ is not equationally affine.
\end{itemize}
\end{proposition}
\begin{proof}
    The proof is almost identical with the one of Proposition~\ref{prop:CATT-nontrivial}, but we include it for the convenience of the reader. The implications from~(1) to~(2), from~(2) to~(3), and from~(3) to~(4) are trivial. 

We show the implication from~(4) to~(1). By Lemma~\ref{lem:generating-g2},  $\Pol(\sA)$ contains $g_2$. Moreover, by~\cite{Post}, we have that $\CAT\actson\lra$ contains a ternary majority operation; hence, it contains for all $n\geq 3$ a WNU operation. For any such operation $f$,  setting    $h_2:=g_2(x_1,x_2)$ and then for all $n\geq 3$ recursively  $h_n:=g_2(x_1,h_{n-1}(x_2,\ldots,x_n))$ we have that the function $f(h(x_1,\ldots,x_n),\ldots,h(x_n,x_1,\ldots,x_{n-1}))$  is an element of $\CATT$ which is a pseudo-WNU operation modulo $\overline{\Aut(\sT)}$, proving~(1).
\end{proof}

\begin{lemma}\label{lem:bw-inclusion}
    Let $\sA$ be a first-order reduct of $\sT$ that is a model-complete core. Suppose that $\CAT\actson\lra$ is equationally  affine.
    Then $\CAT\subseteq \CAA$.
\end{lemma}
\begin{proof}
    By the classification of automorphism groups of  first-order reducts of $\sT$~\cite{Bennett-thesis}, for all $k\geq 2$ we have that 
    the orbit of an injective $k$-tuple $a$ under $\Aut(\sA)$
    is defined by parity conditions on the number of $\rightarrow$ in the list of orbits of $(a_i,a_j)$ for $1\leq i<j\leq k$.
    Since $\CAT\actson\lra$ is equationally affine, we have by Post's classification~\cite{Post} that $\CAT\actson\lra$ consists of linear maps over $\mathbb Z_2$ (where one arbitrarily puts $\lra$ and $\{0,1\}$ in bijection).
    Thus, for any $k\geq 1$ and any linear subspace of $ (\mathbb Z_2)^k$, the operations of $\CAT$ act on every coset of this subspace.
    By our remark above, orbits of injective tuples under $\Aut(\sA)$
    correspond to intersections of such cosets,
    and therefore $\CAT$ acts on the set of orbits under $\Aut(\sA)$.
\end{proof}

\begin{lemma}\label{lem:bw-interpolation}
Let $\sA$ be a first-order reduct of $\sT$ that is a model-complete core and not a first-order reduct of $(T;=)$. 
    Suppose that $\CAT\actson\lra$ is equationally affine.
    Then $\CAT$ is locally interpolated by $\CA$.
\end{lemma}
\begin{proof}
    Every essentially unary function in $\Pol(\sA)$ which is canonical with respect to $(\sT;<)$ is also canonical with respect to $\sT$ and hence an element of $\CAT$, or else \todo{$\sA$ is a first-order reduct of $(T;=)$ by Lemma~\ref{lem:cores}}. Hence, since $\CA$ locally interpolates its  canonical functions with respect to $(\sT;<)$, we may assume that $\CA$ contains an essential operation, and  Lemma~\ref{lem:generateProjection} gives that it  contains a binary injective operation that behaves like a projection on $\lra$.
    Thus, by Lemma~\ref{lem:canonization-without-order}, either every operation of $\CA$  that is canonical with respect to $(\sT,<)$ is in $\CAT$, or there is $f\in\CATT$ that is a majority operation in its action on $\lra$.
    The latter is impossible  since $\CAT\actson\lra$ is equationally affine.
    Since every operation in $\CA$ locally interpolates an operation that is canonical with respect to $(\sT,<)$, we are done.
\end{proof}

\begin{lemma}\label{lem:CAA-homo-module}
    Let $\sA$ be a first-order reduct of $\sT$ that is a model-complete core and not a first-order reduct of $(T;=)$. 
    If  $\CAT\actson\lra$ is equationally  affine, then so is $\CAA$.
\end{lemma}
\begin{proof}
    We prove the lemma by contradiction; the proof is similar to that of Lemma~\ref{lem:canonicalSiggers}. Let $\mathscr M$ be a clone of affine maps over a finite module such that $\CAT\actson\lra$ has a clone homomorphism to $\mathscr M$. 
    Suppose that $\CAA$ has no clone homomorphism to $\mathscr M$. Let $U$ be the set of injective tuples of elements of $T$ of a fixed length $n\geq 2$ large enough so that $\Aut(\sA)$ acts with at least two orbits on $U$.  
    Then also the action of $\CAA$ on the (finite) set $U/{\Aut(\sA)}$ of $\Aut(\sA)$-orbits of injective tuples does not admit a clone homomorphism to $\mathscr M$.
    Thus, by~\cite{Maltsev-Cond},
    $\CAA$ contains operations $w_3,w_4$ of arity 3 and 4 whose actions
    on $U/{\Aut(\sA)}$ induce weak near-unanimity operations,
    satisfying moreover $w_3(x,x,y)=w_4(x,x,x,y)$ for all $x,y\in U/{\Aut(\sA)}$.
    Let $w'_3,w'_4\in\CAT$ be locally interpolated by $w_3$ and $w_4$; such  operations exist by Lemma~\ref{lem:bw-interpolation}.
    The actions of $w'_3,w'_4$ on $\lra$ are linear maps over $\mathbb Z_2$ by~\cite{Post} (by an arbitrary identification of $\lra$ and $\{0,1\}$),
    and they still satisfy the weak near-unanimity identities on  $U/{\Aut(\sA)}$ (by the same argument as in Lemma~\ref{lem:canonicalSiggers}).
    Moreover, one can assume that $w'_3$ and $w'_4$ induce idempotent maps on $\lra$.
    Therefore, we can write $w'_3(x_1,x_2,x_3)$ as $\sum_{i\in I} x_i$,
    with $I\subseteq\{1,2,3\}$ and $|I|$ odd,
    and similarly $w'_4(x_1,x_2,x_3,x_4)$ as $\sum_{j\in J} x_j$ with $|J|$ odd.
    
    Note that there is a one-to-one correspondence between the $\Aut(\sT)$-orbits of injective $n$-tuples 
    and the tuples from $\lra^{\binom{n}{2}}$ since the orbit of a tuple is determined by the orbits of its projections to two coordinates. Moreover,  the action $\CAT\actson U/\Aut(\sT)$ is isomorphic to the power $\CAT\actson\lra^{\binom{n}{2}}$ of the  action $\CAT\actson\lra$.
    Let $a,b$ be two tuples in $U$ that are not in the same orbit under $\Aut(\sA)$,
    and let $O$ and $P$ be their orbits under $\Aut(\sT)$.
    If $|I|=1$, say $I=\{1\}$, then $O=w'_3(O,P,P)$ and $P=w'_3(P,O,P)$ must be in the same orbit in $\Aut(\sA)$, a contradiction.
    We arrive at the same conclusion if $|J|=1$.
    Otherwise, $I=J=\{1,2,3\}$ without loss of generality. 
    Then $P=w'_3(P,O,O)$ and $w'_4(P,O,O,O)$ are subsets of the same orbit under $\Aut(\sA)$,
    and $w'_4(P,O,O,O)$ and $w'_4(O,O,O,P)=O$ are subsets of the same orbit
    under $\Aut(\sA)$ so that again, $a$ and $b$ are in the same orbit, a contradiction.
\end{proof}

 \section{Homogeneous Graphs}\label{sect:graphs}
Let $\sH=(H;E)$ be any  universal homogeneous $K_n$-free graph, where $n\geq 3$, or the universal homogeneous graph also known as the random graph. 

\subsection{The P/NP dichotomy}\label{subsect:H:PNP} Using exactly  the same strategy as for the random tournament, we are going to prove the following.

\begin{theorem}\label{thm:H:mathdichotomy}
	Let $\rel A$ be a first-order reduct of $\rel H$ \todom{that is a model-complete core}. 
	Then precisely one of the following holds:
	\begin{itemize}
		\item $\Pol(\rel A)$ has a uniformly continuous clone homomorphism to $\Projs$;
		\item $\Pol(\rel A)$ contains a ternary operation that is canonical with respect to $\sH$ and pseudo-cyclic modulo $\overline{\Aut(\rel H)}$.
	\end{itemize}
\end{theorem}

If $\sA$ is a CSP template, this yields a complexity dichotomy in the same fashion as for the random tournament, \todom{and Corollary~\ref{cor:T:main} holds also for first-order reducts of $\rel H$.}

The structure of our  proof of Theorem~\ref{thm:H:mathdichotomy} is precisely the same as before, and the proof overview of Section~\ref{sect:rt} applies here with $\sH$ instead of $\sT$, and $\en$ replacing $\lra$, where $N(x,y):\Leftrightarrow \neg E(x,y)\wedge x\neq y$. We call pairs in $N$ \emph{non-edges}.  The clones $\CAHH, \CAH$, and $\CAA$ are defined as for the random tournament; the explanatory equivalence relation $\Theta$ which appears there becomes $
\Theta((x,y),(u,v))\;:\Leftrightarrow\; (E(x,y)\Leftrightarrow E(u,v))\; 
$. We shall not repeat the overview of the proof, but provide now every ingredient it refers to for the case of $\sH$. Except for a few elementary calculations, most ingredients will be proven as before, and even simplify as a result of the relation $E$ being undirected;  we shall only shortly explain the differences, if any.

\subsubsection{Model-compete cores} \label{subsect:H:cores}

\begin{lemma}\label{lem:H:cores}
	Let $\rel A$ be a first-order reduct of $\rel H$. Then the model-complete core of $\rel A$ is either a one-element structure, or is again a first-order reduct of\/ $\rel H$.
\end{lemma}
\begin{proof}
The model-complete core $\rel A'$ is a first-order reduct of $\sB'$, where $\sB'$ is a homogeneous Ramsey substructure of $(\sH,<)$, and there exists $g\in\End(\rel A)$ which is  range-rigid with respect to $\Aut(\sH,<)$ and whose range induces exactly the age of $\sB'$. 

If the range of $g$ is a single point, then $\rel B'$ is a one-element structure. Otherwise, the range of $g$ contains an edge or a non-edge of $\sH$. If it only contains edges, or if it only contains non-edges, then it readily follows that $\sA$ is a reduct of $(H;=)$ as in Lemma~\ref{lem:cores}. Otherwise, $\sB'$ is isomorphic to $(\sH,<)$, and $\sA$ and $\sA'$ are isomorphic.
\end{proof}

\begin{lemma}\label{lem:H:definition-neq}
Let $\rel A$ be a first-order reduct of\/ $\sH$ that is a model-complete core.   Then the binary disequality relation $\neq$ is pp-definable in $\rel A$.
\end{lemma}
\begin{proof}
The same definition as in the proof of  Proposition~\ref{prop:definition-neq} works for $\sH$; in fact, except for the universal homogeneous $K_3$-free graph the proposition can be applied directly to $\sH$. 
\end{proof}

\subsubsection{Binary injections for $\sH$} \label{subsect:H:injective}

\begin{lemma}\label{lem:H-binary-injective}
	Let $\rel A$ be a first-order reduct of $\rel H$ that is a model-complete core. If $\Pol(\rel A)$ contains an essential operation, then it contains a binary injective operation.
\end{lemma}
\begin{proof}
	By Lemma~\ref{lem:H:definition-neq}, every polymorphism of $\rel A$ preserves $\neq$. Corollary~\ref{cor:binaryinjections} thus implies that $\Pol(\rel A)$ contains a binary injective polymorphism.
\end{proof}

\begin{lemma}\label{lem:H:proj_semilattice}
	Let $\rel A$ be a first-order reduct of $\rel H$ that is a model-complete core and not a reduct of $(H;
	=)$. Suppose that $\Pol(\rel A)$ contains a binary injection. Then $\Pol(\rel A)$ contains a binary injection which  acts like a projection or like a semilattice operation on $\{E,N\}$.
\end{lemma}
\begin{proof}
The proof is the same as the one for Lemma~\ref{lem:generateProjection}, writing $E$ instead of $\rightarrow$ and $N$ instead of $\leftarrow$, with the only exception of the case where $f$ behaves  like one semilattice operation on all inputs, which is one of the possibilities of our statement.
\end{proof}

\subsubsection{The case when $\CAHH$ is non-trivial}
The following proposition is the same as Proposition~\ref{prop:CATT-nontrivial} for $\sH$; its proof is a simplification of the proof for $\sT$.

\begin{proposition}\label{prop:CAHH-nontrivial}
Suppose that $\sA$ is a first-order reduct of\/ $\sH$ that is a model-complete core. The following are equivalent.
\begin{itemize}
\item[(1)] $\CAHH$ contains a ternary operation that is pseudo-cyclic modulo $\overline{\Aut(\sH)}$; 
    \item[(2)] $\CAHH$ is equationally non-trivial;
    \item[(3)] $\CAH$ is equationally non-trivial;
    \item[(4)] $\CAH\actson\en$ is equationally non-trivial.
\end{itemize}
\end{proposition}
\begin{proof}
The implications from~(1) to~(2), from~(2) to~(3), and from~(3) to~(4) are trivial. 

We show the implication from~(4) to~(1). By Lemma~\ref{lem:H:proj_semilattice}, $\CA$ contains a binary injection $g$ which acts like a projection or like a semilattice operation on $\en$. We may assume that $g$ is canonical with respect to $(\sH,<)$ since $(\sH,<)$ is a Ramsey structure and the above property, being invariant under composition with automorphisms of $(\sH,<)$, is not lost in the canonization process.  Since $g$ already acts on $\en$, and since for pairs $a,b\in H^2$ the $\Aut(\sH)$-orbit of $g(a,b)$ and of $g(b,a)$ does not depend on the order on $b$ if $a$ is constant by the symmetry of $E$ and $N$, we have that $g$ also acts on $\{=,E,N\}$ where $=$ denotes the $\Aut(\sH)$-orbit of constant pairs. Hence, $g$ is canonical with respect to $\sH$.

The clone of idempotent functions of $\CAH\actson\en$ satisfying non-trivial identities by~\cite{Post}, we have that $\CAH$ contains a ternary function $f$ which acts as a cyclic operation on $\en$, by~\cite{Cyclic}. Setting $h(x,y,z):=g(x,g(y,z))$, we then have that the function $$f(h(x,y,z),h(y,z,x),h(z,x,y))$$  is an element of $\CAHH$ which is pseudo-cyclic modulo $\overline{\Aut(\sH)}$, proving~(1).
\end{proof}

\begin{lemma}\label{lem:CAHinCAA}
   Let $\sA$ be a first-order reduct of\/ $\sH$ that is a model-complete core. If $\CAH$ is equationally trivial, then $\CAH\subseteq\CAA$. 
\end{lemma}
\begin{proof}
    The proof is identical with the one of Lemma~\ref{lem:CATinCAA}, with $\en$ replacing $\lra$.
\end{proof}

\begin{lemma}\label{lem:H:canonization-without-order}
    Let $\rel A$ be a first-order reduct of\/ $\sH$ that is a model-complete core and not a first-order reduct of $(H;=)$.
    Suppose that $\Pol(\rel A)$ contains 
    \begin{itemize}
        \item a binary injection which  acts like a projection on $\en$, and 
        \item a function that is canonical with respect to
    $(\sH,<)$ but not an element of $\CAH$.
    \end{itemize}
    Then $\CAH\actson\en$ contains a semilattice  operation.
\end{lemma}
\begin{proof}
Let $f\in\CA$ be a function that is  canonical with respect to $(\sH,<)$ but not an element of $\CAH$, and denote its arity by $n$.

The map $\hat f$  defined by $x\mapsto f(x,\dots,x)$ either flips $E$ and $N$ or preserves them; in the first case, we replace $f$ by $\hat f\circ f$ and may thus henceforth assume that $\hat f$ acts as the identity on $\en$.

As in Lemma~\ref{lem:canonization-without-order},  there exist $\Aut(\sH,<)$-orbits $o_1,\dots,o_n,q_1,\dots,q_n$ of pairs such that $o_i$ is  increasing and $o_i$ and $q_i$ belong to distinct  $\Aut(\sH)$-orbits for all $1\leq i\leq n$, and such that $f(o_1,\ldots,o_n)$ and $f(q_1,\ldots,q_n)$ belong to the same $\Aut(\sH)$-orbit.

Assume without loss of generality that  $f(o_1,\dots,o_n)\subseteq{} E$. We prove by induction on $m\geq 2$ that for all tuples $ a, b\in(\neq)^m$ such that for every $i\in\{1,\dots,m\}$ at most one of the pairs $a_i,b_i$ is in $N$ 
there exists $h\in\Pol(\rel A)$ such that $h(a,b)\in(E)^m$. A standard compactness argument then yields  $h\in\CA$ which acts on $\en$ as a semilattice operation;  in particular  $h\in\CAH$.

By assumption, there is an injective binary operation $g$ in $\Pol(\rel A)$ that acts as a projection on $\en$; without loss of of generality, it acts as  the first projection.
This implies that, replacing $a$ by $g(a,b)$ and $b$ by $g(b,a)$, we can assume that $\ker(a)=\ker(b)=:K$. As in Lemma~\ref{lem:canonization-without-order},  there exist linear orders $\prec,\prec'$ on $\{1/K,\ldots,(2m)/K\}$ which agree everywhere except between $1/K,2/K$, where they disagree. By flipping $\prec$ and $\prec'$, we may assume that they are increasing on $\{3/K,4/K\}$.

In the base case of the induction $m=2$, we may assume that none of $a,b$ is an element of $(E)^2$, in which case the statement would be trivial. Hence, up to reordering the tuples, we may assume that $a\in (E\times N)$ and $b\in(N\times E)$.  
For every $k\in\{1,\dots,n\}$, let $<_{o_k}$ be the increasing order on $\{3/K,4/K\}$ and $<_{q_k}$ be the order on $\{1/K,2/K\}$ given by $q_k$.  Hence,  for all $k\in\{1,\dots,n\}$, there exist $\alpha_k\in\Aut(\rel H)$ and $u^k\in\{a,b\}$ such that  $\alpha_k(u^k_1,u^k_2)$ is in $q_k\times o_k$.
Then $f(\alpha_1 u^1,\dots,\alpha_nu^n)$ is an element of $(E)^2$, 
and this concludes the base step of the induction.

For the induction step, let  $a,b\in (\neq)^m$ for some $m> 2$. By induction hypothesis, there exist binary $h, h'\in\Pol(\rel A)$ such that all coordinates of $h(a,b)$ are an element of $E$ except possibly for the second, and all coordinates of $h'(a,b)$ are an element of $E$ except possibly for the first. Define $<_{o_k}$ and $<_{q_k}$ as in the base case. By the above, there exists a linear order $<_k$ on $\{1/K,\ldots,2m/K\}$ which extends these two, and moreover we may assume that all these orders agree everywhere except on $\{1/K,2/K\}$. Thus, as in the base case, we can apply automorphisms and $f$ to these two tuples in order to obtain a tuple in $(E)^m$.
\end{proof}

\begin{corollary}\label{cor:H:canonization-without-order}
    Let $\rel A$ be a first-order reduct of $\rel H$ that is a model-complete core and not a first-order reduct of $(H;=)$.
    If $\CAH$ is equationally trivial, then $\CA$ 
    locally interpolates  $\CAH$ modulo $\Aut(\rel H)$.
\end{corollary}
\begin{proof}
Same as for Corollary~\ref{cor:canonization-without-order}.
\end{proof}

\begin{lemma}\label{lem:H:canonicalSiggers}
Let $\rel A$ be a first-order reduct of $\rel H$ that is a model-complete core and not a first-order reduct of $(H;=)$. 
If $\CAH$ is equationally trivial, then so is  $\CAA$.
\end{lemma}
\begin{proof}
Same as for Lemma~\ref{lem:canonicalSiggers}.
\end{proof}

\subsubsection{From $\CAA$ to $\CA$}\label{subsect:H:CAA2CA}

\begin{corollary}\label{cor:H:CAA_on_orbits_trivial}
Let $\sA$ be a first-order reduct of\/ $\sH$ that is a model-complete core and not a first-order reduct of $(H;=)$. If $\CAA$ is equationally trivial, then there exists $k\geq 1$ such that the action of $\CAA$ on the $\Aut(\sA)$-orbits of injective $k$-tuples is equationally trivial.
\end{corollary}
\begin{proof}
As for the random tournament, this is a consequence of the general  Proposition~\ref{prop:CAA_on_orbits_trivial}.
\end{proof}

\begin{lemma}\label{lem:H:sl-from-symmetric}
Let $\rel A$ be a first-order reduct of\/ $\sH$ which is a model-complete core and not a first-order reduct of $(H;=)$. Suppose that $\CA$ contains a binary function $f$ such that there exist $k\geq 1$, $S\subseteq H^k$ consisting of injective tuples, and $\sim$ an equivalence relation on $S$ with two $\Aut(\sA)$-invariant classes such that $f(a,b)\sim f(b,a)$ for all disjoint injective tuples $a,b\in H^k$ with  $f(a,b),f(b,a)\in S$. Then $\CAH\actson\en$  contains a semilattice  operation.
\end{lemma}

\begin{proof}
As in the proof of Lemma~\ref{lem:maj-from-symmetric}, we see that $f$ is essential and that we may assume it to be diagonally canonical with respect to $\Aut(\sH,<)$.

Let $O:=(N\cap <)$, $O':=(N\cap >)$. Then $f$ is canonical on $O$ and on $O'$ with respect to $\Aut(\rel H,<)$. We may assume it has the same property for injective tuples $a,b$ and with respect to $\Aut(\sH)$ on each of these two sets $O$ and $O'$, since otherwise we are done by Lemma~\ref{lem:H:canonization-without-order} - see the argument in the proof of  Lemma~\ref{lem:maj-from-symmetric}.

This means that on $O$ as well as on $O'$, the function $f$ acts on $\en$. 
If one of these actions is a semilattice operation, then we are done by considering $f(e_1(x),e_2(y))$, where $e_1,e_2$ are self-embeddings of $\sH$ such that $(x,y)\in O$ (or in the other case, $O'$) for all $x$ in the range of $e_1$ and all $y$ in the range of $e_2$. Hence, we assume that the two actions on $\en$ are essentially unary. If one of the two actions is not a projection, then $\Aut(\sA)$ contains a function $\alpha$ which acts on $\en$ as the non-trivial permutation: this function is obtained as $f(e_1(x),e_2(x))$ using self-embeddings $e_1,e_2$ as above. Hence, in that case the $\Aut(\sA)$-orbits are invariant under flipping  edges and non-edges on a tuple.

It follows that whenever $a,b\in S$ are such that $(a_i,b_j)\in O$ for all $i,j\in\{1,\ldots,k\}$, then $f(a,b)$ and $f(b,a)$ are elements of $S$. Furthermore, the two  actions of the restrictions of $f$ to $O$ and $O'$ on $\en$ cannot depend on the same argument. Without loss of generality, $f$ depends on the first argument on $O$, and on the second argument on $O'$.

By considering $f(e_1(x),e_2(y))$ for suitable self-embeddings  $e_1,e_2$ of $(\sH,<)$ such that $(e_1,e_2)$, acting componentwise on $T^2$, preserves $<$ and $>$, and such that $(e_1(x),e_2(y))\in O\cup O'$ for all elements $x,y$ of $\sH$, we may assume that for all $a,b\in\neq$, the relations of $\sH$ which hold between $a$ and $b$ do not influence the $\Aut(\sH)$-orbit of $f(a,b)$.

Replacing $f$ by $f(e_1\circ f(x,y),e_2\circ f(x,y))$, where $e_1,e_2$ are self-embeddings of $(\sH,<)$ which ensure that for all $a,b\in{} \neq$, the diagonal order type of $(a,b)$ is equal to that of $(e_1\circ f(a,b),e_2\circ f(a,b))$, we may assume that $f$ acts  idempotently or as a constant function on the diagonal in its action on  $\en$ within  each complete diagonal order type $B$. In particular, $f$ then acts like the first projection on $\en$ when restricted to $O$, and like the second projection when restricted to $O'$. 

Suppose that within some complete diagonal order type $B$, we have that $f$ acts as a constant function on the diagonal in its action on $\en$; that is,  $f(E,E)=f(N,N)$ on $B$. Then $\CA$ contains all permutations on $H$, in contradiction with our assumption that $\rel A$ is not a fist-order reduct of $(H;=)$. To see this, one shows   that in that case, any finite injective tuple of elements of $H$ can be mapped by a unary function in $\Pol(\sA)$ to a tuple containing whose componends induce an independent set or a complete graph in $\sH$. The proof is as in  Lemma~\ref{lem:maj-from-symmetric}; here, the assumption  $f(E,E)=f(N,N)={}N$ on $B$ leads to an independent set, and the opposite assumption to a complete graph.

Summarizing, within every fixed complete diagonal order type for two pairs in $\neq$, we have that $f$ acts idempotently on $\en$. Hence,  within each such type,
$f$ \todom{acts} as a semilattice operation or as a projection on $\en$. 

Similarly as in the proof of Lemma~\ref{lem:H:canonization-without-order}, we now show  by induction on $m\geq 2$ that for all tuples $a, b\in(\neq)^m$ such that for every $i\in\{1,\dots,m\}$ at most one of the pairs $a_i,b_i$ is in $N$ 
there exists $h\in\Pol(\rel A)$ such that $h( a, b)\in(E)^m$; or the same is true with $E$ and $N$ flipped. A standard compactness argument then implies that $\CAT$ contains a function such that $\CAT\actson\en$ is a semilattice operation.

The base case $m=2$ is a subset of the induction step, which we shall perform now. Let  $a,b\in (\neq)^{m}$ for some $m\geq 2$. Since $f$ must be essential, it generates a binary injection which acts as a projection on $\en$, and hence, as in the proof of Lemma~\ref{lem:H:canonization-without-order},  we may assume that the kernels of $a,b$ are identical. By induction hypothesis, we may moreover assume that all components of $a_1\in N$, $a_2\in E$, $b_1\in E$, $a_2\in N$, and $a_i, b_i$ are in the same $\Aut(\sH)$-orbit for all $3\leq i\leq m$. In the following, we distinguish different behaviours of $f$ which do, however, not depend on $a,b$. We will prove that we can map $a,b$ to a tuple $c\in (\neq)^m$ such that the $\Aut(\sH)$-orbits of $c_1,c_2$ are equal, and the $\Aut(\sH)$-orbit of $c_i$ is equal to that of $a_i$ (and hence of $b_i$) for all $i\geq 3$; this is clearly sufficient.

If no equalities hold between $a_1$ and $a_{2}$, then there exists $\beta\in\Aut(\sH)$ such that all components of $a_2$ are smaller than those of $\beta(b_{2})$
with respect to $<$, and  all components of $a_{1}$ are larger  than those of $\beta(b_{1})$. In that case, $f(a,\beta(b))$ yields the desired tuple.

Otherwise, assume that the second coordinate of $a_1$ equals the first coordinate of $a_2$; the other cases are treated similarly. Consider the case where $f$ acts as a semilattice operation on $\en$ within some crossing type $B$; say  it prefers  $N$ over $E$ on $B$. We may moreover assume that $B$ demands $u_1<v_1$   for a pair $(u,v)\in (\neq)^2$ by flipping the conditions. But now we see that there exists $\beta\in\Aut(\sH)$ such that $(a_1,\beta(b_1))\in B$ and both components of  $a_{2}$ are larger  than both components of $\beta(b_{2})$ with respect to $<$. Whence, $f(a,\beta(b))$ yields the desired tuple.

So assume henceforth that $f$ acts as a projection on $\en$ within each crossing type $B$. Then, by the symmetry of the relations $E$ and $N$,  we can find a crossing type on which $f$ behaves like the first projection and  which demands $u_1<v_1$ for a pair $(u,v)\in (\neq)^2$. The same argument as above works.
\end{proof}

\subsection{Bounded width}\label{subsect:H:bw}

As for the random tournament, we now characterize among the first-order reducts  of $\sH$ those 
that have bounded width as those whose polymorphism clone does not have a uniformly continuous minion homomorphism to any clone of affine maps over a finite module.

\begin{theorem}\label{thm:H-bw}
Let $\sA$ be a first-order reduct of\/ $\sH$ \todom{that is a model-complete core}. Then precisely one of the following holds:
\begin{itemize}
    \item $\CA$ has a uniformly continuous clone homomorphism to the clone of affine maps over a finite module;
    \item $\CA$ contains for all $n\geq 3$ an $n$-ary operation that is canonical with respect to $\sH$  and a pseudo-WNU modulo $\overline{\Aut(\sH)}$.
\end{itemize}
\end{theorem}
\todom{This proves the bounded width classification of  Theorem~\ref{thm:bounded-width-classification} for the structure $\sB=\sH$ as before in the case of $\mathbb T$; also, Corollary~\ref{cor:T:bw:main} holds with  $\sH$ replacing $\sT$.}

The proof strategy is the same as for the random tournament, and similar as for Theorem~\ref{thm:H:mathdichotomy}. We refer to Section~\ref{subsect:T:bw} for the overview, and now only provide the statements which are the ingredients for it, reformulated for $\sH$. We do not give any details of the proofs, since they are isomorphic to the corresponding proofs for the  random tournament; the minor  differences have already been overcome in Section~\ref{subsect:H:PNP}.

\subsubsection{Ingredients of the proof}

\begin{proposition}\label{prop:H:CATT-nonaffine}
Suppose that $\sA$ is a first-order reduct of\/ $\sH$ that is a model-complete core. The following are equivalent.
\begin{itemize}
\item[(1)] $\CAHH$ contains for all $n\geq 3$ an $n$-ary pseudo-WNU modulo $\overline{\Aut(\sH)}$;
    \item[(2)] $\CAHH$ is not equationally affine;
    \item[(3)] $\CAH$ is not equationally affine;
    \item[(4)] $\CAH\actson\en$ is not equationally affine.
\end{itemize}
\end{proposition}
\begin{proof}
    The proof is almost identical with the one of Proposition~\ref{prop:CATT-nonaffine}, with the only slight differences already explained in  Proposition~\ref{prop:CAHH-nontrivial}.
\end{proof}

\begin{lemma}\label{lem:H:bw-inclusion}
    Let $\sA$ be a first-order reduct of\/ $\sH$ that is a model-complete core. Suppose that $\CAH\actson\en$ is equationally  affine.
    Then $\CAH\subseteq \CAA$.
\end{lemma}
\begin{proof}
    As the proof of Lemma~\ref{lem:H:bw-inclusion}; we only have to refer to the classification of automorphism groups of first-order reducts of $\sH$~\cite{RandomReducts} rather than $\sT$.
\end{proof}

\begin{lemma}\label{lem:H:bw-interpolation}
Let $\sA$ be a first-order reduct of\/ $\sH$ that is a model-complete core and not a first-order reduct of $(H;=)$. 
    Suppose that $\CAH\actson\en$ is equationally affine.
    Then $\CAH$ is locally interpolated by $\CA$.
\end{lemma}
\begin{proof}
    Same as the proof of Lemma~\ref{lem:bw-interpolation}.
\end{proof}

\begin{lemma}\label{lem:H:CAA-homo-module}
    Let $\sA$ be a first-order reduct of\/ $\sH$ that is a model-complete core and not a first-order reduct of $(H;=)$. 
    If  $\CAH\actson\en$ is equationally  affine, then so is $\CAA$.
\end{lemma}
\begin{proof}
The proof is almost identical with the proof of Lemma~\ref{lem:CAA-homo-module}.
\end{proof}

 \section{Temporal CSPs}\label{sect:temp}
We now give a short proof of the complexity classification from~\cite{tcsps-journal,tcsps} for first-order expansions of $(\mathbb Q;<)$ using the theory of smooth approximations.
Using Cameron's classification of first-order reducts of $(\mathbb Q;<)$~\cite{Cameron5}, it is easy to either derive the complexity classification for arbitrary first-order reducts from the one for expansions, or to see that all non-trivial reducts have a hard CSP.
Either way, we rather choose to present the result only for expansions so as to keep the argument concise, in order to showcase the most important part of our proof.

\begin{theorem}[see \cite{tcsps-journal}]\label{thm:Q}
    Let $\sA$ be a first-order expansion of $(\mathbb Q;<)$ that is a model-complete core. Then one of the following holds.
    \begin{itemize}
        \item $\Pol(\sA,0)$ has a uniformly continuous clone homomorphism to $\P$, and $\CSP(\sA)$ is NP-complete;
        \item $\Pol(\sA,0)$ has no such clone homomorphism, and $\CSP(\sA)$ is in P.
    \end{itemize}
\end{theorem}
In the situation of the first item, NP-completeness follows by the general theory~\cite{Topo-Birk}. 
In order to show polynomial-time solvability of $\CSP(\sA)$ in the absence of a uniformly continuous clone homomorphism from $\Pol(\sA,0)$ to $\P$, we apply twice the combination of Theorem~\ref{thm:2-cases-general}, Theorem~\ref{thm:sa}, and Lemma~\ref{lem:weird-symmetric}.
Since the tractability borderline is here known not to be  described by canonical polymorphisms, one cannot use the reduction to finite-domain CSPs from~\cite{Bodirsky-Mottet} in the same fashion as for, say, the universal homogeneous  tournament.
Rather than that, we give a new description of the algorithms given in~\cite{tcsps-journal}, writing them as a polynomial-time Turing reduction to a finite-domain CSP whose template we denote by $\Af$.
Finally, we obtain that $\csp(\Af)$ is in P using the polymorphisms of $\sA$ that arise from the application of Lemma~\ref{lem:weird-symmetric}.

In this section, $\mathbb Q_{>0}$ denotes the set of positive rational numbers, and we also use $\mathbb Q_{\leq 0}$ and $\mathbb Q_{<0}$ with the obvious meaning. Abusing notation, we write $<$ and $>$ for the two $\Aut(\mathbb Q;<)$-orbits of injective pairs, $=$ for the third orbit of pairs, and $\neq{}:={}<\cup >$.  
Moreover,  $\nabla:=\{(x,y)\in\mathbb Q^2\;|\; x<y\}$ and $\Delta:=\{(x,y)\in\mathbb Q^2\;|\; x>y\}$ (so in fact, $\nabla={}<$ and $\Delta={}>$, but sometimes this notation will be clearer).  
We denote by  $pp$ any fixed  binary operation on $\mathbb Q$  such that $pp(x,y)=x$ for $x\leq 0$ and $\epsilon(y)$ for $x>0$, where $\epsilon\colon \mathbb Q\to\mathbb Q_{>0}$ is strictly increasing.
The operation $lex$ is any binary operation on $\mathbb Q$ such that $lex(x,y)<lex(x',y')$ iff $x<x'$ or $x=x',y<y'$; $ll$ is any binary operation on $\mathbb Q$ satisfying the following properties:
\begin{itemize}
    \item $ll(0,0)=0$,
    \item $ll(x,y)<ll(x',y')$ if $x\leq 0<x'$,
    \item $ll(x,y)<ll(x',y')$ if $x,x'\leq 0$ and $lex(x,y)<lex(x',y')$,
    \item $ll(x,y)<ll(x',y')$ if $x,x'>0$ and $lex(y,x)<lex(y',x')$.
\end{itemize}
The \emph{dual} of a binary operation $f$ on $\mathbb Q$ is the operation $(x,y)\mapsto -f(-x,-y)$.

\begin{proposition}\label{prop:pp-or-ll}
    Let $\sA$ be a first-order expansion of $(\mathbb Q;<)$ that is a model-complete core. 
    Suppose that $\Pol(\sA)$ does not have a uniformly continuous clone homomorphism to $\P$. 
    Then $\Pol(\sA)$ contains $pp$, $ll$, or one of their duals.
\end{proposition}
\begin{proof}
    Let $\CAQ$ be the clone of those  polymorphisms of $\sA$ that are canonical with respect to $(\mathbb Q;<)$.
    It is well-known and easy to see that $(\CAQ)^2/{\Aut(\mathbb Q;<)}$ is equationally trivial. 
    We remark that in order to prove Theorem~\ref{thm:Q}, it is not necessary to know this: if  $(\CAQ)^2/{\Aut(\mathbb Q;<)}$  were equationally non-trivial, then  $\csp(\sA)$ would be in P by the general reduction in~\cite{Bodirsky-Mottet}. Hence, we could afford to add equational triviality to the  hypotheses of our proof.
    
    By the loop lemma of smooth approximations (Theorem~\ref{thm:2-cases-general}) with $n=2$ and clones $\CAQ\subseteq \CA$,
    we are in one of two cases.
    
    In the first case, there is a naked set $(S,\sim)$ for $\CAQ$ such that $\sim$ is \todo{presmoothly} approximated by a $\Pol(\sA)$-invariant equivalence relation.
        Because $\Aut(\mathbb Q;<)$ is 2-``primitive'',
        and $\neq$ is preserved by canonical functions,
        Lemma~\ref{lem:primitive-neq-very-smooth} implies that the approximation is very smooth with respect to that group. 
        Then by the fundamental theorem (Theorem~\ref{thm:sa}), we obtain a uniformly continuous clone homomorphism from  $\Pol(\sA)$ to $\Projs$, a contradiction.

    In the second case, there is a naked set $(S,\sim)$ for $\CAQ$ such that every
        binary $\Pol(\sA)$-invariant symmetric relation $R$ that contains a pair $(a,b)\in S^2$ where $a,b$ are disjoint and such that $a\not\sim b$ contains a pseudo-loop modulo $\Aut(\mathbb Q;<)$.
        By Lemma~\ref{lem:weird-symmetric}, there is a binary $f\in\Pol(\sA)$
        such that whenever $f(a,b),f(b,a)$ are in $S$ and disjoint,
        then $f(a,b)\sim f(b,a)$.
        Without loss of generality, $f$ can be assumed to be diagonally canonical with respect to $\Aut(\mathbb Q;<)$.
        On $\nabla$ as well as on $\Delta$, $f$ acts as a canonical operation, and therefore it acts  there as a  projection or as  $lex$ or its dual.
        
        If $f$ behaves like a projection on $\nabla$ and $\Delta$, then
        by the property from Lemma~\ref{lem:weird-symmetric}, the two projections depend on different arguments.
        \todo{Supposing that $f$ behaves like the first projecion on $\nabla$, we show that it locally interpolates $pp$ (the other case is similar, where $f$ instead locally interpolates the dual of $pp$).
        Note that for all $(x,y)\in\nabla$ and $(x',y')\in\Delta$ with $x<y'$, $f(x,y)<f(x',y')$ holds.
        Indeed, for every $\epsilon>0$, we have $f(x,x+\epsilon)=f(x,y)$ and $f(y'+\epsilon,y')=f(x',y')$. Now since $x<y'$, for $\epsilon$ small enough we have $x+\epsilon<y'$, so that by preservation of $<$ we have $f(x,x+\epsilon)<f(y'+\epsilon.y)$.
        Thus, $f$ locally interpolates $pp$:
        for any finite set $S$, let $\alpha$ be an arbitrary automorphism of $(\mathbb Q;<)$ such that $\alpha(S\cap \mathbb Q_{\leq 0})<S<\alpha(S\cap\mathbb Q_{>0})$.
        Then $(x,y)\mapsto f(\alpha x,y)$ behaves like $pp$ on $S$.}
        
        Otherwise, $f$ behaves like $lex$ on $\nabla$ or on $\Delta$; without loss of generality, assume the latter. 
        Then the set $S$ must contain both orbits $<$ and $>$, and moreover, these two orbits must be  $\sim$-inequivalent: the reason for this is that $lex\in\CAQ$
        acts on $\{<,=\}$, $\{>,=\}$, and on $\{=,\neq\}$ like a semilattice operation, so the naked set $(S,\sim)$ cannot be given by any of these options. 
        Again, by the property from Lemma~\ref{lem:weird-symmetric},
        we get that the restrictions of $f$ to $\nabla$ and $\Delta$, acting on  $\{<,>\}$, depend on  different arguments.
        Then $g\colon (x,y)\mapsto lex(f(x,y),f(y,x))$ satisfies the same assumptions as $f$ but is now injective both on $\nabla$ and $\Delta$.
        By a similar argument as above, $ll$ or its dual is locally interpolated by $g$.
\end{proof}

Let $\Theta$ be the equivalence relation on $\mathbb Q_{\geq 0}$ defined as $\{(0,0)\}\cup (\mathbb Q_{>0})^2$. For $\sA$ a first-order expansion of $(\mathbb Q;<)$, we then have that $\Pol(\sA,0,\Theta)$ acts idempotently on the two $\Theta$-classes $\{0\}$ and $\mathbb Q_{>0}$; denote this action by $\Cf$.

\begin{proposition}\label{prop:poly-1-types}
    Let $\sA$ be a first-order expansion of $(\mathbb Q;<)$ that is a model-complete core. 
    Suppose that $\Pol(\sA,0)$ has no uniformly continuous clone homomorphism to $\Projs$, and that $\Pol(\sA)$ contains $pp$ or $ll$.
    Then $\Cf$ is equationally non-trivial.
\end{proposition}
\begin{proof}
    If $ll$ is in $\Pol(\sA)$, it acts as a binary commutative operation on $\{\{0\},\mathbb Q_{>0}\}$ and we are done. Therefore, let us assume that $pp$ is in $\Pol(\sA)$. 
    
    Suppose for contradiction that $\Cf$ is equationally trivial. 
    Then $\Theta$ cannot be \todo{presmoothly} approximated by $\Pol(\sA,0)$,
    as otherwise $\Pol(\sA,0)$ has a uniformly continuous clone homomorphism to $\Projs$ by Lemma~\ref{lem:primitive-neq-very-smooth} and Theorem~\ref{thm:sa}.
    
    We claim that every non-empty binary symmetric relation $R$ invariant under $\Pol(\sA,0)$ contains a pseudoloop modulo $\Aut(\mathbb Q;<,0)$.
    Consider the ternary relation
    \[ T:=\{ (\alpha(x),\alpha(y),\alpha(0)) \mid (x,y)\in R, \alpha\in\Aut(\sA)\}, \]
    which is invariant under $\Pol(\sA)$.
    
    First, suppose that $R$ intersects $\mathbb Q_{>0}\times\mathbb Q_{<0}$.
    Then $T$ contains in particular $(0,2,1)$ and $(2,0,1)$,
    so that $pp((0,2,1),(2,0,1))=(0,\epsilon(0),\epsilon(1))$ is in $T$.
    By definition of $T$, there are $\alpha\in\Aut(\sA)$ and $(x,y)\in R$
    such that $\alpha(x)=0$, $\alpha(y)=\epsilon(0)$, and $\alpha(0)=\epsilon(1)$; moreover, $\alpha(x)<\alpha(y)<\alpha(0)$ by definition of $\epsilon$.
    Let $\beta\in\Aut(\mathbb Q;<)$ be such that $\beta(\alpha(0))=0$.
    Then $\beta\circ\alpha\in\Aut(\sA;0)$, therefore it preserves $R$, so that $(\beta(\alpha(x)),\beta(\alpha(y)))\in R$.
    Note that this pair is in $(\mathbb Q_{<0})^2$ and is therefore a pseudoloop modulo $\Aut(\mathbb Q;<,0)$.
    
    Thus, we can suppose that $R$ does not intersect $\mathbb Q_{>0}\times\mathbb Q_{<0}$, and further assume that it intersects $\{0\}\times \mathbb Q_{<0}$.
    Then $T$ contains $(1,0,1)$ and $(0,1,1)$,
    so that $(\epsilon(0),0,\epsilon(1))$ is in $T$.
    Thus there are $\alpha\in\Aut(\sA)$, $(x,y)\in R$ such that $\alpha(x)=\epsilon(0),\alpha(y)=0,\alpha(0)=\epsilon(1)$.
    Then by letting $\beta\in\Aut(\mathbb Q;<)$ be such that $\beta(\alpha(0))=0$,
    we get that $(\beta(\alpha(x)),\beta(\alpha(y)))\in R$,
    and $\beta(\alpha(x)),\beta(\alpha(y))<0$.
    
    Finally, suppose that $R$ intersects $\{0\}\times \mathbb Q_{>0}$
    and none of the sets above.
    If $R$ did not have a pseudoloop, then $R\circ R$ would smoothly approximate $\Theta$, which contradicts what we established above.

    This finishes the proof of the claim.
    We can thus apply Lemma~\ref{lem:weird-symmetric} %
    with $S=\mathbb Q$ and $\sim$ being orbit-equivalence modulo $\Aut(\mathbb Q;<,0)$, and obtain a polymorphism $f'\in\Pol(\sA,0)$ such that for all $a,b$, $f'(a,b)$ and $f'(b,a)$ are in the same orbit under $\Aut(\mathbb Q;<,0)$.
    We look at the behaviour of $f'$ on $\{\mathbb Q_{>0},\{0\}\}$.
    If $f'$ preserves this set,
    then $f'$ induces a binary commutative operation in $\Cf$, which contradicts the fact that $\Cf$ is equationally trivial. 
    Therefore it must be the case that $f'(0,\mathbb Q_{>0}),f'(\mathbb Q_{>0},0)\subseteq\mathbb Q_{<0}$.
    Let $g(x,y):=f'(pp(x,y),pp(y,x))$.
    One sees that $g$ has the same behavior as $f'$ on orbits of $\Aut(\mathbb Q;<,0)$,
    but now $g(0,y)=g(0,y')$ for all $y,y'>0$ and $g(x,0)=g(x',0)$ for all $x,x'>0$.
    Let $a,b$ be these two values.
    Without loss of generality, $b\leq a$.
	
	If $b<a$, let $\alpha$ map $b$ to $0$.
	Then $(x,y)\mapsto \alpha g(\alpha g(x,y),y)$ is similar to $g$, but now satisfies $a,b>0$, i.e., it induces a binary commutative operation in $\Cf$, a contradiction.
	
	So we must have $a=b<0$.
	For $\alpha\in\Aut(\mathbb Q;<)$ mapping $a$ to $0$ we obtain that $h:=\alpha g$ satisfies $h(x,0)=h(0,x)=0$ for all $x>0$, and $h(0,0)>0$.
	Now the ternary $j(x,y,z):=h(x,h(y,z))$ induces a minority operation in $\Cf$, a contradiction.
\end{proof}

We now show how the existing polynomial-time algorithms for tractable CSPs of first-order expansions $\sA$ of $(\mathbb Q;<)$ can be seen as polynomial-time Turing reductions to a CSP on a two-element domain,  whose template $\Af$ admits the operations from $\Cf$ as polymorphisms.
Thus, if $\Pol(\sA)$ has no uniformly continous minion homomorphism to $\Projs$, then $\Cf$ is equationally non-trivial by Proposition~\ref{prop:poly-1-types}, so that $\CSP(\Af)$ is tractable by the finite-domain dichotomy theorem~\cite{BulatovFVConjecture,ZhukFVConjecture} (in fact, by Schaefer's theorem~\cite{Schaefer}). 
We then obtain that $\Pol(\sA)$ is in P, which proves the dichotomy.

The algorithms for the tractable cases are based on the notion of \emph{free sets}:
given an instance $\sX$ of  $\csp(\sA)$, a free set for $\sX$ is a \todo{non-empty} subset $S$ of its domain such that for every tuple $t=(t_1,\dots,t_k)$ in a relation $R^{\sX}$,
either $\{t_1,\dots,t_k\}$ does not intersect $S$ or there is $h\colon \{t_1,\dots,t_k\}\to\mathbb Q$ such that $h(t_1,\dots,t_k)\in R^{\sA}$ and $\argmin(h)=\{t_1,\dots,t_k\}\cap S$, i.e., the minimum of the range of $h$ is attained exactly by those elements of $\{t_1,\ldots,t_k\}$ which are elements of $S$.

Let $\Af$ be the relational structure whose two-element domain $D$ consists of the two $\Theta$-classes $\{0\}$ and $\mathbb Q_{>0}$, and which has one unary relation for each of its two elements, as well as  for every $k\geq 1$ and every $k$-ary relation $R^{\sA}$ of $\sA$ a corresponding relation
\[ R^{\Af} := \{(t_1,\dots,t_k)\in D^k \mid \exists u\in R^{\sA}\; (u\in (\mathbb Q_{\geq 0})^k) \text{ and } \forall 1\leq i \leq k\; (t_i=\{0\}\Leftrightarrow u_i=0)\}.\]
In other words, $\Af$ is obtained by first restricting $\sA$ to $\mathbb Q_{\geq 0}$, then factoring by $\Theta$, and then  adding two additional  relations naming the two $\Theta$-classes. It is clear that every element of $\Cf$ is a polymorphism of $\Af$.

\begin{lemma}\label{lem:free-sets}
    The following problems
    have a polynomial-time Turing-reduction to $\csp(\Af)$:
    \begin{itemize}
        \item the problem of, given an instance $\sX$, computing a free set for $\sX$, or deciding that none exist;
        \item the problem of, given a free set $S$ for an instance $\sX$, finding a proper subset $T$ that is a free set for $\sX$, or deciding that none exist.
    \end{itemize}
\end{lemma}
\begin{proof}
    Let $\sX$ be an instance of $\CSP(\sA)$.
    For any element $x$ of the domain of $\sX$, let $\sY_x$ be the instance of $\csp(\Af)$
    obtained by adding the constraint $x=\{0\}$ to $\sX$.
    
    Then free sets of $\sX$ containing $x$ are in bijection with homomorphisms from $\sY_x$ to $\Af$.
    Indeed, let $h\colon \sY_x\to\Af$ be a homomorphism.
    Then $h^{-1}(\{0\})$ contains $x$ by definition and it is a free set for $\sX$: indeed, let $(t_1,\dots,t_k)\in R^{\sX}$ and suppose that $\{t_1,\dots,t_k\}\cap h^{-1}(\{0\})$ is non-empty.
    The instance $\sY_x$ contains the corresponding constraint,
    so that $h(t_1,\dots,t_k)\in R^{\Af}$,
    which by definition of the relation gives that there exists
    $u\in R^{\sA}$ such that for all $1\leq i\leq k$, we have $u_i\geq 0$ and $u_i=0\Leftrightarrow h(t_i)=\{0\}\Leftrightarrow t_i\in h^{-1}(\{0\})$.
    Since $h^{-1}(\{0\})\cap\{t_1,\dots,t_k\}\neq\emptyset$, we have $u_i=0$ for at least one $i$. 
    Thus $u$ is minimal precisely for those $i$ such that $t_i\in h^{-1}(\{0\})$, i.e., $h^{-1}(\{0\})$ is a free set.
    
    Conversely, any free set $S$ of $\sX$ containing $x$ defines a homomorphism $h\colon\sY_x\to\Af$ by $h(y)=\{0\}\Leftrightarrow y\in S$.
    
    Note that the search problem for $\csp(\Af)$ Turing-reduces to the decision problem.
    Therefore, to compute a free set for $\sX$ (or decide that none exists),
    it suffices to iterate over every element $x$ of $\mathbb X$, and construct a homomorphism from  $\sY_x$ to $\Af$ if one exists.
    
    Similarly, if one constructs the instance $\sY'_{(x,y)}$ 
    as $\sX$ together with the constraints $x=\{0\},y=\mathbb Q_{>0}$, the homomorphisms from  $\sY'_{(x,y)}$ to $\Af$ correspond precisely to the free sets containing $x$ and omitting $y$.
    Thus, given a free set $S$, it suffices to iterate the following for every $x\in S$: for every $y\neq x$, decide whether there exists a homomorphism $\sY'_{(x,y)}\to\Af$.
    In such a case, we obtain a free set $T\subseteq S\setminus\{y\}$.
\end{proof}

We summarize the proof of Theorem~\ref{thm:Q}.

\begin{proof}[Proof of Theorem~\ref{thm:Q}] If the first item holds, we are done by~\cite{Topo-Birk}. Otherwise, 
    Proposition~\ref{prop:pp-or-ll} shows that $\Pol(\sA)$ contains $pp$, $ll$, or one of their duals. Suppose that $\Pol(\sA)$ contains $pp$ or $ll$, the other case being symmetric.
    By Proposition~\ref{prop:poly-1-types}, $\Pol(\sA)$ contains a polymorphism $f$ that induces an idempotent cyclic action on $\{\{0\},\mathbb Q_{>0}\}$.
    Thus, $f$ induces a cyclic polymorphism of $\Af$,
    so that $\csp(\Af)$ is in P by Schaefer's theorem~\cite{Schaefer}.
    Thus, the two computational tasks from Lemma~\ref{lem:free-sets} are in P.
    A short argument (see~\cite{tcsps-journal,ll}) shows that $\csp(\sA)$ can be solved using an oracle to the problems above, thus $\csp(\sA)$ is in P.
\end{proof}

\todo{The P/NP-complete complexity dichotomy for CSPs of first-order reducts of $(\mathbb Q;<)$ has been refined in~\cite{RydvalFP}, where the descriptive complexity of such CSPs is investigated.
We observe that the positive results about membership in fixed-point logic (with or without counting, and with or without the rank operator) in~\cite{RydvalFP} (see Proposition 3.6, Proposition 3.13, Proposition 3.23, Proposition 4.10 in~\cite{RydvalWithoutPakusa} for the most recent exposition) can also be derived from our approach by directly lifting the membership of $\csp(\Af)$ in FP/FPR.}
\begin{corollary}
    Let $\sA$ be a first-order reduct of $(\mathbb Q;<)$ such that $pp$ or $ll$ is a polymorphism of $\sA$.
    If $\csp(\Af)$ is in FP (resp. FPR), then so is $\csp(\sA)$.
\end{corollary}
\begin{proof}
    Suppose first that $\sA$ is preserved by $pp$.
    Consider the following decision problem: on an input $(\sX,x)$ where $x\in X$, determine whether $x$ belongs to a free set of $\sX$.
By Corollary 3.2 in~\cite{RydvalWithoutPakusa}, if this problem is in FP (resp.\ FPR), then so is $\csp(\sA)$.\footnote{In~\cite{RydvalWithoutPakusa}, the authors consider a more general problem also taking as input a set $U\subseteq X$ and asking whether $x$ belongs to a free set of the projection of $\sX$ onto $U$. We simply observe that this projection is first-order interpretable over $\sX$, so that it suffices to consider our problem with input $(\sX,x)$.}
By our proof of Lemma~\ref{lem:free-sets}, solving this problem amounts to computing the structure $\sY_{x}$ from $(\sX,x)$ and testing for a homomorphism to $\Af$. Recall that $\sY_x$ is the structure obtained by simply adding the constraint $x=0$ to $\sX$, in particular it has a first-order interpretation in $(\sX,x)$.
Thus, the FP (resp.\ FPR) sentence defining $\sY_x\to\Af$ can be lifted to a formula over $(\sX,x)$.

Suppose now that $\sA$ has $ll$ as a polymorphism.
Consider this time the following decision problem: on an input $(\sX,x,y)$ with $x,y\in X$, determine whether $x$ and $y$ belong to the same minimal free set of $\sX$.
By the proof of Proposition~3.23 in~\cite{RydvalWithoutPakusa}, if this problem is in FP then so is $\csp(\sA)$.\footnote{Similarly as above, the problem considered in~\cite{RydvalWithoutPakusa} also takes as input binary relations $E,T$ and construct from $\sX$ another instance obtained by contracting $T$ and projecting onto $U:=\{z\in X \mid (z,z)\not\in E\}$. This instance is FP-interpretable in $\sX$ and therefore one does not need to consider the general case.}
For any $t,z\in X$, recall that $\sY'_{(t,z)}$ is the structure obtained from $\sX$ by adding the constraints $t=0,z>0$.
The set $F_t:=\{z\in X\mid \sY_t\to\Af \text{ and } \sY'_{(t,z)}\not\to\Af\}$ defines the minimal free set containing $t$ (and is empty if $t$ is not contained in any free set).
Since $ll$ is a polymorphism of $\sA$, $\Af$ has a semilattice polymorphism, which implies that $\csp(\Af)$ is in FP.
Therefore $F_t$ is FP-definable for all $t$.
Moreover the formula $F_t\neq\emptyset\land\forall z(z\in F_t\Rightarrow t\in F_z)$ states that the minimal free set containing $t$ exists and is minimal (among arbitrary free sets).
Now the formula $\emptyset\neq F_x=F_y\land \forall z(z\in F_x\Rightarrow x\in F_z)$
states precisely that $x$ and $y$ belong to the same minimal free set,
and this formula is equivalent to an FP formula.
\end{proof}

It can be observed \emph{a posteriori} that $\csp(\sA)$ is in FP (resp.\ FPR) iff $\csp(\Af)$ is, although no direct proof is known at the moment. We see it as an interesting question whether this result can also be proved directly with our approach. %

\section{Acknowledgements}

We are grateful to Libor Barto and Simon Kn\"auer for their numerous comments on earlier versions of this paper.

\bibliographystyle{alpha}
\bibliography{global}

\newcommand{\etalchar}[1]{$^{#1}$}
\def\cprime{$'$} \def\cprime{$'$} \def\cprime{$'$}
\begin{thebibliography}{BtCLW14}

\bibitem[BB21]{UniqueInterpolation}
Manuel Bodirsky and Bertalan Bodor.
\newblock {Canonical polymorphisms of Ramsey structures and the unique
  interpolation property}.
\newblock In {\em Proceedings of LICS}, 2021.

\bibitem[BCP10]{BodChenPinsker}
Manuel Bodirsky, Hubie Chen, and Michael Pinsker.
\newblock The reducts of equality up to primitive positive interdefinability.
\newblock {\em Journal of Symbolic Logic}, 75(4):1249--1292, 2010.

\bibitem[Ben97]{Bennett-thesis}
James~H. Bennett.
\newblock {\em The reducts of some infinite homogeneous graphs and
  tournaments}.
\newblock PhD thesis, Rutgers university, 1997.

\bibitem[BG08]{BodirskyGrohe}
Manuel Bodirsky and Martin Grohe.
\newblock Non-dichotomies in constraint satisfaction complexity.
\newblock In Luca Aceto, Ivan Damgard, Leslie~Ann Goldberg, Magn\'us~M.
  Halld\'orsson, Anna Ing\'olfsd\'ottir, and Igor Walukiewicz, editors, {\em
  Proceedings of the International Colloquium on Automata, Languages and
  Programming (ICALP)}, Lecture Notes in Computer Science, pages 184 --196.
  Springer Verlag, July 2008.

\bibitem[BJ01]{BulatovJeavons}
Andrei~A. Bulatov and Peter Jeavons.
\newblock Algebraic structures in combinatorial problems.
\newblock Technical report MATH-AL-4-2001, Technische Universit\"at Dresden,
  2001.

\bibitem[BJ17]{Qualitative-Survey}
Manuel Bodirsky and Peter Jonsson.
\newblock A model-theoretic view on qualitative constraint reasoning.
\newblock {\em Journal of Artificial Intelligence Research}, 58:339--385, 2017.

\bibitem[BJMM18]{ClassificationTransfer}
Manuel Bodirsky, Peter Jonsson, Barnaby Martin, and Antoine Mottet.
\newblock Classification transfer for qualitative reasoning problems.
\newblock In {\em IJCAI-ECAI, Stockholm}, 2018.

\bibitem[BJP17]{Phylo-Complexity}
Manuel Bodirsky, Peter Jonsson, and Trung~Van Pham.
\newblock {The Complexity of Phylogeny Constraint Satisfaction Problems}.
\newblock {\em ACM Transactions on Computational Logic (TOCL)}, 18(3), 2017.
\newblock An extended abstract appeared in the conference STACS 2016.

\bibitem[BK08]{tcsps}
Manuel Bodirsky and Jan K\'ara.
\newblock The complexity of temporal constraint satisfaction problems.
\newblock In Cynthia Dwork, editor, {\em Proceedings of the Annual Symposium on
  Theory of Computing (STOC)}, pages 29--38. ACM, May 2008.

\bibitem[BK09]{tcsps-journal}
Manuel Bodirsky and Jan K\'ara.
\newblock The complexity of temporal constraint satisfaction problems.
\newblock {\em Journal of the ACM}, 57(2):1--41, 2009.
\newblock An extended abstract appeared in the Proceedings of the Symposium on
  Theory of Computing (STOC).

\bibitem[BK10]{ll}
Manuel Bodirsky and Jan K\'ara.
\newblock A fast algorithm and {D}atalog inexpressibility for temporal
  reasoning.
\newblock {\em ACM Transactions on Computational Logic}, 11(3), 2010.

\bibitem[BK12]{Cyclic}
Libor Barto and Marcin Kozik.
\newblock Absorbing subalgebras, cyclic terms and the constraint satisfaction
  problem.
\newblock {\em Logical Methods in Computer Science}, 8/1(07):1--26, 2012.

\bibitem[BK14]{BoundedWidthJournal}
Libor Barto and Marcin Kozik.
\newblock Constraint satisfaction problems solvable by local consistency
  methods.
\newblock {\em Journal of the {ACM}}, 61(1):3:1--3:19, 2014.

\bibitem[BK21]{RelationAlgebras-Knaeuer}
Manuel Bodirsky and Simon Kn\"auer.
\newblock Network satisfaction for symmetric relation algebras with a flexible
  atom.
\newblock In {\em 35th AAAI Conference on Artificial Intelligence}, 2021.

\bibitem[BKO{\etalchar{+}}17]{BKOPP}
Libor Barto, Michael Kompatscher, Miroslav Ol\v{s}\'{a}k, Trung~Van Pham, and
  Michael Pinsker.
\newblock The equivalence of two dichotomy conjectures for infinite domain
  constraint satisfaction problems.
\newblock In {\em Proceedings of the 32nd Annual {ACM/IEEE} Symposium on Logic
  in Computer Science -- LICS'17}, 2017.
\newblock Preprint arXiv:1612.07551.

\bibitem[BKO{\etalchar{+}}19]{BKOPP-equations}
Libor Barto, Michael Kompatscher, Miroslav Ol\v{s}\'{a}k, Trung~Van Pham, and
  Michael Pinsker.
\newblock Equations in oligomorphic clones and the constraint satisfaction
  problem for $\omega$-categorical structures.
\newblock {\em Journal of Mathematical Logic}, 19(2):\#1950010, 2019.

\bibitem[BM16]{Bodirsky-Mottet}
Manuel Bodirsky and Antoine Mottet.
\newblock Reducts of finitely bounded homogeneous structures, and lifting
  tractability from finite-domain constraint satisfaction.
\newblock In {\em Proceedings of the 31th Annual IEEE Symposium on Logic in
  Computer Science (LICS)}, pages 623--632, 2016.
\newblock Preprint available at ArXiv:1601.04520.

\bibitem[BM18]{BodMot-Unary}
Manuel Bodirsky and Antoine Mottet.
\newblock A dichotomy for first-order reducts of unary structures.
\newblock {\em Logical Methods in Computer Science}, 14(2), 2018.

\bibitem[BMM18]{MMSNP}
Manuel Bodirsky, Florent Madelaine, and Antoine Mottet.
\newblock {A universal-algebraic proof of the complexity dichotomy for Monotone
  Monadic SNP}.
\newblock In {\em Proceedings of the Symposium on Logic in Computer Science
  (LICS)}, pages 105--114, 2018.
\newblock Preprint arXiv:1802.03255.

\bibitem[BMM21]{MMSNP-journal}
Manuel Bodirsky, Florent Madelaine, and Antoine Mottet.
\newblock {A proof of the algebraic tractability conjecture for Monotone
  Monadic SNP}.
\newblock {\em SIAM Journal on Computing}, 2021.
\newblock To appear.

\bibitem[BMO{\etalchar{+}}19]{TopologyIsRelevant}
Manuel Bodirsky, Antoine Mottet, Miroslav Ol\v{s}\'ak, Jakub Opr\v{s}al,
  Michael Pinsker, and Ross Willard.
\newblock Topology is relevant (in the infinite-domain dichotomy conjecture for
  constraint satisfaction problems).
\newblock In {\em Proceedings of the Symposium on Logic in Computer Science --
  LICS'19}, 2019.
\newblock Preprint arXiv:1901.04237.

\bibitem[BMO{\etalchar{+}}21]{BMOOPW}
Manuel Bodirsky, Antoine Mottet, Miroslav Ol\v{s}\'ak, Jakub Opr\v{s}al,
  Michael Pinsker, and Ross Willard.
\newblock $\omega$-categorical structures avoiding height~1 identities.
\newblock {\em Transactions of the American Mathematical Society},
  374:327--350, 2021.

\bibitem[BMPP19]{BMPP16}
Manuel Bodirsky, Barnaby Martin, Michael Pinsker, and Andr{\'{a}}s
  Pongr{\'{a}}cz.
\newblock Constraint satisfaction problems for reducts of homogeneous graphs.
\newblock {\em SIAM Journal on Computing}, 48(4):1224--1264, 2019.
\newblock A conference version appeared in the Proceedings of the 43rd
  International Colloquium on Automata, Languages, and Programming, {ICALP}
  2016, pages 119:1-119:14.

\bibitem[BN06]{BodirskyNesetrilJLC}
Manuel Bodirsky and Jaroslav Ne\v{s}et\v{r}il.
\newblock Constraint satisfaction with countable homogeneous templates.
\newblock {\em Journal of Logic and Computation}, 16(3):359--373, 2006.

\bibitem[Bod07]{Cores-journal}
Manuel Bodirsky.
\newblock Cores of countably categorical structures.
\newblock {\em Logical Methods in Computer Science ({LMCS})}, 3(1):1--16, 2007.

\bibitem[Bod15]{BodirskyRamsey}
Manuel Bodirsky.
\newblock Ramsey classes: Examples and constructions.
\newblock In {\em Surveys in Combinatorics. London Mathematical Society Lecture
  Note Series 424}. Cambridge University Press, 2015.
\newblock Invited survey article for the British Combinatorial Conference;
  ArXiv:1502.05146.

\bibitem[Bod21]{Book}
Manuel Bodirsky.
\newblock {\em Complexity of Infinite-Domain Constraint Satisfaction}.
\newblock Cambridge University Press, 2021.

\bibitem[BOP18]{wonderland}
Libor Barto, Jakub Opr\v{s}al, and Michael Pinsker.
\newblock The wonderland of reflections.
\newblock {\em Israel Journal of Mathematics}, 223(1):363--398, 2018.

\bibitem[BP11]{BP-reductsRamsey}
Manuel Bodirsky and Michael Pinsker.
\newblock Reducts of {R}amsey structures.
\newblock {\em AMS Contemporary Mathematics, vol. 558 (Model Theoretic Methods
  in Finite Combinatorics)}, pages 489--519, 2011.

\bibitem[BP14]{RandomMinOps}
Manuel Bodirsky and Michael Pinsker.
\newblock Minimal functions on the random graph.
\newblock {\em Israel Journal of Mathematics}, 200(1):251--296, 2014.

\bibitem[BP15a]{BodPin-Schaefer-both}
Manuel Bodirsky and Michael Pinsker.
\newblock Schaefer's theorem for graphs.
\newblock {\em Journal of the ACM}, 62(3):52 pages (article number 19), 2015.
\newblock A conference version appeared in the Proceedings of STOC 2011, pages
  655-664.

\bibitem[BP15b]{Topo-Birk}
Manuel Bodirsky and Michael Pinsker.
\newblock Topological {B}irkhoff.
\newblock {\em Transactions of the American Mathematical Society},
  367:2527--2549, 2015.

\bibitem[BP16a]{BartoPinskerDichotomy}
Libor Barto and Michael Pinsker.
\newblock The algebraic dichotomy conjecture for infinite domain constraint
  satisfaction problems.
\newblock In {\em Proceedings of the 31th {A}nnual {IEEE} {S}ymposium on
  {L}ogic in {C}omputer {S}cience -- {LICS}'16}, pages 615--622, 2016.
\newblock Preprint arXiv:1602.04353.

\bibitem[BP16b]{BodPin-CanonicalFunctions}
Manuel Bodirsky and Michael Pinsker.
\newblock {C}anonical {F}unctions: a {P}roof via {T}opological {D}ynamics.
\newblock Preprint available under http://arxiv.org/abs/1610.09660, 2016.

\bibitem[BP20]{Topo}
Libor Barto and Michael Pinsker.
\newblock Topology is irrelevant.
\newblock {\em SIAM Journal on Computing}, 49(2):365--393, 2020.

\bibitem[BPP19]{BPP-projective-homomorphisms}
Manuel Bodirsky, Michael Pinsker, and Andr\'{a}s Pongr\'acz.
\newblock Projective clone homomorphisms.
\newblock {\em Journal of Symbolic Logic}, 2019.
\newblock To appear. doi:10.1017/jsl.2019.23.

\bibitem[BPR20]{RydvalFP}
Manuel Bodirsky, Wied Pakusa, and Jakub Rydval.
\newblock Temporal constraint satisfaction problems in fixed-point logic.
\newblock In Holger Hermanns, Lijun Zhang, Naoki Kobayashi, and Dale Miller,
  editors, {\em {LICS} '20: 35th Annual {ACM/IEEE} Symposium on Logic in
  Computer Science, Saarbr{\"{u}}cken, Germany, July 8-11, 2020}, pages
  237--251. {ACM}, 2020.

\bibitem[BPT13]{BPT-decidability-of-definability}
Manuel Bodirsky, Michael Pinsker, and Todor Tsankov.
\newblock Decidability of definability.
\newblock {\em Journal of Symbolic Logic}, 78(4):1036--1054, 2013.
\newblock A conference version appeared in the Proceedings of LICS 2011.

\bibitem[BR21]{RydvalWithoutPakusa}
Manuel Bodirsky and Jakub Rydval.
\newblock On the descriptive complexity of temporal constraint satisfaction
  problems.
\newblock Accessible at https://arxiv.org/abs/2002.09451v6, 2021.

\bibitem[BtCLW14]{OBDA}
Meghyn Bienvenu, Balder ten Cate, Carsten Lutz, and Frank Wolter.
\newblock Ontology-based data access: A study through disjunctive {D}atalog,
  {CSP}, and {MMSNP}.
\newblock {\em ACM Transactions of Database Systems}, 39(4):33, 2014.

\bibitem[Bul05]{BulatovHColoring}
Andrei~A. Bulatov.
\newblock {H}-coloring dichotomy revisited.
\newblock {\em Theoretical Computer Science}, 349(1):31--39, 2005.

\bibitem[Bul17]{BulatovFVConjecture}
Andrei~A. Bulatov.
\newblock A dichotomy theorem for nonuniform {CSP}s.
\newblock In {\em 58th {IEEE} Annual Symposium on Foundations of Computer
  Science, {FOCS} 2017, {B}erkeley, {CA}, {USA}, {O}ctober 15-17}, pages
  319--330, 2017.

\bibitem[Cam76]{Cameron5}
Peter~J. Cameron.
\newblock Transitivity of permutation groups on unordered sets.
\newblock {\em Mathematische Zeitschrift}, 148:127--139, 1976.

\bibitem[EHN19]{EvansHubickaNesetril}
David~M. Evans, Jan Hubi\v{c}ka, and Jaroslav Ne\v{s}et\v{r}il.
\newblock Automorphism groups and {Ramsey} properties of sparse graphs.
\newblock {\em Proceedings of the London Mathematical Society}, 119:515--546,
  2019.

\bibitem[FV99]{FederVardi}
Tom\'as Feder and Moshe~Y. Vardi.
\newblock The computational structure of monotone monadic {SNP} and constraint
  satisfaction: {a} study through {D}atalog and group theory.
\newblock {\em {SIAM} Journal on Computing}, 28:57--104, 1999.

\bibitem[GJK{\etalchar{+}}20]{GJKMP-conf}
Pierre Gillibert, Julius Jonu\v{s}as, Michael Kompatscher, Antoine Mottet, and
  Michael Pinsker.
\newblock Hrushovski's encoding and $\omega$-categorical {CSP} monsters.
\newblock In {\em 47th International Colloquium on Automata, Languages, and
  Programming, {ICALP} 2020, July 8-11, 2020, Saarbr{\"{u}}cken, Germany
  (Virtual Conference)}, volume 168 of {\em LIPIcs}, pages 131:1--131:17.
  Schloss Dagstuhl - Leibniz-Zentrum f{\"{u}}r Informatik, 2020.

\bibitem[GP18]{uniformbirkhoff}
Mai Gehrke and Michael Pinsker.
\newblock Uniform {B}irkhoff.
\newblock {\em Journal of Pure and Applied Algebra}, 222(5):1242--1250, 2018.

\bibitem[HN19]{Hubicka-Nesetril-All-Those}
Jan Hubi\v{c}ka and Jaroslav Ne\v{s}et\v{r}il.
\newblock All those {R}amsey classes ({R}amsey classes with closures and
  forbidden homomorphisms).
\newblock {\em Advances in Mathematics}, 356:106791, 2019.

\bibitem[KKVW15]{Maltsev-Cond}
Marcin Kozik, Andrei Krokhin, Matt Valeriote, and Ross Willard.
\newblock Characterizations of several {M}altsev conditions.
\newblock {\em Algebra universalis}, 73(3):205--224, 2015.

\bibitem[KP17]{posetCSP16}
Michael Kompatscher and Trung~Van Pham.
\newblock {A Complexity Dichotomy for Poset Constraint Satisfaction}.
\newblock In {\em 34th Symposium on Theoretical Aspects of Computer Science
  (STACS)}, volume~66 of {\em Leibniz International Proceedings in Informatics
  (LIPIcs)}, pages 47:1--47:12, 2017.

\bibitem[KP18]{posetCSP18}
Michael Kompatscher and Trung~Van Pham.
\newblock A complexity dichotomy for poset constraint satisfaction.
\newblock {\em IfCoLog Journal of Logics and their Applications ({FLAP})},
  5(8):1663--1696, 2018.

\bibitem[Lad75]{Ladner}
Richard~E. Ladner.
\newblock On the structure of polynomial time reducibility.
\newblock {\em Journal of the ACM}, 22(1):155--171, 1975.

\bibitem[LZ07]{LaroseZadori}
Benoit Larose and L{\'a}szl{\'o} Z\'adori.
\newblock Bounded width problems and algebras.
\newblock {\em Algebra Universalis}, 56(3-4):439--466, 2007.

\bibitem[MM08]{MarotiMcKenzie}
Mikl\'os Mar\'oti and Ralph McKenzie.
\newblock Existence theorems for weakly symmetric operations.
\newblock {\em Algebra Universalis}, 59(3), 2008.

\bibitem[MNPW21]{ApproxWidth}
Antoine Mottet, Tom\'as Nagy, Michael Pinsker, and Micha{\l} Wrona.
\newblock Smooth approximations and relational width collapses.
\newblock In {\em 48th International Colloquium on Automata, Languages, and
  Programming}, 2021.

\bibitem[MP21]{MottetPinskerCores}
Antoine Mottet and Michael Pinsker.
\newblock Cores over {Ramsey} structures.
\newblock {\em The Journal of Symbolic Logic}, 2021.

\bibitem[NR83]{NesetrilRoedlOrderedStructures}
Jaroslav Ne\v{s}et\v{r}il and Vojt\v{e}ch R\"odl.
\newblock Ramsey classes of set systems.
\newblock {\em Journal of Combinatorial Theory, Series A}, 34(2):183--201,
  1983.

\bibitem[Pos41]{Post}
Emil~L. Post.
\newblock The two-valued iterative systems of mathematical logic.
\newblock {\em Annals of Mathematics Studies}, 5, 1941.

\bibitem[Sch78]{Schaefer}
Thomas~J. Schaefer.
\newblock The complexity of satisfiability problems.
\newblock In {\em Proceedings of the Symposium on Theory of Computing (STOC)},
  pages 216--226, 1978.

\bibitem[Sig10]{Siggers}
Mark~H. Siggers.
\newblock A strong {M}al'cev condition for varieties omitting the unary type.
\newblock {\em Algebra Universalis}, 64(1):15--20, 2010.

\bibitem[Tho91]{RandomReducts}
Simon Thomas.
\newblock Reducts of the random graph.
\newblock {\em Journal of Symbolic Logic}, 56(1):176--181, 1991.

\bibitem[Wro20a]{Wrona:2020b}
Michal Wrona.
\newblock On the relational width of first-order expansions of finitely bounded
  homogeneous binary cores with bounded strict width.
\newblock In Holger Hermanns, Lijun Zhang, Naoki Kobayashi, and Dale Miller,
  editors, {\em {LICS} '20: 35th Annual {ACM/IEEE} Symposium on Logic in
  Computer Science, Saarbr{\"{u}}cken, Germany, July 8-11, 2020}, pages
  958--971. {ACM}, 2020.

\bibitem[Wro20b]{Wrona:2020a}
Michal Wrona.
\newblock Relational width of first-order expansions of homogeneous graphs with
  bounded strict width.
\newblock In Christophe Paul and Markus Bl{\"{a}}ser, editors, {\em 37th
  International Symposium on Theoretical Aspects of Computer Science, {STACS}
  2020, March 10-13, 2020, Montpellier, France}, volume 154 of {\em LIPIcs},
  pages 39:1--39:16. Schloss Dagstuhl - Leibniz-Zentrum f{\"{u}}r Informatik,
  2020.

\bibitem[Zhu17]{ZhukFVConjecture}
Dmitriy~N. Zhuk.
\newblock A proof of {CSP} dichotomy conjecture.
\newblock In {\em 58th {IEEE} Annual Symposium on Foundations of Computer
  Science, {FOCS} 2017, {B}erkeley, {CA}, {USA}, {O}ctober 15-17}, pages
  331--342, 2017.
\newblock https://arxiv.org/abs/1704.01914.

\bibitem[Zhu20]{Zhuk:2020}
Dmitriy Zhuk.
\newblock {A Proof of the CSP Dichotomy Conjecture}.
\newblock {\em Journal of the ACM}, 67(5), August 2020.

\end{thebibliography}
\end{document}